\documentclass[12pt]{article}

\usepackage{amsmath,amssymb,amsfonts,amsthm}
\usepackage{tikz}
\usetikzlibrary{arrows,positioning,calc,arrows.meta,shapes.geometric}
\usepackage{xcolor}
\usepackage{natbib}
\usepackage{geometry}
\usepackage{setspace}
\usepackage{hyperref}
\usepackage{graphicx}
\usepackage{booktabs}
\usepackage{placeins}
\usepackage{xparse}
\usepackage{adjustbox}
\usepackage{rotating}
\usepackage{enumitem}
\usepackage{framed}

\newcommand{\ind}{\mathbf{1}}

\DeclareMathOperator*{\argmax}{arg\,max}

\newcommand{\inputIfExists}[1]{%
  \IfFileExists{#1}{\input{#1}}{%
    \begin{tabular}{lccc}\toprule
    Output & Cluster & Factor & Sparse\\
    \midrule
    Placeholder & -- & -- & --\\
    \bottomrule
    \end{tabular}%
  }%
}
\NewDocumentCommand{\includegraphicsIfExists}{O{}m}{%
  \IfFileExists{#2}{\includegraphics[#1]{#2}}{%
    \fbox{\begin{minipage}{0.75\textwidth}\centering
    Figure file \texttt{\detokenize{#2}} not found. Run the replication package to generate this file.
    \end{minipage}}%
  }%
}

\IfFileExists{results/tables/macros_lambda_bar.tex}{%
\newcommand{\lambdaBarRealized}{0.1187}
\newcommand{\lambdaBarZscore}{1.877}
}{%
  \newcommand{\lambdaBarRealized}{[run replication package]}%
  \newcommand{\lambdaBarZscore}{[run replication package]}%
}

\newtheorem{assumption}{Assumption}
\newtheorem{definition}{Definition}
\newtheorem{lemma}{Lemma}
\newtheorem{theorem}{Theorem}

\newtheorem{proposition}{Proposition}
\newtheorem{remark}{Remark}

\title{\bf Learning Dependence Structures for Econometric Inference:
Identification, Ambiguity, and Adaptive Inference}

\author{
Ulrich Hounyo\thanks{%
Department of Economics, University at Albany -- State University of New
York, Albany, NY 12222, USA. E-mail: khounyo@albany.edu. }\\
Department of Economics\\
University at Albany, SUNY
}

\date{}

\begin{document}

\maketitle

\begin{abstract}
\begingroup
\setstretch{1.0}
Econometric inference usually conditions on a dependence structure chosen in advance, even though the data may support clustering, latent factors, sparse interactions, or mixtures of these mechanisms. This paper studies the prior problem of learning the dependence structure that is relevant for inference. We represent candidate structures as covariance geometries in a common Hilbert space and project an estimable dependence operator onto them. The resulting geometric dependence profile is a low-dimensional diagnostic of their relative empirical support; an off-diagonal companion profile isolates cross-sectional dependence and drives procedure selection. We establish well-definedness, consistency, asymptotic normality, and finite-sample classification bounds under local projection regularity and geometric separation, and show that tangent-space overlap creates a first-order impossibility region in which competing geometries cannot be reliably distinguished. Formulating inference-procedure choice as a statistical decision problem, we prove that when one off-diagonal geometry is uniquely separated and profile rankings are compatible with inferential loss, profile-guided inference is asymptotically equivalent to an infeasible oracle and has vanishing regret. The framework thus links dependence diagnostics, learnability, ambiguity, and adaptive inference in a single data-to-decision procedure.


\smallskip

\endgroup
\end{abstract}


\section{Introduction}
\label{sec:introduction}

Applied researchers routinely face a consequential choice before conducting
inference: which dependence structure should govern the construction of
standard errors and critical values? Cluster-robust procedures require a
clustering partition \citep{LiangZeger1986,Arellano1987,CameronGelbachMiller2011,
CameronMiller2015}; factor-based corrections require a specification of the
latent factor structure \citep{Bai2003,ChamberlainRothschild1983}; and HAC,
spatial, and network procedures require a bandwidth, kernel, distance metric,
or network architecture \citep{White1980,NeweyWest1987,Andrews1991,Conley1999,
DriscollKraay1998,Auerbach2019,Leung2022}. These methods answer how to conduct
inference conditional on a dependence model. They do not answer the prior
question that arises when the relevant dependence model is itself unknown.

That prior question is empirically important. Firms may share industry shocks,
latent macroeconomic factors, and localized network linkages at the same time.
Regional outcomes may reflect institutional clustering together with spatial
spillovers. Financial returns may contain pervasive common components and
sparse pairwise interactions. Choosing one robust procedure by assumption can
therefore be difficult to justify and can materially affect empirical
conclusions. The statistical problem studied in this paper is consequently not
another variance-estimation problem under a maintained dependence structure.
It is the problem of learning the dependence structure that is most relevant
for inference and translating that evidence into a defensible inferential
decision.

We develop a unified framework for this problem. We observe a panel
$\{u_t\}_{t=1}^T$ of $N$-dimensional cross-sectional vectors, or residuals
thereof, with fixed $N$ and $T\to\infty$, so the covariance operator
$\Gamma_0=E(u_tu_t')$ is estimable. Rather than assuming $\Gamma_0$ belongs to
one covariance class, we compare it with a researcher-specified dictionary of
economically meaningful geometries---in the benchmark, cluster, factor, and
sparse covariance cones. Metric projections produce similarity scores whose
normalization forms a low-dimensional \emph{geometric dependence profile}: the
full profile records overall covariance fit, and an off-diagonal companion
isolates the cross-sectional dependence relevant for choosing among robust
procedures. The profile is a diagnostic, not a claim that the data reveal a
unique structural mechanism---it summarizes the relative support the operator
assigns to the dictionary, while projection residuals reveal when the dictionary
is inadequate and geometric overlap reveals when the data cannot discriminate
among its elements.

The resulting data-to-decision procedure has four linked components: the
operator is estimated and projected onto the candidate geometries; the profile
and residuals diagnose dominance, ambiguity, and dictionary misspecification;
geometric separation determines whether the dominant geometry is learnable; and
the learned geometry is mapped into an inference procedure through a decision
rule---recommending the matched procedure when one geometry dominates, reporting
ambiguity and supporting hybrid procedures near a tie, and warning against
geometry-specific inference when every projection fits poorly. Dependence is
therefore treated as an estimand, and adaptive inference becomes the prescriptive
consequence of learning it.

The central theoretical question is whether the profile can reliably learn a
dominant geometry from an estimated dependence operator. The answer depends on
local projection regularity and on the geometric separation of the candidate
covariance classes. Principal angles between their off-diagonal tangent
directions measure separation between local approximations to the candidate
geometries. Positive angles provide pairwise local separation, whereas shared
tangent directions generate first-order indistinguishability. The same
geometric construction therefore characterizes both the circumstances under
which dependence learning is informative and the circumstances under which
reliable discrimination is impossible.

This yields a fundamental limit of dependence learning: near regions of
geometric overlap, no statistical procedure can distinguish competing
dependence structures at first order. Ambiguous dependence profiles
therefore reflect intrinsic nonidentification rather than finite-sample
uncertainty.

Beyond relative comparisons among candidate geometries, the framework
provides diagnostics for assessing the adequacy of the covariance
dictionary itself. Projection-residual diagnostics measure the distance
between the empirical dependence operator and its closest geometric
approximation, distinguishing three practically important situations: a
clearly dominant geometry, intrinsic ambiguity, and a misspecified
dictionary.

The inferential payoff is naturally decision theoretic \citep{Wald1950}.
Different robust procedures have different losses across dependence
geometries, so the empirical problem is procedure selection under uncertainty
about the dependence state. Existing methods typically resolve that
uncertainty by assumption before inference begins. We instead let the estimated
off-diagonal profile generate a data-driven decision rule. When its dominant
geometry is uniquely separated and profile rankings are compatible with
inferential loss, the selected procedure is asymptotically equivalent to an
infeasible oracle that knows the relevant geometry in advance and its regret
vanishes asymptotically. Thus the descriptive problem of dependence learning
and the prescriptive problem of inference selection are two stages of one
statistical decision problem.

The paper contributes to the literature in five ways.

\medskip

\noindent
\textbf{First}, we introduce a geometric framework for dependence
learning. We observe a panel and estimate the cross-sectional covariance
matrix $\Gamma_0$, which is a fixed $N\times N$ object consistently
estimable as $T\to\infty$. Covariance geometries are defined as closed covariance cones contained in the positive-semidefinite cone of the Hilbert space $(\mathbb{R}^{N\times N},\langle\cdot,
\cdot\rangle_F)$, and the dependence profile is obtained by projecting
the estimated $\hat\Gamma_T$ onto each candidate geometry. This
provides a unified representation of cluster, factor, sparse, and more
general dependence structures within a common mathematical framework.

\medskip

\noindent
\textbf{Second}, we establish an identification theory for dependence
learning. Under local projection regularity, dependence profiles are well defined; a principal-angle separation condition establishes pairwise separation of the candidate geometries. Under additional sampling and projection-%
differentiability conditions, the full and off-diagonal profiles are consistently estimable and asymptotically normal. We further derive finite-sample classification
error bounds that quantify the reliability of dependence learning.

\medskip

\noindent
\textbf{Third}, we characterize the fundamental limits of dependence
learning. We show that when covariance geometries share common tangent
directions, no statistical procedure can distinguish them at first
order. This impossibility theorem demonstrates that ambiguous
dependence profiles may reflect intrinsic nonidentification rather than
finite-sample uncertainty; projection-residual diagnostics complement
these results by testing whether the candidate dictionary itself is
adequate.

\medskip

\noindent
\textbf{Fourth}, we develop a framework for adaptive econometric
inference under unknown dependence. Estimated off-diagonal dependence profiles guide
the selection of dependence-robust inference procedures. Under a unique
profile margin and an explicit profile--loss compatibility condition, the
selected procedure is asymptotically equivalent to an infeasible oracle
possessing prior knowledge of the dominant dependence structure.

\medskip

\noindent
\textbf{Fifth and finally}, we provide a decision-theoretic interpretation. We
formulate profile-guided inference as a procedure-selection problem
under uncertainty and show that the resulting procedure achieves
asymptotically vanishing regret. Under the same unique-margin and profile--loss compatibility conditions,
learning the dependence structure incurs no first-order asymptotic loss.

\paragraph{Related literature.}
The paper connects four literatures that have largely developed separately.
\emph{Dependence-robust inference} provides procedure-specific solutions
conditional on a known structure---heteroskedasticity- and cluster-robust
variance estimation \citep{White1980,LiangZeger1986,Arellano1987,
CameronGelbachMiller2011,CameronMiller2015}, HAC and spatial corrections
\citep{NeweyWest1987,Andrews1991,Conley1999,Conley2008,DriscollKraay1998}, and
network-robust inference \citep{Auerbach2019,Leung2022}---whereas we take the
\emph{choice} among such procedures as the object of study. \emph{Structured
covariance estimation} studies each geometry in isolation: approximate factor
structure \citep{ChamberlainRothschild1983,BaiNg2002,Bai2003}, sparse covariance
and precision estimation \citep{BickelLevina2008,CaiLiu2011,
FriedmanHastieTibshirani2008}, and low-rank-plus-sparse decompositions
\citep{CandesLiMaWright2011,Chandrasekaran2011}; we instead treat these as
competing geometries in a common Hilbert space and characterize when the data
cannot distinguish them. The \emph{model-selection and averaging} literatures
\citep{BatesGranger1969,Hansen2007,ClaeskensHjort2008} quantify relative fit; the
geometric formulation adds a notion of distance between candidate structures
themselves, enabling the local-separation and impossibility results and
connecting to adaptation theory \citep{Lepski1991}. Finally, because the
procedure selects an inference method using the same data, the oracle and regret
analysis relates to \emph{post-selection inference} \citep{LeebPoetscher2005,
BerkBrownBujaZhangZhao2013,TibshiraniEtAl2016}, delineating when selection
consistency is asymptotically innocuous and when---near overlap---no procedure
escapes the ambiguity. It also relates to the robust-statistics tradition
\citep{Huber1981,HampelEtAl1986}, where robustness concerns distributional
departures within a given dependence model rather than the dependence structure
itself; the asymptotic theory builds on standard semiparametric and
empirical-process tools \citep{BickelKlaassenRitovWellner1993,Newey1994,
VanDerVaartWellner1996,LeCamYang2000}.

The geometric formulation is not adopted merely for mathematical
convenience; it is essential relative to direct model selection based on
cross-validation, information criteria, or prediction error. It represents
heterogeneous dependence structures---clustering, latent factors, sparse
networks---within a single parameter space, so structures of very different
functional form become directly comparable; it supplies a unified notion of
distance and ambiguity, since the same Frobenius geometry that measures fit also
measures how close two candidate structures are; and it yields projection-based
identification conditions, local impossibility results, and adaptive
inference---none of which have counterparts in generic model-selection criteria,
which quantify relative fit but are silent on when candidate structures are
intrinsically indistinguishable.

\emph{Our contribution.} Existing work assumes, estimates, selects
among, or averages over covariance structures. The present paper does
not propose another covariance estimator, robust variance formula, or
selection criterion. It proposes a statistical theory for learning which
covariance geometry should govern subsequent inference: an identification
theory for when the geometry is learnable, an impossibility theory for when it
is not, and an adaptive inference theory establishing that, whenever the
geometry is learnable, learning it costs nothing asymptotically relative to
knowing it in advance. In this sense the paper shifts the focus from
dependence-robust inference under a specified covariance model to adaptive
inference under unknown dependence.

\paragraph{Theoretical roadmap and central result.}
The theory is deliberately sequential. Theorem~\ref{thm:identification} shows the
profile is well defined at regular operators and that positive principal angles
separate geometries locally; Theorems~\ref{thm:profile_consistency}
and~\ref{thm:profile_clt} transfer covariance estimation into consistent,
asymptotically normal estimation of the full and off-diagonal profiles;
Theorems~\ref{thm:dominant_geometry_consistency} and~\ref{thm:local_ties}
characterize the positive and boundary classification regimes; and
Theorem~\ref{thm:fundamental_limit} shows tangent-space overlap creates a
first-order impossibility region. The analysis culminates in
Theorem~\ref{thm:oracle_adaptivity_asymptotic_optimality}, the paper's central
data-to-decision result: whenever the dominant off-diagonal geometry is learnable
and profile rankings are compatible with inferential loss, the profile-guided
procedure is first-order equivalent to the infeasible oracle and has vanishing
regret. Each preceding result supplies a necessary link in the chain from
covariance estimation to adaptive inference.

\paragraph{Notation guide.}
Objects are introduced in pipeline order. The population operator is
$\Gamma_0\in\mathbb H_+$, estimated by $\widehat\Gamma_T$; candidate geometries
$\mathcal S_d$, $d\in\mathfrak D$, have metric projections $P_d(\Gamma)$ and
similarity scores $S_d(\Gamma)=\|P_d(\Gamma)\|_F^2$; $\omega(\Gamma)$ is the
normalized full profile and $\omega^{\mathrm{off}}(\Gamma)$ its off-diagonal
companion. The dominant population and sample geometries are $d^\star$ and
$\widehat d$, with separation margins $\Delta_{\omega^{\mathrm{off}}}$ and
$\widehat\Delta_{\omega^{\mathrm{off}}}$; $\rho_d,\widehat\rho_d$ are projection
residuals, and $\kappa_0,\widehat\kappa$ combine separation and dictionary fit.
Finally $a_d$ is the action matched to geometry $d$, $\delta$ the data-dependent
decision rule, and $\mathcal R(\delta,\Gamma_0)$ its regret. The same symbols
carry unchanged meanings in the paper, appendix, and replication code.

The paper is organized as follows. Section~\ref{sec:geometry}
introduces the statistical experiment, defines the population
dependence operator, establishes its estimability, and develops the
covariance geometry framework. Section~\ref{sec:identification}
develops geometric separation and ambiguity results.
Sections~\ref{sec:estimation} and~\ref{sec:asymptotics} study
estimation and asymptotic theory. Section~\ref{sec:classification}
develops dependence diagnostics and classification.
Section~\ref{sec:profile_guided_inference} presents profile-guided
inference. Sections~\ref{sec:simulation} and~\ref{sec:empirical}
report simulation and empirical evidence. Technical proofs and
additional results are in the Appendix. The data-to-decision pipeline is
summarized schematically in Figure~\ref{fig:workflow} of
Section~\ref{sec:profile_guided_inference}.


\section{Statistical Experiment, Estimability, and Covariance Geometry}
\label{sec:geometry}

\subsection{Statistical Experiment}
\label{subsec:stat_experiment}

Let $N\geq 1$ denote the cross-sectional dimension, treated as fixed
throughout. We observe a panel
\[
\{u_t\}_{t=1}^T,
\qquad u_t=(u_{1t},\ldots,u_{Nt})'\in\mathbb R^N,
\]
where $u_t$ is a zero-mean random vector with covariance matrix
\[
\Sigma = E(u_t u_t') \in \mathbb{R}^{N\times N}.
\]
Asymptotics are driven by the time dimension: $T\to\infty$ with $N$
fixed.

\begin{remark}[Asymptotic regime]
\label{rem:asymptotic_regime}
Throughout, $N$ is fixed and $T\to\infty$. The cross-sectional dimension $N$
governs the \emph{geometry} of dependence---the dimension of the Hilbert space,
the projection operators, and the principal-angle conditions---while $T$ governs
\emph{estimation}, the rate at which the empirical operator concentrates around
its population counterpart. This separation between the geometric and asymptotic
dimensions is what makes the framework well-posed.
\end{remark}

\subsection{Population Dependence Operator}
\label{subsec:pop_dep_operator}

We equip the space of symmetric $N\times N$ matrices with the Frobenius
inner product
\[
\langle A, B\rangle_F = \operatorname{tr}(A'B),
\qquad
\|A\|_F = \langle A,A\rangle_F^{1/2},
\]
giving the Hilbert space
\[
\mathbb{H} = \bigl\{A\in\mathbb{R}^{N\times N}: A=A',\;\|A\|_F<\infty\bigr\}.
\]

\begin{definition}[Population Dependence Operator]
\label{def:pop_dep_op}
The \emph{population dependence operator} is the element
$\Gamma_0\in\mathbb{H}$ represented by the covariance matrix
\[
\Gamma_0 \;=\; \Sigma \;=\; E(u_t u_t').
\]
\end{definition}

\noindent
The Hilbert-space formulation in Definition~\ref{def:pop_dep_op} is not
merely notational. Representing $\Sigma$ as an element of
$(\mathbb{H},\langle\cdot,\cdot\rangle_F)$ enables the use of
projection operators, tangent spaces, and principal angles---the central
geometric objects of the paper. In finite dimensions, $\Gamma_0$ is
identified with $\Sigma$; in extensions to functional or network data,
the same abstract framework accommodates more general covariance kernels
or operators without changing the theoretical results.

\subsection{Empirical Dependence Operator and Estimability}
\label{subsec:estimability}

We observe $T$ cross-sectional vectors $u_1,\ldots,u_T$. In practice $u_t$ is
typically not observed directly but recovered as residuals from a regression
model; the leading case, used in the simulation and empirical sections, is the
panel regression
\begin{equation}
y_{it} = x_{it}'\beta + u_{it},
\qquad i=1,\ldots,N,\quad t=1,\ldots,T,
\label{eq:panel_regression}
\end{equation}
with $u_t=(u_{1t},\ldots,u_{Nt})'$ and $\widehat\Gamma_T$ built from the OLS
residuals $\widehat u_t$. The estimand is $\Gamma_0=\Sigma=E(u_tu_t')$, a fixed
$N\times N$ positive semidefinite matrix: since $N$ is fixed, this is not the
problem of estimating a covariance matrix that grows with the sample, but of
estimating a fixed matrix using an increasing number of time-series
observations. The natural estimator is the sample covariance of the (estimated)
residuals,
\begin{equation}
\widehat\Gamma_T
=
\frac{1}{T}
\sum_{t=1}^{T}
\widehat u_t\widehat u_t',
\label{eq:sample_cov}
\end{equation}
$\widehat u_t = (\widehat u_{1t},\ldots,\widehat u_{Nt})'$. We state the formal
conditions and consistency as a proposition, giving the estimator an explicit
foundation before the geometric theory.

\begin{assumption}[Stationarity, Mixing, and Moments]
\label{ass:mixing}
The joint process $\{(u_t,X_t)\}$ is strictly stationary and ergodic,
where $X_t\in\mathbb{R}^{N\times p}$ collects the regressors at time
$t$. There exists $\delta>0$ such that $E\|u_t\|^{4+\delta}<\infty$
and $E\|X_t\|_F^{2+\delta}<\infty$. The mixing condition is imposed
directly on the \emph{quadratic} process
$\psi_t=\operatorname{vec}(u_tu_t'-\Sigma)$, which is the process to
which the central limit theorem is applied: writing $\eta=\delta/2$, so
that $E\|\psi_t\|^{2+\eta}<\infty$ whenever
$E\|u_t\|^{4+\delta}<\infty$, the process satisfies strong mixing with
\[
\sum_{k=0}^\infty \alpha(k)^{\eta/(2+\eta)}
=
\sum_{k=0}^\infty \alpha(k)^{\delta/(4+\delta)}
<\infty ,
\]
where $\alpha(k)$ is the strong mixing coefficient of
$\{(u_t,X_t)\}$ at lag $k$ (mixing coefficients of $\{\psi_t\}$ are
bounded by those of $\{u_t\}$, since $\psi_t$ is a measurable function of
$u_t$). This is the summability exponent required by the mixing central
limit theorem applied to $\{\psi_t\}$, and it is the binding condition;
the more familiar exponent $\delta/(2+\delta)$ would be appropriate for
$\{u_t\}$ itself rather than for the quadratic process.
The regressors are strictly exogenous, the OLS estimator satisfies
$\hat\beta-\beta=O_p(T^{-1/2})$, and the regressor--error cross
moment obeys
\[
\frac1T\sum_{t=1}^T X_t\otimes u_t=O_p(T^{-1/2}).
\]
The last rate is the condition used to make residual replacement
first-order negligible; it follows, for example, from an appropriate
mixing central limit theorem for the cross-product process. In particular,
$T^{-1}\sum_{t=1}^T\|X_t\|_F^2=O_p(1)$ by the ergodic theorem.
\end{assumption}

\begin{proposition}[Consistency of the Empirical Dependence Operator]
\label{prop:consistency_operator}
Under Assumption~\ref{ass:mixing},
\[
\|\widehat\Gamma_T - \Gamma_0\|_F = O_p(T^{-1/2}).
\]
\end{proposition}

\begin{remark}[Generic residual constructions]
\label{rem:generic_residuals}
Assumption~\ref{ass:mixing} describes a fixed-dimensional, common-$\beta$ OLS
problem. Empirical designs often use richer residual constructions---most
commonly a two-way within transformation---not covered directly, so it is useful
to state the requirement generically. Let $\widehat u_t$ be any residual
construction and $u_t$ the corresponding population disturbance, and suppose
\[
\frac1T\sum_{t=1}^T(\widehat u_t-u_t)(\widehat u_t-u_t)'=O_p(T^{-1}),
\qquad
\frac1T\sum_{t=1}^T(\widehat u_t-u_t)u_t'=O_p(T^{-1}),
\tag{R}
\label{eq:generic_residual}
\]
together with the stationarity, moment, and mixing conditions of
Assumption~\ref{ass:mixing} imposed on $\{u_t\}$. Then the conclusion of
Proposition~\ref{prop:consistency_operator} holds verbatim and, by
Remark~\ref{rem:primitive_clt}, the influence function is unchanged.
Condition~\eqref{eq:generic_residual} is verified for common-$\beta$ OLS under
strict exogeneity in the proof of Proposition~\ref{prop:consistency_operator}.
For a two-way fixed-effects specification, the relevant transformed disturbance
and its covariance operator must be defined from the complete within
transformation, and \eqref{eq:generic_residual} verified for that construction;
with fixed $N$ this is standard under suitable moment, exogeneity, and
weak-dependence conditions, but it is design-specific and not an automatic
consequence of Proposition~\ref{prop:consistency_operator}. Researchers using
other constructions should verify~\eqref{eq:generic_residual} directly.
\end{remark}

\noindent
Proposition~\ref{prop:consistency_operator} confirms that $\Gamma_0$ is a
well-defined, consistently estimable estimand at the $\sqrt T$-rate that drives
all the CLTs of Section~\ref{sec:asymptotics}. The proof decomposes
$\widehat\Gamma_T-\Gamma_0$ into the sampling error of the infeasible covariance
$T^{-1}\sum_t u_tu_t'$ (which is $O_p(T^{-1/2})$ under Assumption~\ref{ass:mixing}
by standard results for weakly dependent processes, e.g.\
\citealt{Davidson1994}) and a residual-replacement remainder that is
$O_p(T^{-1})$ under strict exogeneity. Estimating $\Gamma_0$ is not the goal in
itself: $\widehat\Gamma_T$ is an intermediate object from which the estimand of
interest, the population profile $\omega_0$, is constructed.

\begin{remark}[Two notions of dependence]
\label{rem:two_dependence_notions}
The weak-dependence conditions on $\{u_t\}$ in
Proposition~\ref{prop:consistency_operator} concern the \emph{temporal} sequence
across $t$ and serve only to guarantee consistent estimation of $\Gamma_0$; they
place \emph{no} restriction on the \emph{cross-sectional} dependence structure
encoded in $\Gamma_0=E(u_tu_t')$, which is precisely the object being learned.
Temporal and cross-sectional dependence are therefore distinct, and assumptions
on the former do not restrict the latter.
\end{remark}

\subsection{Covariance Classes}
\label{subsec:cov_classes}

Rather than assuming $\Gamma_0$ belongs to a single class, we characterize its
dependence architecture through geometric proximity to several economically
meaningful covariance classes. Each geometry, defined next, is a
\emph{deterministic} subset of $\mathbb H$, so the projection of $\Gamma_0$ onto
it is well defined. Let $\mathcal S_C,\mathcal S_F,\mathcal S_S$ denote the
cluster, factor, and sparse classes.

\begin{definition}[Covariance-cone dictionary]
\label{def:covariance_cone_dictionary}
Write $\mathbb H_+=\{\Gamma\in\mathbb H:\Gamma\succeq0\}$.  Every
candidate geometry used in this paper is a nonempty closed cone contained
in $\mathbb H_+$.  Consequently every projected point is itself a valid
covariance matrix.  This restriction is substantive: the dictionary
classifies covariance structure, rather than merely symmetric support
patterns.
\end{definition}

For a prespecified cluster-support matrix $M_C$ with unit diagonal, define
\[
\mathcal S_C
=
\{\Gamma\in\mathbb H_+:\Gamma_{ij}=0\text{ whenever }(M_C)_{ij}=0\}.
\]
Thus the cluster class is the intersection of the covariance cone with a
fixed linear support space.  It is a closed convex cone.  For one-way
partitions the support mask is block diagonal and masking preserves
positive semidefiniteness.  For overlapping multiway supports the metric
projection is the solution of a convex PSD-constrained least-squares
problem; entrywise masking alone is not, in general, the metric
projection.

The support matrix is constructed from the clustering dimensions by
\[
(M_C)_{ij}
=
\ind\Bigl\{g_m(i)=g_m(j)\text{ for at least one }m\in\{1,\ldots,M\}\Bigr\}.
\]
This includes one-way, two-way, and multiway cluster dictionaries.

\begin{remark}[Cluster geometry versus many-cluster asymptotics]
\label{rem:cluster_vs_manycluster}
The paper's objective is not the asymptotic validity of cluster-robust variance
estimators under many-cluster asymptotics. The cluster covariance class here is a
finite-dimensional geometry on a fixed $N\times N$ matrix describing which entries
are compatible with a given one-way, two-way, or multiway support---distinct from
the many-cluster theory used to justify conventional cluster-robust standard
errors. The present fixed-$N$, $T\to\infty$ framework studies learning the
cross-sectional geometry; validity of a particular geometry-specific inference
action is imposed separately when that action is used.
\end{remark}

For a fixed rank $r$, define the factor class as
\[
\mathcal S_F(r)
=
\Big\{
\Gamma\succeq0:
\Gamma=L+D,\;
L\succeq0,\;
\operatorname{rank}(L)\le r,\;
D\succeq0 \text{ diagonal}
\Big\}.
\]

For a fixed sparsity budget $k_S\le N(N-1)/2$, define the sparse
covariance cone as
\[
\mathcal S_S
=
\Big\{
\Gamma\in\mathbb H_+:
|\operatorname{supp}_{\mathrm{off}}(\Gamma)|\le k_S
\Big\},
\]
where $\operatorname{supp}_{\mathrm{off}}(\Gamma)=\{(i,j):i<j,
\Gamma_{ij}\neq0\}$.  This is a finite union of closed convex cones and
is therefore a closed, generally nonconvex cone.  Its exact metric
projection jointly selects a support and solves a PSD-constrained
least-squares problem on that support.  The computational appendix uses
a transparent support-selection plus constrained-refinement algorithm;
the population theory concerns the exact metric projection.

The dictionary is intentionally illustrative rather than exhaustive: the
framework extends to any family of closed covariance geometries satisfying the
projection regularity assumptions. Cluster, factor, and sparse geometries are
adopted because they encompass the most widely used dependence structures in
econometric practice---group-correlated shocks, pervasive common variation, and
localized or network linkages. A researcher who suspects an additional structure
may enlarge the dictionary by adding the corresponding closed class, with no
change to the identification, impossibility, or adaptive-inference theory,
provided the new class admits a well-defined projection.

\subsection{Projection Operators}
\label{subsec:projections}

Throughout the paper, let
\[
\mathfrak D=\{C,F,S\}
\]
denote the covariance-cone dictionary, corresponding respectively to
cluster, factor, and sparse dependence. Define
\[
P_d(\Gamma)
=
\arg\min_{\Gamma'\in\mathcal S_d}
\|\Gamma-\Gamma'\|_F.
\]
$P_d(\Gamma)$ represents the closest approximation to $\Gamma$ within
the covariance class $\mathcal S_d$. Although we use projection
notation, $P_d$ need not be a linear orthogonal projection because the
covariance classes may be nonlinear or nonconvex.

\begin{lemma}[Existence of Projections]
\label{lem:projectionexistence}
If $\mathcal S_d$ is closed in $\mathbb H$, then $P_d(\Gamma)$ exists
for every $\Gamma\in\mathbb H$.
\end{lemma}

\begin{lemma}[Continuity of Projection Operators]
\label{lem:projectioncontinuity}
Suppose $P_d$ is locally unique at $\Gamma$. If $\Gamma_m\to\Gamma$,
then $P_d(\Gamma_m)\to P_d(\Gamma)$.
\end{lemma}

\subsection{Why the Frobenius Geometry?}
\label{subsec:frobenius}

The Frobenius norm is induced by the Hilbert inner product
$\langle A,B\rangle_F=\operatorname{tr}(A'B)$ and therefore provides a natural
geometry for projections, tangent spaces, principal angles, and local
separation. It is also computationally convenient: low-rank projections use
spectral updates, one-way cluster projections reduce to block support masks, and
general cluster or sparse projections reduce to PSD-constrained fixed-support
least squares (with support selection for sparse classes). Alternative metrics
may suit other applications, but the Frobenius geometry gives the most
transparent framework for the theory below.

\subsection{Dependence Profile}
\label{subsec:dep_profile}

Define the \emph{population similarity score} for geometry $d$ as
\[
S_{d,0}=\|P_d(\Gamma_0)\|_F^2,
\qquad d\in\mathfrak D.
\]
The normalized dependence similarity scores are
\[
\omega_{d,0}=\frac{S_{d,0}}{\sum_{d\in\mathfrak D}S_{d,0}},
\qquad d\in\mathfrak D.
\]
The vector
\[
\omega_0=(\omega_{C,0},\omega_{F,0},\omega_{S,0})'
\]
is called the \emph{population dependence profile}. It is a functional
of the estimable object $\Gamma_0$ and therefore inherits its
estimability from Proposition~\ref{prop:consistency_operator}.

\begin{remark}[No Variance Decomposition]
\label{rem:no_variance_decomp}
The quantities $\omega_{C,0},\omega_{F,0},\omega_{S,0}$ should not be
interpreted as fractions of covariance explained. Since the covariance
geometries overlap and the projections are not generally orthogonal,
$\sum_{d\in\mathfrak D}S_{d,0}\neq\|\Gamma_0\|_F^2$ in general. The
dependence profile measures relative geometric similarity, not additive
covariance contributions.
\end{remark}

\section{Identification of Dependence Geometry}
\label{sec:identification}
This section studies whether the dependence profile is well defined and
whether distinct covariance geometries are locally separated. The first
question is governed by regularity of the individual metric projections;
the second depends on whether the cluster, factor, and sparse geometries
have distinct local off-diagonal directions near the projected population
points.

Identification concerns the population operator; estimation concerns its
empirical counterpart. The identification theory of this section
therefore studies properties of $\Gamma_0$ and its projections,
independently of sampling variability.

For each $d\in\mathfrak D$, let
\[
\Sigma_d = P_d(\Gamma_0)
\]
denote the population projection of $\Gamma_0$ onto geometry $d$.
For each geometry \(d\in\mathfrak D\), let \(T_d(\Sigma_d)\) denote the
corresponding tangent space at this projected population point.
Since diagonal perturbations are common to all covariance geometries and
carry no information about cross-sectional dependence, identification is
based on the off-diagonal tangent spaces

\[
T_d^{\mathrm{off}}(\Sigma_d)
=
\left\{
H\in T_d(\Sigma_d)
:
H_{ii}=0
\ \text{for all } i
\right\},
\qquad
d\in\mathfrak D.
\]
\paragraph{Principal-Angle Condition.}
For two tangent spaces \(U\) and \(V\), define the smallest principal angle by
\[
\theta(U,V)
=
\arccos
\left(
\sup_{\substack{u\in U,\;v\in V\\ \|u\|_F=\|v\|_F=1}}
\bigl|\langle u,v\rangle_F\bigr|
\right),
\]
with the convention $\theta(U,V)=\pi/2$ if $U=\{0\}$ or $V=\{0\}$.

\begin{assumption}[Principal-Angle Identification]
\label{ass:principalangle}
There exists \(\theta_0>0\) such that
\[
\theta\bigl(T_i^{\mathrm{off}}(\Sigma_i),\,T_j^{\mathrm{off}}(\Sigma_j)\bigr)\ge \theta_0
\]
for all \(i\neq j\), \(i,j\in\mathfrak D\).
\end{assumption}
Formal definitions of the covariance geometries, regularity conditions,
and tangent spaces are given in the Online Appendix.

\begin{assumption}[Sparse Projection Regularity]
\label{ass:sparse_unique}

Let $\mathfrak I_S$ denote the finite collection of admissible symmetric
off-diagonal supports with cardinality at most $k_S$, and let
$\mathcal C_I=\{\Gamma\succeq0:\operatorname{supp}_{\mathrm{off}}(\Gamma)
\subseteq I\}$. The population sparse projection has a unique active
support $I_0$: there exists $\varepsilon_S>0$ such that
\[
\operatorname{dist}^2(\Gamma_0,\mathcal C_I)
\ge
\operatorname{dist}^2(\Gamma_0,\mathcal C_{I_0})+\varepsilon_S,
\qquad I\in\mathfrak I_S,\ I\neq I_0.
\]
This objective-gap condition, rather than an ordering of raw entries, is
the appropriate support-stability condition once positive
semidefiniteness is imposed.

\end{assumption}

\begin{lemma}[Principal Angle and Tangent-Space Transversality]
\label{lem:angle_transversality}
Let $U$ and $V$ be closed linear subspaces of $\mathbb H$. Then
$\theta(U,V)>0$ if and only if $U\cap V=\{0\}$.
\end{lemma}

\begin{theorem}[Well-Defined Dependence Profiles and Pairwise Geometric
Separation]
\label{thm:identification}
Suppose $\sum_{d\in\mathfrak D}S_{d,0}>0$ and the projected points
$\Sigma_d=P_d(\Gamma_0)$ are regular points at which the projections are
locally single-valued and continuous (Assumption~\ref{ass:sparse_unique}
for the sparse geometry). Then:
\begin{enumerate}
\item[(i)] \emph{(Well-definedness.)} The map
$\Gamma\mapsto\omega(\Gamma)=(\omega_C,\omega_F,\omega_S)(\Gamma)$ is
single-valued and continuous in a neighborhood of $\Gamma_0$; in
particular $\omega(\Gamma_0)$ is a well-defined functional of the
population dependence operator and is estimable at the rate of
Theorem~\ref{thm:profile_consistency}.
\item[(ii)] \emph{(Pairwise tangent separation.)} If in addition
Assumption~\ref{ass:principalangle} holds, then for every $i\neq j$ the
off-diagonal tangent spaces at the corresponding projected points
intersect trivially,
\[
T_i^{\mathrm{off}}(\Sigma_i)\cap T_j^{\mathrm{off}}(\Sigma_j)=\{0\}.
\]
\end{enumerate}
\end{theorem}

\begin{remark}[What Theorem~\ref{thm:identification} does and does not
assert]
\label{rem:identification_scope}
Part~(i) is a statement of well-definedness and estimability, not of
injectivity: because $\omega$ takes values in the two-dimensional simplex
while $\mathbb H$ has dimension $N(N+1)/2$, the profile map cannot be
locally injective on a full neighborhood of $\Gamma_0$, and no such claim
is made. The profile is a deliberate low-dimensional summary; the
projection residuals $\rho_d$ record what it discards.

Part~(ii) asserts pairwise trivial intersection only. A positive
principal angle means the tangent spaces do not intersect; it does not
make them orthogonal, so a perturbation lying in $T_i^{\mathrm{off}}$ may
still have a nonzero projection onto $T_j^{\mathrm{off}}$ and move more
than one score to first order. Nor do pairwise trivial intersections
imply a direct-sum decomposition when three or more geometries are
present---three distinct lines in a plane intersect pairwise trivially
without spanning independently. Attributing a first-order profile change
to a \emph{unique} geometry therefore requires the strictly stronger
condition that
$T_C^{\mathrm{off}}\oplus T_F^{\mathrm{off}}\oplus T_S^{\mathrm{off}}$ be
a direct sum on the relevant perturbation space, or equivalently that the
Jacobian of the score vector have full rank; we do not impose it and do
not claim unique attribution. What Assumption~\ref{ass:principalangle}
delivers, and what the subsequent theory uses, is that distinct
geometries are locally distinguishable in the pairwise sense of~(ii). The
complementary configuration, in which two geometries share a common point
\emph{and} a common tangent direction, is exactly the source of the
fundamental limit in Theorem~\ref{thm:fundamental_limit}.
\end{remark}

\begin{assumption}[Local Asymptotic Normality]
\label{ass:lan}
Let $P_\Gamma^T$ denote the joint distribution of $(u_1,\ldots,u_T)$
when the cross-sectional covariance matrix is $\Gamma$. The statistical
model $\{P_\Gamma^T:\Gamma\in\mathbb H\}$ satisfies local asymptotic
normality (LAN) at $\Gamma_0$, \emph{uniformly over bounded local
perturbations}: for every bounded sequence $H_T\to H$ in $\mathbb H$, with
the $o_p(1)$ remainder below uniform over $\|H\|_{\mathcal I}\le B$ for
each fixed $B<\infty$,
\[
\log\frac{dP_{\Gamma_0+T^{-1/2}H_T}^T}{dP_{\Gamma_0}^T}
=
\frac{1}{\sqrt T}\sum_{t=1}^T \ell_t(H)
-
\tfrac12\|H\|_{\mathcal I}^2
+o_p(1),
\]
where $\ell_t(H)$ is a zero-mean score with $E[\ell_t(H)^2]=\|H\|_{\mathcal I}^2$ for
a positive-definite inner product $\|\cdot\|_{\mathcal I}$ on $\mathbb H$.
The local perturbations $\Gamma_0+T^{-1/2}H_T$ are required to be
positive semidefinite and to belong to the parameter space for all
sufficiently large $T$; when $\Gamma_0$ is positive definite this holds
for every fixed $H\in\mathbb H$, and otherwise $H$ is restricted to
admissible tangent directions.
In addition, the local experiment is asymptotically continuous in
Hellinger distance: for any bounded local sequences $H_T$ and $G_T$ with
$\|H_T-G_T\|_{\mathcal I}\to0$,
\[
H\!\left(
P_{\Gamma_0+T^{-1/2}H_T}^T,\,
P_{\Gamma_0+T^{-1/2}G_T}^T
\right)\to0.
\]
This property follows from differentiability in quadratic mean in the
regular finite-dimensional covariance experiments for which LAN is
invoked.
\end{assumption}
Assumption~\ref{ass:lan} is imposed only for the local impossibility
result (Theorem~\ref{thm:fundamental_limit}). It holds in regular
finite-dimensional covariance models, for example Gaussian experiments
with covariance parameter $\Gamma$; see \citet{LeCamYang2000},
Chapter~7. The assumption should be interpreted as a local regularity
condition on the statistical experiment at $\Gamma_0$, not as an
additional parametric restriction on the cross-sectional covariance
geometry being learned.
\subsection{Fundamental Limits of Dependence Learning}
The next result shows that ambiguity in dependence learning is not merely a
finite-sample phenomenon: when two geometries share a common off-diagonal tangent
direction, the data contain no first-order information distinguishing them along
local alternatives, so near-ties should be read as intrinsic ambiguity at the
relevant local scale rather than procedure failure.

\begin{theorem}[Fundamental limit of dependence learning]
\label{thm:fundamental_limit}
Let $\mathcal S_i$ and $\mathcal S_j$ be two covariance geometries with
$i\neq j$.
Suppose Assumption~\ref{ass:lan} holds at a regular point
$\Gamma_0\in\mathcal S_i\cap\mathcal S_j$, so that
$P_i(\Gamma_0)=P_j(\Gamma_0)=\Gamma_0$ and the off-diagonal tangent
spaces $T_i^{\mathrm{off}}(\Gamma_0)$, $T_j^{\mathrm{off}}(\Gamma_0)$
are evaluated at $\Gamma_0$ itself. Suppose there exists a nonzero
off-diagonal direction
$H \in T_i^{\mathrm{off}}(\Gamma_0)\cap T_j^{\mathrm{off}}(\Gamma_0)$.
 
Let $\Gamma_{i,T}\in\mathcal S_i$ and $\Gamma_{j,T}\in\mathcal S_j$ satisfy
\[
\Gamma_{i,T}
=
\Gamma_0+\frac{1}{\sqrt T}H+o(T^{-1/2}),
\qquad
\Gamma_{j,T}
=
\Gamma_0+\frac{1}{\sqrt T}H+o(T^{-1/2}).
\]
 
Then, for any sequence of tests $\varphi_T\in[0,1]$,
\[
\limsup_{T\to\infty}
\bigl|E_{\Gamma_{i,T}}\varphi_T-E_{\Gamma_{j,T}}\varphi_T\bigr|
=0.
\]
Consequently, no test can have asymptotic size tending to zero and power
tending to one for distinguishing $\mathcal S_i$ from $\mathcal S_j$ along
these local sequences.
\end{theorem}

Theorem~\ref{thm:fundamental_limit} establishes \emph{pairwise local
indistinguishability}: along these local sequences every test has asymptotically
identical rejection probability under the two geometries, so no classifier
separates them with non-trivial asymptotic power. We deliberately do not describe
this as a minimax lower bound---no quantified minimax risk bound is proved---but
as an impossibility statement along specified local sequences, which is what the
subsequent theory uses. It explains why the profile is deliberately continuous
rather than a hard model-selection device.

\begin{remark}[Identification versus Ambiguity]
\label{rem:identification_vs_ambiguity}
Theorems~\ref{thm:identification} and
\ref{thm:fundamental_limit} characterize the boundary between identifiable and ambiguous dependence
structures. Positive principal angles imply pairwise local separation of
the geometries, whereas overlapping tangent spaces imply local
indistinguishability.

\end{remark}


\section{Statistical Learning of Dependence Profiles}
\label{sec:estimation}

This section describes how the population dependence profile $\omega_0$
is estimated from data. The empirical dependence operator
$\widehat\Gamma_T$ was introduced in~\eqref{eq:sample_cov} and shown to
be consistent for $\Gamma_0$ in Proposition~\ref{prop:consistency_operator}.
Dependence profiles, projection-residual diagnostics, and their
sample counterparts are defined below.

\subsection{Choice of Dependence Operator}
\label{subsec:operatorchoice}

The empirical operator $\widehat\Gamma_T$ should capture the dependence relevant
for the inferential problem: the residual covariance
operator~\eqref{eq:sample_cov} for contemporaneous cross-sectional dependence, a
long-run covariance operator under serial dependence, or distance-, adjacency-,
or cluster-based operators for spatial, network, or clustered data. The framework
does not select a unique operator; it studies the covariance geometry encoded in
a chosen $\Gamma_0$. The dependence profile is invariant to positive rescalings
of the operator; additional examples and technical details are in the Online
Appendix.

\subsection{Projection-Residual Diagnostics}
\label{subsec:projection_residuals}

The dependence profile provides a relative measure of geometric fit.
To assess goodness-of-fit in an absolute sense, we introduce
projection-residual diagnostics.

Define the normalized projection residual
\[
\rho_d(\Gamma)
=
\frac{
\|
\Gamma-P_d(\Gamma)
\|_F
}
{
\|\Gamma\|_F
},
\qquad
d\in\mathfrak D,\ \Gamma\neq0.
\]
The quantity $\rho_d(\Gamma)$ measures the fraction of the operator
that remains unexplained after projection onto geometry $d$. Since
$0\in\mathcal S_d$, the definition of the projection implies
$\|\Gamma-P_d(\Gamma)\|_F\le\|\Gamma\|_F$, so $0\le\rho_d(\Gamma)\le1$.
Small values indicate that geometry $d$ provides a good approximation
to the dependence operator, whereas large values indicate substantial
lack of fit.

Define the minimum residual
\[
\rho_{\min}(\Gamma)
=
\min_{d\in\mathfrak D}
\rho_d(\Gamma).
\]
We say that a dependence operator exhibits a ``none-of-the-above''
pattern whenever $\rho_{\min}(\Gamma)>\delta$ for a prespecified
threshold $\delta\in(0,1)$. Large values of $\rho_{\min}$ indicate
that none of the cluster, factor, or sparse geometries provides an
adequate approximation to the observed dependence structure.

\begin{proposition}
\label{prop:residual_consistency}
Suppose Assumptions~\ref{ass:local_projection_regularity} and~\ref{ass:mixing} hold. Then
\[
\rho_d(\widehat\Gamma_T)
\overset{p}{\longrightarrow}
\rho_d(\Gamma_0),
\qquad
d\in\mathfrak D,
\]
and $\rho_{\min}(\widehat\Gamma_T)\overset{p}{\longrightarrow}
\rho_{\min}(\Gamma_0)$.
\end{proposition}

\subsection{Similarity-Score Estimation}
\label{subsec:similarity_score_estimation}

Define the estimated projection, similarity score, and dependence
profile as
\[
\widehat P_d = P_d(\widehat\Gamma_T),
\qquad
\widehat S_d = \|\widehat P_d\|_F^2,
\qquad
\widehat\omega_d = \frac{\widehat S_d}{\sum_{d\in\mathfrak D}\widehat S_d},
\qquad d\in\mathfrak D.
\]
The estimated dependence profile is
$\widehat\omega=(\widehat\omega_C,\widehat\omega_F,\widehat\omega_S)'$.

The use of the projection \emph{magnitude} $\|P_d(\Gamma)\|_F^2$ rather
than the projection \emph{distance} $\|\Gamma-P_d(\Gamma)\|_F^2$ requires
justification, since the classes are nonconvex and the two are equivalent
only through a Pythagorean identity that is automatic for linear
subspaces but not for general sets. The following lemma supplies it.

\begin{lemma}[Conic Pythagoras]
\label{lem:conic_pythagoras}
Each geometry $\mathcal S_d$, $d\in\mathfrak D$, is a cone containing the
origin: $\Gamma\in\mathcal S_d$ and $a\ge 0$ imply $a\Gamma\in\mathcal
S_d$. Consequently, for every $\Gamma\in\mathbb H$ and every metric
projection $p=P_d(\Gamma)$,
\[
\langle \Gamma-p,\;p\rangle_F=0,
\qquad
\|\Gamma\|_F^2=\|P_d(\Gamma)\|_F^2+\|\Gamma-P_d(\Gamma)\|_F^2 .
\]
\end{lemma}

\noindent
The proof is in the Online Appendix.

\noindent
Lemma~\ref{lem:conic_pythagoras} shows that the similarity score and the
projection residual are two encodings of the same information:
\begin{equation}
S_d(\Gamma)=\|\Gamma\|_F^2\bigl(1-\rho_d(\Gamma)^2\bigr),
\qquad
\rho_d(\Gamma)=\frac{\|\Gamma-P_d(\Gamma)\|_F}{\|\Gamma\|_F},
\qquad \Gamma\neq0 .
\label{eq:score_residual_duality}
\end{equation}
In particular $S_i(\Gamma)>S_j(\Gamma)$ if and only if
$\|\Gamma-P_i(\Gamma)\|_F<\|\Gamma-P_j(\Gamma)\|_F$: ranking geometries by
projection magnitude is \emph{equivalent} to ranking them by
best-approximation distance, so the dominant geometry of
Section~\ref{sec:classification} is exactly the best-fitting geometry, and
no separate justification of the score is required.

\begin{lemma}[Diagonal invariance of the geometric ranking]
\label{lem:diagonal_invariance}
Suppose each projection reproduces the diagonal of its argument,
$\operatorname{diag}(P_d(\Gamma))=\operatorname{diag}(\Gamma)$ for all
$d\in\mathfrak D$. Then
\[
S_d(\Gamma)
=
\|\operatorname{diag}(\Gamma)\|_F^2
+
\|\operatorname{off}(P_d(\Gamma))\|_F^2 ,
\]
so the scores differ across geometries only through their off-diagonal
parts. Consequently the dominant geometry satisfies
$d^\star=\arg\max_d\|\operatorname{off}(P_d(\Gamma))\|_F^2$ and is
invariant to the diagonal of $\Gamma$.
\end{lemma}

\noindent
The proof is in the Online Appendix.

\noindent
The diagonal condition holds exactly for the cluster geometry (unit-diagonal
support) and the sparse geometry (off-diagonal thresholding), and for the factor
geometry whenever the idiosyncratic nonnegativity constraint is slack, since
$D=\operatorname{diag}(\Gamma-L)$ then reproduces $\operatorname{diag}(\Gamma)$.
When it holds for every geometry, the $S_d$ ranking depends only on off-diagonal
dependence, aligning classification with the principal-angle analysis. For the
factor geometry, when the constraint binds the diagonal can affect the
full-profile ranking, so the full and off-diagonal profiles need not rank
geometries identically. The off-diagonal profile is therefore a formal companion
estimand, not an optional robustness check:
\begin{equation}
S_d^{\mathrm{off}}(\Gamma)
=\|\operatorname{off}\{P_d(\Gamma)\}\|_F^2,
\qquad
\omega_d^{\mathrm{off}}(\Gamma)
=\frac{S_d^{\mathrm{off}}(\Gamma)}
{\sum_{j\in\mathfrak D}S_j^{\mathrm{off}}(\Gamma)},
\label{eq:off_profile}
\end{equation}
whenever the denominator is positive, with sample counterpart
\[
\widehat\omega_d^{\mathrm{off}}
=\frac{\|\operatorname{off}(\widehat P_d)\|_F^2}
{\sum_{j}\|\operatorname{off}(\widehat P_j)\|_F^2}.
\]
The full profile measures overall
covariance fit; the off-diagonal profile measures dependence architecture after
removing marginal variances. Classification and profile-guided inference use the
off-diagonal profile as the primary rule (the identification analysis is on the
off-diagonal tangent spaces), with the full profile as a companion diagnostic;
all consistency and delta-method results below apply jointly to
$(\omega,\omega^{\mathrm{off}})$ by stacking the two smooth score maps whenever
both denominators are positive.

\paragraph{Why the off-diagonal profile.}
It is natural to ask why geometry learning is based on
$\omega^{\mathrm{off}}$ rather than on the full profile $\omega$. The reason is
that the diagonal of the covariance operator carries \emph{marginal variance}
information that is common to all candidate geometries and largely uninformative
about which \emph{dependence} structure generated the data. Every geometry in the
dictionary reproduces (or nearly reproduces) $\operatorname{diag}(\Gamma)$, so the
diagonal contributes a common, geometry-invariant term
$\|\operatorname{diag}(\Gamma)\|_F^2$ to each score $S_d$. Including it in the
profile therefore adds the \emph{same} constant to every numerator, mechanically
pulling the full profile toward the uniform vector $(1/3,1/3,1/3)$ and
compressing exactly the differences across geometries that classification must
detect. This diagonal term also makes the full profile sensitive to the
\emph{scale} of individual series: rescaling one unit's disturbance inflates its
variance and hence its diagonal contribution, shifting $\omega$ without changing
the dependence pattern at all. The off-diagonal profile removes both effects. By
restricting the scores to $\operatorname{off}(P_d)$ it discards the common
diagonal term and the marginal variances, isolating the cross-sectional
co-movement that actually distinguishes clustering, factor, and sparse
dependence. This is also why the identification and impossibility results are
stated on the off-diagonal tangent spaces: it is there, not on the diagonal, that
two geometries either separate or overlap. The full profile is retained because
it answers a different and still useful question---how well the dictionary
approximates the entire operator, diagonal included---and because a large gap
between $\omega$ and $\omega^{\mathrm{off}}$ is itself a diagnostic that marginal
heterogeneity is masking the dependence geometry.

\begin{proposition}[Joint consistency of full and off-diagonal profiles]
\label{prop:joint_profile_consistency}
Under Assumptions~\ref{ass:local_projection_regularity}
and~\ref{ass:mixing}, and if both normalizing denominators are positive,
then
\[
(\widehat\omega-\omega_0,
 \widehat\omega^{\mathrm{off}}-\omega_0^{\mathrm{off}})
=O_p(T^{-1/2}).
\]
Under Assumptions~\ref{ass:operatorclt}
and~\ref{ass:projection_differentiability}, and the same positivity
conditions,
\[
\sqrt T
\begin{pmatrix}
\widehat\omega-\omega_0\\
\widehat\omega^{\mathrm{off}}-\omega_0^{\mathrm{off}}
\end{pmatrix}
\Rightarrow N(0,\Xi_{\mathrm{joint}}),
\]
where $\Xi_{\mathrm{joint}}$ is obtained by applying the stacked
full- and off-diagonal-profile Jacobian to $\Omega_\Gamma$.  We write
$\Xi_{\mathrm{off}}$ for its off-diagonal-profile marginal block.
\end{proposition}

\section{Asymptotic Theory}
\label{sec:asymptotics}

This section establishes consistency and asymptotic normality of the
estimated dependence profile $\widehat\omega$. Asymptotics are
throughout in $T\to\infty$ with $N$ fixed, as stated in
Remark~\ref{rem:asymptotic_regime}.

\begin{assumption}[Asymptotic Linearity of the Dependence Operator]
\label{ass:operatorclt}
There exists a mean-zero random vector $\psi_t\in\mathbb{R}^{N^2}$ such
that
\[
\sqrt T\,
\operatorname{vec}
(
\widehat\Gamma_T-\Gamma_0
)
=
\frac{1}{\sqrt T}
\sum_{t=1}^T
\psi_t
+
o_p(1),
\]
and
\[
\frac{1}{\sqrt T}\sum_{t=1}^T\psi_t
\Rightarrow
N(0,\Omega_\Gamma),
\qquad
\Omega_\Gamma
=
\sum_{h=-\infty}^{\infty}
E(\psi_0\psi_h'),
\]
the series converging absolutely. Because $\{\psi_t\}$ is serially
dependent in general, $\Omega_\Gamma$ is the \emph{long-run} covariance
of $\{\psi_t\}$ and not the contemporaneous second moment
$E(\psi_t\psi_t')$; the distinction matters because $\Xi$ in
Theorem~\ref{thm:profile_clt} is built from $\Omega_\Gamma$.
\end{assumption}

This assumption isolates the sampling uncertainty in the estimated dependence
operator and is the only ingredient needed to transport a time-series central
limit theorem through the geometric projection and normalization maps.

\begin{remark}[Primitive conditions for Assumption~\ref{ass:operatorclt}]
\label{rem:primitive_clt}
Assumption~\ref{ass:operatorclt} is a high-level linear-representation condition
implied by Assumption~\ref{ass:mixing}: the functional CLT for the sample
covariance holds with $\psi_t=\operatorname{vec}(u_tu_t'-\Sigma)$ and
$\Omega_\Gamma$ the long-run variance of $\{u_tu_t'\}$ \citep[Theorem~27.4]{Davidson1994}
when $u_t$ is observed. When $\widehat\Gamma_T$ is built from residuals
$\widehat u_t=u_t-X_t(\widehat\beta-\beta)$, residual replacement is
asymptotically negligible and $\psi_t$ is unchanged: the leading correction
$T^{-1}\sum_t X_t(\widehat\beta-\beta)u_t'$ is a bilinear form in
$\widehat\beta-\beta=O_p(T^{-1/2})$ and the cross-moment
$T^{-1}\sum_t X_t\otimes u_t$, which is itself $O_p(T^{-1/2})$ under the strict
exogeneity $E(X_tu_t')=0$, so the correction is $O_p(T^{-1})=o_p(T^{-1/2})$ and no
first-step term enters $\Omega_\Gamma$. (If strict exogeneity fails, the
cross-moment is $O_p(1)$, the correction is $O_p(T^{-1/2})$, and an explicit
first-step term must be added.) These conditions restrict only the
\emph{temporal} dependence of $\{u_t\}$, not the cross-sectional operator
$\Gamma_0$ (Remark~\ref{rem:two_dependence_notions}).
\end{remark}

\begin{assumption}[Local Projection Regularity]
\label{ass:local_projection_regularity}
For every $d\in\mathfrak D$, the metric projection $P_d$ is single-valued
and locally Lipschitz on a neighborhood of $\Gamma_0$.  For the sparse
geometry, Assumption~\ref{ass:sparse_unique} is a sufficient support-stability
condition once the fixed-support PSD projection is regular.  For the factor
geometry this is a high-level local regularity condition; no global uniqueness
of the nonconvex projection is imposed.
\end{assumption}

Local regularity rules out points at which an arbitrarily small perturbation of
$\Gamma_0$ changes the relevant projection branch or active support; it is a
local learnability requirement, not a global uniqueness claim.

\begin{lemma}[Uniform Consistency of Projection Estimators]
\label{lem:uniform_projection_consistency}
Suppose Assumptions~\ref{ass:local_projection_regularity} and~\ref{ass:mixing} hold. Then
\[
\max_{d\in\mathfrak D}
\|
\widehat P_d-P_{d,0}
\|_F
=
O_p(T^{-1/2}).
\]
\end{lemma}

\begin{theorem}[Consistency of the Dependence Profile]
\label{thm:profile_consistency}
Suppose Assumptions~\ref{ass:local_projection_regularity} and~\ref{ass:mixing} hold. If
$\sum_{d\in\mathfrak D}S_{d,0}>0$, then
$\widehat\omega-\omega_0=O_p(T^{-1/2})$.
\end{theorem}

To state the limiting distribution, define the score map
\[
\mathcal S(\Gamma)
=
\left(
\|P_C(\Gamma)\|_F^2,
\|P_F(\Gamma)\|_F^2,
\|P_S(\Gamma)\|_F^2
\right)'
\]
and let $s_0=\mathcal S(\Gamma_0)=(S_{C,0},S_{F,0},S_{S,0})'$. Let
$\pi(s)=s/(\mathbf 1's)$ denote the normalization map, so that
$\omega_0=\pi(s_0)$.

\begin{assumption}[Projection Differentiability]
\label{ass:projection_differentiability}
Each projection map $P_d$, $d\in\mathfrak D$, is Hadamard differentiable
at $\Gamma_0$, with derivative $\dot P_{d,\Gamma_0}:\mathbb H\to\mathbb H$.
\end{assumption}

Differentiability converts local perturbations of the covariance operator into
local perturbations of the profile, and is precisely the condition that makes
standard errors for profile coordinates and pairwise contrasts available by the
functional delta method.

\begin{remark}[Sufficient conditions]
Assumption~\ref{ass:projection_differentiability} is stronger than
Assumption~\ref{ass:local_projection_regularity}. For a fixed-support cluster or
sparse cone, the projection is continuously differentiable where the active PSD
face is locally stable and the projected point lies on a regular stratum,
reducing to the orthogonal projection onto the support subspace when the PSD
constraint is slack; at an active PSD boundary the derivative is the derivative of
the metric projection onto the active convex cone, not a simple support mask. For
the nonconvex factor geometry, local single-valuedness and Hadamard
differentiability are imposed at a regular projected point (distinct positive
leading eigenvalues help but do not by themselves suffice). The high-level
formulation thus covers all three geometries without asserting invalid global or
boundary formulas.
\end{remark}

\begin{lemma}[Differentiability of the Score Map]
\label{lem:score_differentiability}
Suppose Assumption~\ref{ass:projection_differentiability} holds. Then
the score map $\mathcal S(\cdot)$ is Hadamard differentiable at
$\Gamma_0$, with derivative $\dot{\mathcal S}_{\Gamma_0}$ given
coordinatewise by
\[
\dot S_{d,\Gamma_0}[H]
=
2\,\bigl\langle P_d(\Gamma_0),\,\dot P_{d,\Gamma_0}[H]\bigr\rangle_F,
\qquad d\in\mathfrak D.
\]
\end{lemma}

\begin{theorem}[Asymptotic Distribution of the Dependence Profile]
\label{thm:profile_clt}
Suppose Assumptions~\ref{ass:projection_differentiability} and~\ref{ass:operatorclt} hold. If
$\sum_{d\in\mathfrak D}S_{d,0}>0$, then
\[
\sqrt T
(
\widehat\omega-\omega_0
)
\Rightarrow
N(0,\Xi),
\]
where $\Xi=J_\omega\Omega_\Gamma J_\omega'$, with
$J_\omega=\dot\pi_{s_0}\dot{\mathcal S}_{\Gamma_0}$ and $\dot\pi_{s_0}$
the derivative of $\pi(s)=s/(\mathbf 1's)$ at $s_0$.
\end{theorem}

The theorem makes dependence learning operational: researchers can attach
standard errors to profile coordinates and to the contrasts that determine which
geometry dominates, providing the sampling-uncertainty input for the
classification and near-tie analysis below.

\begin{remark}[Degeneracy of the limiting covariance]
\label{rem:xi_singular}
Because the profile lies on the unit simplex, $\mathbf 1'\widehat\omega=1$
identically, so $\mathbf 1'\Xi\mathbf 1=0$ and $\Xi$ is singular with
rank at most $|\mathfrak D|-1=2$. Inference on $\omega_0$ should
therefore be conducted on any two coordinates (or on contrasts
$\omega_{d}-\omega_{d'}$), using a generalized inverse of $\Xi$ with the
corresponding reduced degrees of freedom; the pairwise contrasts used in
Section~\ref{sec:classification} are of this form and remain
asymptotically $\chi^2$ with the correct degrees of freedom.
\end{remark}


\section{Dependence Diagnostics and Classification}
\label{sec:classification}

The off-diagonal dependence profile is a low-dimensional summary of the
dependence architecture encoded in the empirical operator: large
$\widehat\omega^{\mathrm{off}}_C$, $\widehat\omega^{\mathrm{off}}_F$, and
$\widehat\omega^{\mathrm{off}}_S$ indicate geometric proximity to the cluster,
factor, and sparse classes. It measures relative geometric similarity, not
additive covariance contributions. A natural classification rule uses it
directly,
\begin{equation}
\widehat d
=
\operatorname*{arg\,max}_{d\in\mathfrak D}
\widehat{\omega}^{\mathrm{off}}_d,
\label{eq:classification_rule}
\end{equation}
identifying the class with the largest off-diagonal score. The full dependence
profile is a companion diagnostic and is often more informative than the label
alone: e.g.\ $(0.60,0.30,0.10)$ suggests cluster dominance with meaningful factor
dependence, whereas $(0.35,0.35,0.30)$ suggests hybrid dependence and cautions
against a binary classification.

\subsection{Statistical Uncertainty for Dependence Scores}
\label{subsec:profile_variance_estimation}

Inference for the profile is based on Theorem~\ref{thm:profile_clt},
$\sqrt T(\widehat\omega-\omega_0)\Rightarrow N(0,\Xi)$. With
$e_C,e_F,e_S$ the standard basis, the asymptotic variance of
$\widehat\omega_d$ is $\sigma_d^2=e_d'\Xi e_d$.

\begin{assumption}[Consistent Covariance Estimation]
\label{ass:Xi_estimation}
There exists an estimator \(\widehat{\Xi}\) such that
$\widehat{\Xi}-\Xi=o_p(1)$.
\end{assumption}

Under Assumption~\ref{ass:Xi_estimation}, $\widehat\sigma_d^2=e_d'\widehat\Xi e_d$
consistently estimates $\sigma_d^2$. The construction of $\widehat\Xi$ depends on
the chosen operator: when it admits an asymptotic linear representation, a
plug-in estimator based on the influence function and the delta method may be
used. A universal variance estimator for all possible operators is outside the
paper's scope.

\subsection{Diagnostic Tests}
\label{subsec:score_diagnostics}

Theorem~\ref{thm:profile_clt} yields standard diagnostics for the profile. The
single-coordinate statistic
$T_d=\sqrt T(\widehat\omega_d-\omega_d^\ast)/\widehat\sigma_d$, with
$\widehat\sigma_d^2=e_d'\widehat\Xi e_d$, is asymptotically $N(0,1)$ under
$H_0:\omega_d=\omega_d^\ast$. More generally, for a full-row-rank $q\times3$
matrix $R_\omega$, the Wald statistic
$W_\omega=T(R_\omega\widehat\omega-r_\omega)'(R_\omega\widehat\Xi
R_\omega')^{-1}(R_\omega\widehat\omega-r_\omega)$ is asymptotically $\chi^2_q$
under $H_0:R_\omega\omega_0=r_\omega$ (using a generalized inverse where
$\Xi$ is singular, per Remark~\ref{rem:xi_singular}). These quantify the
uncertainty in the estimated profile and should be read as diagnostics.

\subsection{Dominant and Hybrid Dependence}
\label{subsec:dominant_hybrid_classification}

A dominant geometry is suggested when one score substantially exceeds the others
(e.g.\ $\widehat\omega^{\mathrm{off}}_C>\max\{\widehat\omega^{\mathrm{off}}_F,
\widehat\omega^{\mathrm{off}}_S\}$ indicates proximity to the cluster geometry);
hybrid dependence is suggested when several scores are simultaneously large, in
which case the full profile should be reported rather than a single label. This
matters because geometries may overlap---cluster matrices may also be sparse when
cluster sizes are small---so the profile is a continuous description of
dependence architecture rather than a rigid model-selection device.

\subsection{Dominant Dependence Geometry}
\label{subsec:dominant_geometry}

A central question is whether the estimated profile can identify the dominant
geometry underlying the chosen operator.

\begin{assumption}[Unique Dominant Geometry]
\label{ass:unique_dominant_geometry}

There exists a unique geometry
\(d^{\star}\in\mathfrak D\)
such that

\[
\omega^{\mathrm{off}}_{d^{\star}}
>
\max_{d\neq d^{\star}}
\omega^{\mathrm{off}}_d.
\]
Equivalently, the separation margin
\[
\Delta_{\omega^{\mathrm{off}}}
=
\omega^{\mathrm{off}}_{d^{\star}}
-
\max_{d\neq d^{\star}}
\omega^{\mathrm{off}}_d
\]
satisfies $
\Delta_{\omega^{\mathrm{off}}}>0$.

\end{assumption}

Under Assumption~\ref{ass:unique_dominant_geometry}, define
\begin{equation}
d^{\star}
=
\argmax_{d\in\mathfrak D}
\omega_d^{\mathrm{off}},
\label{eq:dominant_geometry}
\end{equation}
the population dominant geometry. The assumption rules out knife-edge cases in
which two or more geometries have exactly the same population score, a separation
condition common in model selection and classification. Throughout, $P_{\Gamma_0}$
denotes probability under the sample law $P_{\Gamma_0}^T$ of
Assumption~\ref{ass:lan}, under which all stochastic-order statements and
probabilities involving $\widehat\omega$ and $\widehat d$ are taken.

\begin{proposition}[Classification Error Bound]
\label{prop:classification_bound}

Suppose Assumption
\ref{ass:unique_dominant_geometry}
holds. Then
\[
P_{\Gamma_0}(\widehat d\neq d^\star)
\leq
P_{\Gamma_0}
\left(
\max_{d\in\mathfrak D}
|
\widehat\omega_d^{\mathrm{off}}-\omega_d^{\mathrm{off}}
|
\ge
\frac{\Delta_{\omega^{\mathrm{off}}}}{2}
\right),
\]
where
\[
\Delta_{\omega^{\mathrm{off}}}
=
\omega_{d^\star}^{\mathrm{off}}
-
\max_{d\neq d^\star}
\omega_d^{\mathrm{off}}.
\]
\end{proposition}

Classification error can occur only when profile-estimation error exceeds one
half of the population separation margin, linking the finite-sample reliability
of the discrete recommendation to a continuous, estimable geometric quantity.

\begin{theorem}[Consistency of Dominant-Geometry Classification]
\label{thm:dominant_geometry_consistency}

Suppose the consistency conditions of
Proposition~\ref{prop:joint_profile_consistency} hold, including
positivity of the off-diagonal normalizing denominator, and
Assumption~\ref{ass:unique_dominant_geometry} holds. Then

\[
P\!\left(
\widehat d=d^{\star}
\right)
\longrightarrow
1.
\]

\end{theorem}

Theorem~\ref{thm:dominant_geometry_consistency} shows that the
\emph{off-diagonal} dependence profile consistently identifies the dominant
covariance geometry whenever the dominant off-diagonal similarity score is
separated from the remaining scores. Consequently, the off-diagonal dependence
profile is not merely a descriptive summary of dependence but also provides a
statistically consistent basis for dependence classification.

\begin{remark}[Relation to Statistical Classification]

Theorem~\ref{thm:dominant_geometry_consistency}
is analogous to consistency results in model selection and statistical
classification. The distinguishing feature here is that the objects
being classified are covariance geometries rather than parametric
models.

\end{remark}

The classification result provides a formal justification for using the
estimated off-diagonal dependence profile as a guide for procedure recommendation. The implications for inference are developed formally in Section \ref{sec:profile_guided_inference}.
\begin{theorem}[Local Alternatives and Near-Ties]
\label{thm:local_ties}

Suppose

\[
\sqrt T
\left(
\widehat\omega^{\mathrm{off}}-\omega_T^{\mathrm{off}}
\right)
\Rightarrow
N(0,\Xi_{\mathrm{off}}),
\]
where
\[
\omega_T^{\mathrm{off}}
=
(\omega_{C,T}^{\mathrm{off}},\omega_{F,T}^{\mathrm{off}},\omega_{S,T}^{\mathrm{off}})'.
\]

Assume that two geometries,
\(d_1\) and \(d_2\),
satisfy

\[
\omega_{d_1,T}^{\mathrm{off}}
-
\omega_{d_2,T}^{\mathrm{off}}
=
\frac{c}{\sqrt T},
\]
for some constant \(c\in\mathbb R\), while all remaining scores remain
separated by positive constants. Assume the contrast direction
$e=e_{d_1}-e_{d_2}$ satisfies $e'\Xi_{\mathrm{off}} e>0$, so that the pairwise contrast
is asymptotically nondegenerate. Let $P_{\Gamma_T}$ denote probability under the triangular sequence
whose population profile is $\omega_T^{\mathrm{off}}$. Then
\[
P_{\Gamma_T}(\widehat d=d_1)
\]
converges to a nondegenerate limit lying strictly between zero and one,
equal to $\Phi\!\bigl(c/\sqrt{e'\Xi_{\mathrm{off}}e}\bigr)$.

\end{theorem}

A local near-tie is not a defect of the method but a region in which the data
contain only limited first-order evidence favoring one geometry; the
nondegenerate selection probability records that uncertainty rather than
producing an artificially decisive recommendation.
\begin{remark}[Near-Ties]
\label{rem:near_ties}

When the separation margin
\(\Delta_{\omega^{\mathrm{off}}}\)
is close to zero, small sampling fluctuations may alter the dominant
geometry classification. In such situations, the off-diagonal dependence
profile $
\widehat{\omega}^{\mathrm{off}}
=
(\widehat{\omega}_C^{\mathrm{off}},
  \widehat{\omega}_F^{\mathrm{off}},
  \widehat{\omega}_S^{\mathrm{off}})'
$ is typically more informative than the discrete classifier
\(\widehat d\).  In this regime the recommended output is therefore the
profile, its contrast uncertainty, and multiple or hybrid inferential actions,
rather than a single geometry label.
\end{remark}


\section{Profile-Guided Inference}
\label{sec:profile_guided_inference}

The dependence profile provides a low-dimensional summary of the
dependence geometry encoded in an empirical dependence operator. A
natural question is whether the estimated profile can be used to guide
econometric inference. The objective of this section is not to propose a new covariance
estimator. Rather, we show how dependence learning can be combined with
existing dependence-robust procedures in a systematic way.

\begin{framed}
\noindent\textbf{Algorithm 1: Profile-Guided Adaptive Inference}
\label{alg:profile_guided}
\begin{enumerate}[leftmargin=2.2em,itemsep=0.25em,topsep=0.4em]
\item \textbf{Estimate the dependence operator.} Construct
$\widehat\Gamma_T$ from observed disturbances or residuals and, when needed,
apply the positive-semidefinite refinement described in
Section~\ref{sec:geometry}.
\item \textbf{Project onto the covariance dictionary.} For every
$d\in\mathfrak D$, compute the exact metric projection $P_d(\widehat\Gamma_T)$
or the documented numerical approximation to it.
\item \textbf{Construct diagnostics.} Compute the full profile
$\widehat\omega$, the off-diagonal profile $\widehat\omega^{\mathrm{off}}$,
and the projection residuals $\widehat\rho_d$.
\item \textbf{Assess learnability.} Determine the estimated dominant geometry
$\widehat d$, the off-diagonal margin
$\widehat\Delta_{\omega^{\mathrm{off}}}$, and the confidence index
$\widehat\kappa$.
\item \textbf{Choose an inferential action.} Apply the decision rule
$\delta(\widehat\omega^{\mathrm{off}})$: use the geometry-matched procedure
when the evidence is sufficiently separated, and use a hybrid or
multi-procedure report when the profile is ambiguous or the dictionary fit is
poor.
\end{enumerate}
\noindent\emph{Output:} the estimated dependence profile, projection-residual
diagnostics, confidence index, and recommended inference procedure.
\end{framed}

The complete procedure is the single mapping
$\widehat\Gamma_T\to(\widehat\omega,\widehat\omega^{\mathrm{off}},\widehat\rho)
\to(\widehat d,\widehat\Delta_{\omega^{\mathrm{off}}},\widehat\kappa)
\to\delta(\widehat\omega^{\mathrm{off}})\to\widehat{\mathcal P}$, where
$\widehat{\mathcal P}$ is the selected, hybrid, or multi-procedure report;
Figure~\ref{fig:workflow} displays it together with the three decision branches.

\begin{figure}[!t]
\centering
\begin{tikzpicture}[
  font=\footnotesize,
  >={Stealth[length=2mm]},
  x=1cm,y=1cm,
  proc/.style={rectangle, rounded corners=2pt, draw=black!75, fill=black!5,
               align=center, minimum height=9mm, inner sep=3pt, text width=29mm},
  hubn/.style={rectangle, rounded corners=2pt, draw=black!75, fill=black!5,
               align=center, minimum height=9mm, inner sep=3pt, text width=40mm},
  act/.style={rectangle, rounded corners=2pt, draw=black!80, fill=black!10,
              align=center, minimum height=13mm, inner sep=3pt, text width=35mm},
  arr/.style={->, draw=black!75, thick},
  lbl/.style={font=\scriptsize\itshape, fill=white, inner sep=1.5pt}
]
\node[proc] (data) at (0,6.2)  {\textbf{1.} Panel / residuals\\ $\{\widehat u_t\}_{t=1}^T$};
\node[proc] (op)   at (3.9,6.2){\textbf{2.} Estimate operator\\ $\widehat\Gamma_T=\tfrac1T\sum_t\widehat u_t\widehat u_t'$};
\node[proc] (proj) at (7.8,6.2){\textbf{3.} Project onto\\ dictionary $P_d(\widehat\Gamma_T)$};
\node[proc] (prof) at (11.7,6.2){\textbf{4.} Profiles \& residuals\\ $\widehat\omega,\widehat\omega^{\mathrm{off}},\widehat\rho_d$};
\draw[arr] (data) -- (op);
\draw[arr] (op) -- (proj);
\draw[arr] (proj) -- (prof);
\node[hubn] (learn) at (5.85,3.9) {\textbf{5.} Assess learnability:\\ $\widehat d,\ \widehat\Delta_{\omega^{\mathrm{off}}},\ \widehat\kappa=(1-\widehat\rho_{\min})\widehat\Delta_{\omega^{\mathrm{off}}}$};
\draw[arr] (prof.south) |- (learn.east);
\node[act] (matched)  at (0,1.4)  {\textbf{Matched procedure} $a_{\widehat d}$\\[1pt] \emph{oracle-equivalent}};
\node[act] (hybrid)   at (5.85,1.4){\textbf{Hybrid / multiple}\\[1pt] report several procedures};
\node[act] (caution)  at (11.7,1.4){\textbf{Caution}\\[1pt] expand dictionary; avoid geometry-specific SE};
\draw[arr] (learn.south) -- ($(learn.south)+(0,-0.55)$) -| (matched.north)
      node[lbl,pos=0.25,above]{large margin, small $\widehat\rho_{\min}$};
\draw[arr] (learn.south) -- (hybrid.north)
      node[lbl,pos=0.55]{near tie, adequate fit};
\draw[arr] (learn.south) -- ($(learn.south)+(0,-0.55)$) -| (caution.north)
      node[lbl,pos=0.25,above]{large $\widehat\rho_{\min}$};
\end{tikzpicture}
\caption{\textbf{Profile-guided adaptive inference workflow.} Steps 1--4 form the
descriptive pipeline from data to diagnostics; Step 5 summarizes learnability;
the three branches give the prescriptive action---the geometry-matched procedure
when one geometry is clearly separated and fits well (oracle-equivalent by
Theorem~\ref{thm:oracle_adaptivity_asymptotic_optimality}), a hybrid or
multi-procedure report near a tie, and a caution when the dictionary fits poorly.}
\label{fig:workflow}
\end{figure}

Algorithm~1 separates the descriptive learning problem from the prescriptive
decision problem. Steps 1--4 summarize what the data say about the covariance
dictionary; Step 5 maps that evidence into an inferential action. The oracle
result below evaluates the statistical cost of this final mapping.

\subsection{Profile-Guided Procedure Recommendation}
\label{subsec:procedure_recommendation}

Let $\widehat d$ denote the estimated dominant geometry
of~\eqref{eq:classification_rule}. The mapping to procedures is direct: when
$\widehat\omega_C^{\mathrm{off}}$ dominates, cluster-robust procedures are the
natural benchmark; when $\widehat\omega_F^{\mathrm{off}}$ dominates,
factor-robust or common-shock-adjusted procedures; and when
$\widehat\omega_S^{\mathrm{off}}$ dominates, sparse, network, spatial, or
local-dependence robust procedures. The recommendation is evidence-based rather
than deterministic---it summarizes which geometry is most consistent with the
operator, not that one procedure is universally correct. When several scores are
substantial the data support multiple mechanisms and reporting several procedures
is preferable, and when the projection residuals indicate poor fit the entire
dictionary should be viewed with caution.

\subsection{Why Learn the Geometry?}
\label{subsec:why_learn_geometry}

A natural alternative is a single broadly robust procedure---multiway
clustering, HAC inference, or a general sandwich estimator. Dependence learning
remains useful for three reasons. First, broad robustness is not costless:
procedures valid under large classes of dependence may be conservative or poorly
sized when the realized geometry differs from the structure for which they are
best suited. Second, such procedures still require the researcher to specify a
dependence class---clustering dimensions, a spatial metric, a bandwidth, or a
network---which are themselves assumptions about dependence. Third,
Theorem~\ref{thm:oracle_adaptivity_asymptotic_optimality} shows profile-guided
selection is asymptotically equivalent to an infeasible oracle that knows the
dominant geometry. Dependence learning thus provides a data-driven way to choose
among geometry-specific procedures while retaining oracle-equivalent first-order
behavior.

\subsection{Procedure Confidence Index}
\label{subsec:procedure_confidence}

The reliability of a profile-based recommendation depends on two considerations:
the dominant geometry should be well separated from the competitors, and the
dictionary should approximate the operator well. Combining the estimated
separation margin
$\widehat\Delta_{\omega^{\mathrm{off}}}=\max_{d}\widehat\omega_d^{\mathrm{off}}
-\max_{d\neq\widehat d}\widehat\omega_d^{\mathrm{off}}$ with the minimum
projection residual $\widehat\rho_{\min}=\min_d\widehat\rho_d$ yields the
\emph{procedure confidence index}
\[
\widehat\kappa
=
\bigl(1-\widehat\rho_{\min}\bigr)\widehat\Delta_{\omega^{\mathrm{off}}}.
\]
Large $\widehat\kappa$ indicates that one geometry is clearly dominant
\emph{and} the dictionary fits well; small values indicate weak separation, poor
fit, or both. Its population counterpart is
$\kappa_0=(1-\rho_{\min,0})\Delta_{\omega^{\mathrm{off}},0}$, with
$\rho_{\min,0}=\min_d\rho_d$,
$\Delta_{\omega^{\mathrm{off}},0}=\max_d\omega_d^{\mathrm{off}}
-\max_{d\neq d^\star}\omega_d^{\mathrm{off}}$, and
$d^\star=\arg\max_d\omega_d^{\mathrm{off}}$.

\begin{proposition}[Consistency of the Procedure Confidence Index]
\label{prop:kappa_consistency}

Suppose the conditions of
Theorem~\ref{thm:profile_consistency} and
Assumption~\ref{ass:unique_dominant_geometry} hold, so that the dominant
geometry $d^\star$ entering $\kappa_0$ is uniquely defined. Then

\[
\widehat\kappa
\overset{p}{\longrightarrow}
\kappa_0.
\]

%
%
%
%
%
%
%
\end{proposition}

\subsection{Profile-Guided Variance Estimation}
\label{subsec:profile_guided_variance}
Suppose a scalar parameter \(\theta_0\) is estimated by \(\widehat\theta_T\).
For each geometry \(d\in\mathfrak D\), let \(\widehat V_d\) be a measurable
variance estimator and let \(\tau_{T,d}\to\infty\) be a deterministic
normalization such that
\[
\tau_{T,d}^{2}Var(\widehat\theta_T)\longrightarrow V_d,
\qquad 0<V_d<\infty.
\]
Let
\[
d^\star
=
\arg\max_{d\in\mathfrak D}
\omega_d^{\mathrm{off}}
\]
denote the dominant population dependence geometry. Define 
\[
\widehat d
=
\arg\max_{d\in\mathfrak D}\widehat\omega_d^{\mathrm{off}},
\qquad
\widehat V^{*}
=
\widehat V_{\widehat d},
\]
where ties in the argmax are broken by a fixed deterministic rule. Thus, inference is based on the variance estimator associated with the
estimated dominant dependence geometry. 

\begin{proposition}[Consistency of the Profile-Guided Variance Estimator]
\label{prop:vcov_consistency}
Let
\(
\widehat V_d
\)
denote a geometry-specific variance estimator for
\(
d\in\mathfrak D
\).
Suppose the conditions of
Theorem~\ref{thm:dominant_geometry_consistency}
hold and

\[
\widehat V_{d^\star}
\overset{p}{\longrightarrow}
V_{d^\star}.
\]
Then
\[
\widehat V^{*}
\overset{p}{\longrightarrow}
V_{d^\star}.
\]
\end{proposition}

Proposition~\ref{prop:vcov_consistency} shows that profile-guided
selection preserves consistency whenever the dominant geometry is correctly
classified. It separates the two sources of validity: geometry-specific variance
estimation supplies validity conditional on the correct geometry, while
dependence learning supplies the probability of selecting it. No new variance
formula is required.

\subsection{Decision-Theoretic Formulation}
\label{subsec:decision}

The dependence profile is not itself the inferential target. It is a
sufficient low-dimensional summary for selecting among a finite dictionary
of inference procedures: operators with the same profile call for the same
procedure, so the profile retains exactly the selection-relevant
information. The decision-theoretic analysis below concerns this
procedure-selection problem, not covariance estimation per se.

The profile-guided procedure admits a natural decision-theoretic
interpretation in the sense of \citet{Wald1950}: the objective is not to estimate
the dependence structure for its own sake, but to learn enough about the
dependence geometry to select an appropriate inference procedure. Let
$\mathfrak D$ be the finite dictionary ($\mathfrak D=\{C,F,S\}$ in the baseline)
and let $L(a,\Gamma_0)$ be the loss from action $a\in\mathcal A$ when the
population operator is $\Gamma_0$. The loss may measure size distortion, coverage
error, expected interval length subject to coverage, power loss, or a weighted
combination; for example,
$L(a,\Gamma_0)=\lambda_1|P_{\Gamma_0}\{0\notin CI_a\}-\alpha|
+\lambda_2 E_{\Gamma_0}[\ell(CI_a)]$, with $CI_a$ the interval produced by $a$,
$\ell(CI_a)$ its length, and $\lambda_1,\lambda_2\ge0$ weights. The oracle rule
$\delta^\star:\Gamma\mapsto\arg\min_{a\in\mathcal A}L(a,\Gamma)$ is evaluated at
$\Gamma_0$.

\begin{assumption}[Oracle Action]
\label{ass:oracle_action}
\emph{(Profile--loss compatibility.)} The loss function and the action
set are such that the oracle rule selects the procedure associated with
the dominant dependence geometry:
$\delta^\star(\Gamma_0)=a_{d^\star}$.
\end{assumption}

\noindent
Assumption~\ref{ass:oracle_action} is a genuine restriction, not a definition:
$d^\star$ is defined by the largest similarity score---a \emph{descriptive}
ranking of fit---whereas the oracle minimizes an \emph{inferential} loss, and the
two orderings need not coincide, so the regret bound of
Theorem~\ref{thm:oracle_adaptivity_asymptotic_optimality} is conditional on their
alignment. It holds in two leading cases. \emph{(a) Size control under correct
specification:} if $a_d$ is consistent when $\Gamma_0\in\mathcal S_d$, the loss is
absolute asymptotic size distortion, and $\Gamma_0\in\mathcal S_{d^\star}$ exactly,
then $\rho_{d^\star}=0$, $d^\star$ uniquely maximizes the score, and $a_{d^\star}$
alone has zero distortion. \emph{(b) Frobenius risk:} if
$L(a_d,\Gamma_0)=\|\Gamma_0-P_d(\Gamma_0)\|_F^2$, then
by~\eqref{eq:score_residual_duality} minimizing loss is equivalent to maximizing
the score. Compatibility can also fail---if two geometries fit almost equally
well but one yields a far more variable variance estimator, a risk-based oracle
may prefer the worse-fitting but more stable procedure---so it is exactly the
statement that the loss ranks procedures the way projection distance ranks
geometries. In practice $\Gamma_0$ is unknown, and the profile-guided rule uses
the estimated profile $\widehat\omega$.

The profile-guided action is
defined as
\[
        \widehat a
        =
        \delta(\widehat\omega),
\]
where $\delta:\Delta^{|\mathfrak D|-1}\to\mathcal A$ is a decision rule
mapping dependence profiles into inference actions. A simple rule is
\[
        \delta(\widehat\omega)
        =
        a_{\widehat d},
\]
where $\widehat d$ is the estimated dominant geometry. A more conservative rule uses the separation margin
\[
        \widehat\Delta_{\omega^{\mathrm{off}}}
        =
        \widehat\omega_{\widehat d}^{\mathrm{off}}
        -
        \max_{d\neq \widehat d}\widehat\omega_d^{\mathrm{off}}
\]
and sets
\[
        \delta(\widehat\omega)
        =
        \begin{cases}
        a_{\widehat d},
        & \text{if } \widehat\Delta_{\omega^{\mathrm{off}}}>\gamma_T,\\
        a_H,
        & \text{if } \widehat\Delta_{\omega^{\mathrm{off}}}\leq \gamma_T,
        \end{cases}
\]
where $\gamma_T\downarrow 0$ is a tolerance sequence. The hybrid action
$a_H$ is intended for ambiguous regions in which no single dependence
geometry is clearly dominant.


The risk of a decision rule $\delta$ is

\[
R(\delta,\Gamma_0)
=
E_{\Gamma_0}
\left[
L(\delta(\widehat\omega),\Gamma_0)
\right].
\]
The regret of the profile-guided decision rule is
\[
\mathcal R(\delta,\Gamma_0)
=
E_{\Gamma_0}
\!\left[
L(\delta(\widehat\omega),\Gamma_0)
-
L(\delta^\star(\Gamma_0),\Gamma_0)
\right].
\]
Unlike estimation error, which evaluates the accuracy of an estimator,
regret evaluates the inferential cost of using a data-driven decision
rule instead of the infeasible oracle. It therefore provides a natural
decision-theoretic measure of the price of learning the dependence
structure before selecting an inference procedure.

The goal of profile-guided inference is to construct a feasible
decision rule $\delta$ whose regret converges to zero. This formalizes
the idea that learning the dependence structure before choosing an
inference procedure should not impose any first-order asymptotic loss
relative to the infeasible oracle.

It is important to be precise about \emph{what} the procedure adapts to. The
adaptivity here is not adaptivity of a covariance \emph{estimator}: the
coefficient estimator $\widehat\theta_T$ and the geometry-specific variance
estimators $\widehat V_d$ are held fixed and are not being improved. What is
learned, and what the rule adapts to, is the \emph{dependence geometry itself}---
the off-diagonal geometry $\widehat d$ that indexes which robust procedure is
appropriate. The oracle in Theorem~\ref{thm:oracle_adaptivity_asymptotic_optimality}
is the infeasible rule that \emph{knows the dominant off-diagonal geometry}
$d^\star$ and applies its matched procedure; the theorem states that using the
\emph{learned} geometry $\widehat d$ in its place costs nothing to first order.
Adaptation is therefore over the discrete set of candidate geometries, driven by
the off-diagonal dependence profile, not over a continuum of covariance
estimators.

 Let
\(
\widehat V_d
\)
be a variance estimator satisfying

\[
\widehat V_d
\overset{p}{\longrightarrow}
V_d,
\qquad
d\in\mathfrak D.
\] 
Let
\[
T_T^{*}
=
\frac{\tau_{T,\widehat {d}}\,(\widehat\theta_T-\theta_0)}
{\sqrt{\widehat V^{*}}}
\]
and
\[
T_T^{oracle}
=
\frac{\tau_{T,{d}^\star}\,(\widehat\theta_T-\theta_0)}
{\sqrt{\widehat V_{d^\star}}}.
\]

The theorem below concerns the pure profile-guided rule
$\delta(\widehat\omega)=a_{\widehat d}$. The hybrid rule is discussed as a
conservative alternative in ambiguous regions; its regret properties
require additional conditions and are not covered by the theorem.

We present the result as a corollary-strength consequence of
classification consistency rather than as a deep oracle theorem. Once
$P_{\Gamma_0}(\widehat d=d^\star)\to1$, the selected procedure coincides
with the oracle with probability tending to one and pointwise equivalence
follows almost immediately; the substantive content of the theory lies in
the identification and impossibility results, which delimit \emph{when}
classification consistency is available at all. The result is pointwise:
it is not a uniform oracle inequality, carries no regret rate, and is
silent in the near-tie region of Theorem~\ref{thm:local_ties}, where
classification consistency fails by construction. A uniform oracle
inequality or local-margin regret rate --- in particular for the hybrid
rule near ties --- would be a genuine strengthening and is left open.

\begin{theorem}[Main Theorem: Oracle Adaptivity and Decision-Theoretic Optimality]
\label{thm:oracle_adaptivity_asymptotic_optimality}
Suppose the assumptions of Theorem~\ref{thm:dominant_geometry_consistency}
hold. For part (iii), assume additionally
Assumption~\ref{ass:oracle_action} and that the loss is bounded. Then
\begin{enumerate}

\item[(i)]
The profile-guided variance estimator is asymptotically equivalent
to the infeasible oracle estimator,

\[
\widehat {V}^{*}
-
\widehat V_{d^\star}
=
o_p(1).
\]

\item[(ii)] If, in addition, the oracle statistic satisfies

\[
T_T^{oracle}
\Rightarrow
N(0,1),
\]
then

\[
T_T^{*}
-
T_T^{oracle}
=
o_p(1).
\]

Consequently,

\[
T^{*}_{T}
\Rightarrow N(0,1),
\]
whenever the oracle statistic is asymptotically standard normal. Consequently, profile-guided inference and oracle inference have the same first-order asymptotic distribution.

\item[(iii)] Assume further that the loss function satisfies
\[
|L(a,\Gamma_0)|
\le M
\]
uniformly over $a\in\mathcal A$.  The regret of the profile-guided rule satisfies
\[
        \mathcal R(\delta,\Gamma_0)       
        \to 0 .
\]

Moreover,

\[
\mathcal R(\delta,\Gamma_0)
\le
2M
P_{\Gamma_0}(\widehat d\neq d^\star)\leq
        2M
        P_{\Gamma_0}
        \left(
        \max_{d\in\mathfrak D}
        |\widehat\omega_d^{\mathrm{off}}-\omega_{d,0}^{\mathrm{off}}|
        \ge
        \frac{\Delta_{\omega^{\mathrm{off}}}}{2}
        \right),
\]

so the rate at which regret vanishes is governed by the probability of
incorrect dependence classification.

\end{enumerate}

\end{theorem}
Theorem~\ref{thm:oracle_adaptivity_asymptotic_optimality} is the central data-to-decision result of the paper and establishes the positive side of the theory.
When covariance geometries are sufficiently separated,
dependence learning consistently identifies the dominant geometry and
achieves oracle-equivalent inference with asymptotically vanishing
regret. Conversely, when the separation margin becomes small,
the oracle action itself becomes unstable because competing geometries
are locally indistinguishable. In this regime, hybrid or conservative
procedures provide a natural alternative, linking the present
decision-theoretic analysis with the impossibility result of
Section~\ref{sec:identification}.

The oracle is a benchmark for the value of learning, not an estimator available
to the researcher: at fixed separated data-generating processes, estimating the
dominant geometry has no first-order inferential cost relative to knowing it in
advance. The deliberately pointwise scope also clarifies why ambiguity
diagnostics and hybrid actions remain essential near ties.

\subsection{Computation and Profile-Weighted Inference}
\label{subsec:computational_complexity}
For fixed $N$ the entire pipeline is inexpensive: constructing
$\widehat\Gamma_T$ is $O(TN^2)$ and the dominant projection cost is a dense
$O(N^3)$ eigendecomposition, with normalization, margins, residuals, and
$\widehat\kappa$ negligible by comparison. A profile-weighted alternative that
combines geometry-specific variance estimators through the estimated profile,
$\widehat V^{\mathrm{avg}}=\sum_{d}\widehat\omega_d^{\mathrm{off}}\widehat V_d$,
is also available. Computational details, scalable extensions for large $N$, and
the profile-weighted estimator are developed in Online Appendix~B.


\section{Simulation Evidence}
\label{sec:simulation}

This section summarizes Monte Carlo evidence on three questions: whether the
estimated profile recovers the dominant geometry, how classification behaves
under hybrid and near-tie dependence, and whether profile-guided inference
tracks the infeasible oracle that knows the dominant geometry in advance. Data-
generating processes, calibrations, projection algorithms, and additional
robustness exercises are in the Online Appendix. The baseline designs include
pure cluster, pure factor, sparse network, cluster--factor, cluster--sparse,
factor--sparse, all-three hybrid, and two-way cluster dependence.

\subsection{Dependence Learning}
\label{subsec:simulation_learning}

Table~\ref{tab:simulation_profiles} reports, for each design, the population
\emph{full dependence profile} $\boldsymbol\omega$, the population
\emph{off-diagonal dependence profile} $\boldsymbol\omega^{\mathrm{off}}$, the
population procedure confidence index $\kappa_0$, the Monte Carlo means of both
estimated profiles, and the root mean squared error of each. The off-diagonal
profile---the object that drives classification---recovers the dominant geometry
sharply in every design (e.g.\ $\omega^{\mathrm{off}}_F=0.829$ under Factor,
$\omega^{\mathrm{off}}_S=0.936$ under Sparse), whereas the full profile is pulled
toward uniformity by the common diagonal and separates the geometries far less
(cf.\ Section~\ref{sec:estimation}). The confidence index $\kappa_0$ is largest
exactly where one off-diagonal geometry is clearly separated and fits well
($\kappa_0=0.90$ under Sparse, $0.69$ under Factor) and smallest under
near-uniform or hybrid dependence. The estimation columns also make a
finite-sample point visible: with $T<N$ the sample operator is noisy, so both
estimated profiles are attenuated toward uniformity under the cluster and sparse
designs, and this is precisely what the two RMSE columns quantify. Crucially, this
attenuation shrinks the \emph{magnitudes} of the off-diagonal weights but
preserves their \emph{ordering}: the dominant geometry retains the largest
estimated off-diagonal score in each design (e.g.\ under Sparse, the estimated
sparse weight $0.477$ still exceeds the factor and cluster weights), so the
argmax classification rule remains reliable even where the profile magnitudes are
biased---a point confirmed directly by the calibration evidence in
Figure~\ref{fig:kappa_calibration}.

\begin{table}[!ht]
\centering
\caption{Estimated Dependence Profiles}
\label{tab:simulation_profiles}
\resizebox{\textwidth}{!}{%
\begin{tabular}{lccccccccccccccc}\toprule
Design & $\omega_C$ & $\omega_F$ & $\omega_S$ & $\omega_C^{\mathrm{off}}$ & $\omega_F^{\mathrm{off}}$ & $\omega_S^{\mathrm{off}}$ & $\kappa_0$ & $\widehat\omega_C$ & $\widehat\omega_F$ & $\widehat\omega_S$ & $\widehat\omega_C^{\mathrm{off}}$ & $\widehat\omega_F^{\mathrm{off}}$ & $\widehat\omega_S^{\mathrm{off}}$ & $\mathrm{RMSE}(\widehat{\boldsymbol\omega})$ & $\mathrm{RMSE}(\widehat{\boldsymbol\omega}^{\mathrm{off}})$\\
\midrule
Cluster & 0.567 & 0.103 & 0.330 & 0.689 & 0.028 & 0.284 & 0.405 & 0.456 & 0.201 & 0.343 & 0.497 & 0.170 & 0.334 & 0.150 & 0.245\\
Factor & 0.037 & 0.792 & 0.171 & 0.028 & 0.829 & 0.143 & 0.687 & 0.037 & 0.791 & 0.172 & 0.028 & 0.830 & 0.142 & 0.008 & 0.010\\
Sparse & 0.332 & 0.332 & 0.335 & 0.030 & 0.034 & 0.936 & 0.895 & 0.282 & 0.318 & 0.401 & 0.183 & 0.341 & 0.477 & 0.084 & 0.574\\
Cluster--Factor & 0.067 & 0.753 & 0.180 & 0.053 & 0.802 & 0.145 & 0.547 & 0.067 & 0.748 & 0.185 & 0.053 & 0.798 & 0.149 & 0.022 & 0.019\\
Cluster--Sparse & 0.556 & 0.107 & 0.336 & 0.674 & 0.027 & 0.299 & 0.372 & 0.452 & 0.205 & 0.343 & 0.495 & 0.172 & 0.333 & 0.144 & 0.234\\
Factor--Sparse & 0.037 & 0.792 & 0.171 & 0.028 & 0.829 & 0.143 & 0.685 & 0.038 & 0.790 & 0.173 & 0.028 & 0.829 & 0.143 & 0.009 & 0.011\\
All Three & 0.067 & 0.752 & 0.180 & 0.053 & 0.802 & 0.146 & 0.547 & 0.067 & 0.747 & 0.186 & 0.053 & 0.798 & 0.149 & 0.022 & 0.019\\
Two-way Cluster & 0.664 & 0.165 & 0.171 & 0.797 & 0.106 & 0.097 & 0.692 & 0.529 & 0.254 & 0.217 & 0.611 & 0.250 & 0.139 & 0.170 & 0.242\\
\bottomrule
\end{tabular}
}
\begin{minipage}{0.92\textwidth}
\footnotesize
\emph{Notes:} Columns $\omega_C,\omega_F,\omega_S$ report the population
\emph{full} dependence profile $\boldsymbol\omega$ (relative similarity to the
cluster, factor, and sparse geometries); columns
$\omega_C^{\mathrm{off}},\omega_F^{\mathrm{off}},\omega_S^{\mathrm{off}}$ report
the population \emph{off-diagonal} dependence profile
$\boldsymbol\omega^{\mathrm{off}}$; and $\kappa_0$ is the population procedure
confidence index. Columns $\widehat\omega_\bullet$ and
$\widehat\omega_\bullet^{\mathrm{off}}$ are the Monte Carlo means of the
corresponding estimated profiles over $B$ replications. The last two columns
report two distinct root mean squared errors:
$\mathrm{RMSE}(\widehat{\boldsymbol\omega})$ is the Monte Carlo root mean squared
Euclidean error of the estimated \emph{full}-profile vector, and
$\mathrm{RMSE}(\widehat{\boldsymbol\omega}^{\mathrm{off}})$ is the corresponding
error for the \emph{off-diagonal}-profile vector; for
$q\in\{\mathrm{full},\mathrm{off}\}$,
\[
\mathrm{RMSE}(\widehat{\boldsymbol\omega}^{q})
=\Bigl[B^{-1}\sum_{b=1}^{B}
\bigl\|\widehat{\boldsymbol\omega}^{q,(b)}-\boldsymbol\omega^{q}\bigr\|_2^2
\Bigr]^{1/2},
\qquad
\widehat{\boldsymbol\omega}^{\mathrm{off},(b)}
=\bigl(\widehat\omega_C^{\mathrm{off},(b)},
\widehat\omega_F^{\mathrm{off},(b)},
\widehat\omega_S^{\mathrm{off},(b)}\bigr)'.
\]
Similarity weights measure relative geometric proximity and are not additive
variance shares. The data-generating processes, calibration, and number of
replications $B$ are described in Online Appendix~E.
\end{minipage}
\end{table}

Figure~\ref{fig:hybrid_profiles} shows the average estimated off-diagonal profile
along a cluster--factor hybrid path. The profile changes smoothly as the
relative strength of factor dependence increases, illustrating that the
method summarizes hybrid dependence rather than forcing a binary
classification.

\begin{figure}[!ht]
\centering
\includegraphicsIfExists[width=.72\textwidth]{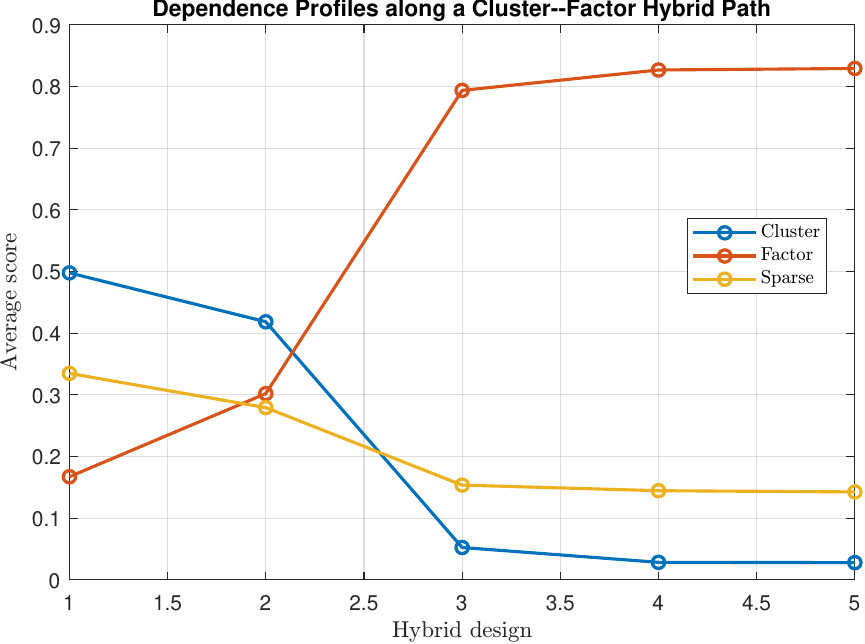}
\caption{Dependence Profiles along a Cluster--Factor Hybrid Path}
\label{fig:hybrid_profiles}
\begin{minipage}{0.86\textwidth}
\footnotesize
\emph{Notes:} The figure plots average estimated off-diagonal similarity weights along
a path moving from cluster-dominant to factor-dominant dependence.
\end{minipage}
\end{figure}

\subsection{Classification and Ambiguity}
\label{subsec:simulation_classification}

The classification theory predicts that dominant-geometry classification
is reliable when the separation margin is large and remains probabilistic
under local near-ties. Table~\ref{tab:near_ties} reports classification
frequencies for local cluster--factor alternatives, and
Figure~\ref{fig:classification_margin} plots the relationship between
classification error and the population separation margin.

\begin{table}[!ht]
\centering
\caption{Near-Ties and Ambiguous Dependence}
\label{tab:near_ties}
\begin{tabular}{lcccc}\toprule
$c$ & $P_B(\widehat d=C)$ & $P_B(\widehat d=F)$ & $P_B(\widehat d=S)$ & $\Delta_\omega$\\
\midrule
-4.0 & 0.000 & 1.000 & 0.000 & 0.673\\
-2.0 & 0.000 & 1.000 & 0.000 & 0.600\\
-1.0 & 0.003 & 0.997 & 0.000 & 0.395\\
0.0 & 0.173 & 0.827 & 0.000 & 0.026\\
1.0 & 0.937 & 0.063 & 0.000 & 0.199\\
2.0 & 1.000 & 0.000 & 0.000 & 0.250\\
4.0 & 1.000 & 0.000 & 0.000 & 0.294\\
\bottomrule
\end{tabular}

\begin{minipage}{0.88\textwidth}
\footnotesize
\emph{Notes:} The table reports Monte Carlo classification frequencies
under local cluster--factor alternatives. Near-ties generate
nondegenerate classification probabilities, as predicted by the local
alternatives theory.
\end{minipage}
\end{table}

\begin{figure}[!ht]
\centering
\includegraphicsIfExists[width=.72\textwidth]{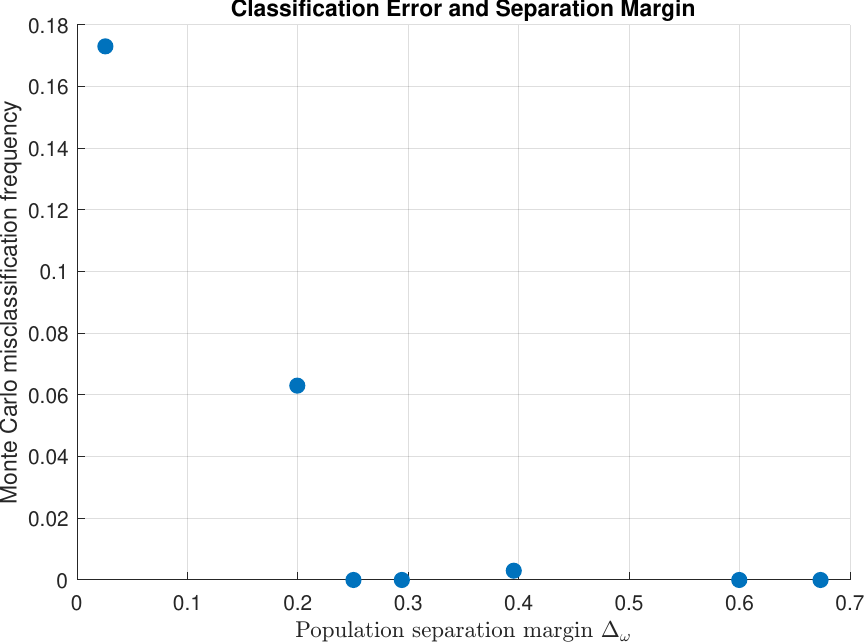}
\caption{Classification Error and Separation Margin}
\label{fig:classification_margin}
\begin{minipage}{0.86\textwidth}
\footnotesize
\emph{Notes:} The figure plots Monte Carlo misclassification frequency
against the population separation margin. Larger margins are associated
with lower classification error.
\end{minipage}
\end{figure}

A more demanding test is whether the procedure confidence index $\kappa_0$ is
\emph{calibrated}: does a larger $\kappa_0$ actually predict a higher probability
of recovering the dominant off-diagonal geometry? Figure~\ref{fig:kappa_calibration}
answers this along a cluster--factor path on which $\kappa_0$ sweeps from near
zero (at the crossover, where the two off-diagonal geometries tie and
Theorem~\ref{thm:fundamental_limit} predicts fundamental ambiguity) up to well-
separated values. The empirical classification probability
$\Pr(\widehat d = d^\star)$ tracks $\kappa_0$ closely: it collapses toward chance
in the near-tie region where $\kappa_0\approx0$ and rises to one once
$\kappa_0$ is bounded away from zero. The confidence index is therefore not
merely a diagnostic label but a calibrated predictor of classification
reliability---which is what licenses its use in the profile-guided decision rule.

\begin{figure}[!ht]
\centering
\includegraphicsIfExists[width=.72\textwidth]{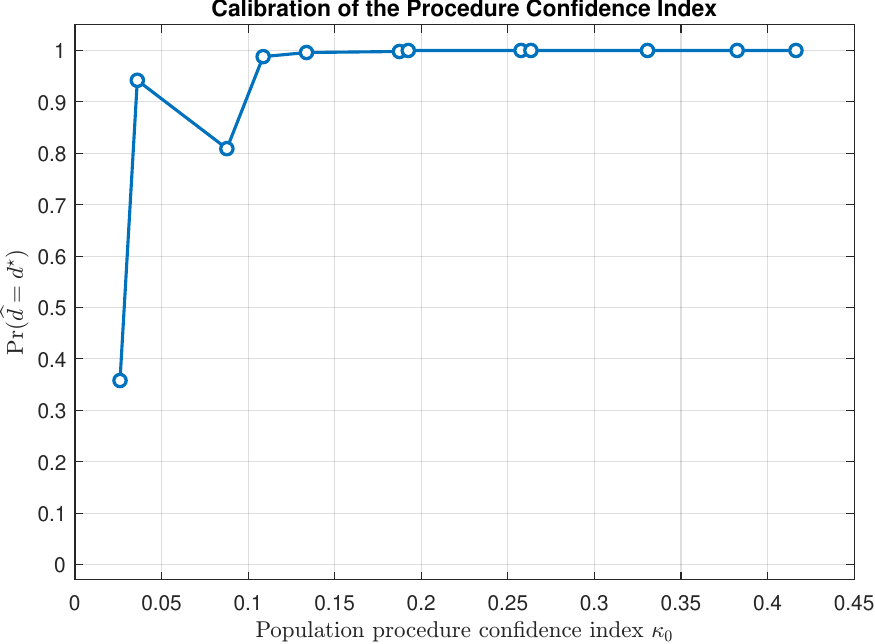}
\caption{Calibration of the Procedure Confidence Index}
\label{fig:kappa_calibration}
\begin{minipage}{0.86\textwidth}
\footnotesize
\emph{Notes:} Monte Carlo probability of correctly identifying the dominant
off-diagonal geometry, $\Pr(\widehat d=d^\star)$, plotted against the population
procedure confidence index $\kappa_0$ along a cluster--factor mixing path. Each
point is a design on the path; vertical bars are Monte Carlo standard errors.
Accuracy falls to near chance at the near-tie point ($\kappa_0\approx0$), where
the off-diagonal tangent spaces overlap, and rises to one as $\kappa_0$
increases, showing that $\kappa_0$ is a calibrated predictor of classification
reliability.
\end{minipage}
\end{figure}

\subsection{Oracle-Equivalent Inference}
\label{subsec:simulation_oracle}

The main econometric implication of Section~\ref{sec:profile_guided_inference}
is that profile-guided inference should be asymptotically equivalent to
an infeasible oracle procedure. Table~\ref{tab:oracle_equivalence}
reports finite-sample coverage probabilities for the oracle procedure,
the profile-guided procedure, and a deliberately misspecified procedure.

\begin{table}[!ht]
\centering
\caption{Oracle Equivalence of Profile-Guided Inference}
\label{tab:oracle_equivalence}
\begin{tabular}{lccc}\toprule
DGP & Oracle Coverage & Profile-Guided Coverage & Wrong Procedure\\
\midrule
Cluster & 93.3 & 93.3 & 20.9\\
Factor & 94.2 & 94.2 & 21.8\\
Sparse & 95.8 & 95.8 & 20.6\\
\bottomrule
\end{tabular}

\begin{minipage}{0.88\textwidth}
\footnotesize
\emph{Notes:} The table reports empirical coverage probabilities for
nominal 95\% confidence intervals. The oracle procedure uses the
variance estimator associated with the true dominant covariance geometry.
The profile-guided procedure uses the estimated dominant geometry. The
``Wrong Procedure'' column reports a deliberately misspecified benchmark.
\end{minipage}
\end{table}

The profile-guided procedure is nearly indistinguishable from the oracle across
the benchmark designs, while the misspecified procedure can exhibit severe
undercoverage. These results provide finite-sample support for
Theorem~\ref{thm:oracle_adaptivity_asymptotic_optimality} and the paper's main
practical message: learning the dependence profile can guide the choice of robust
inference procedure in a way that tracks the infeasible oracle benchmark.

\begin{figure}[!ht]
\centering
\includegraphicsIfExists[width=.72\textwidth]{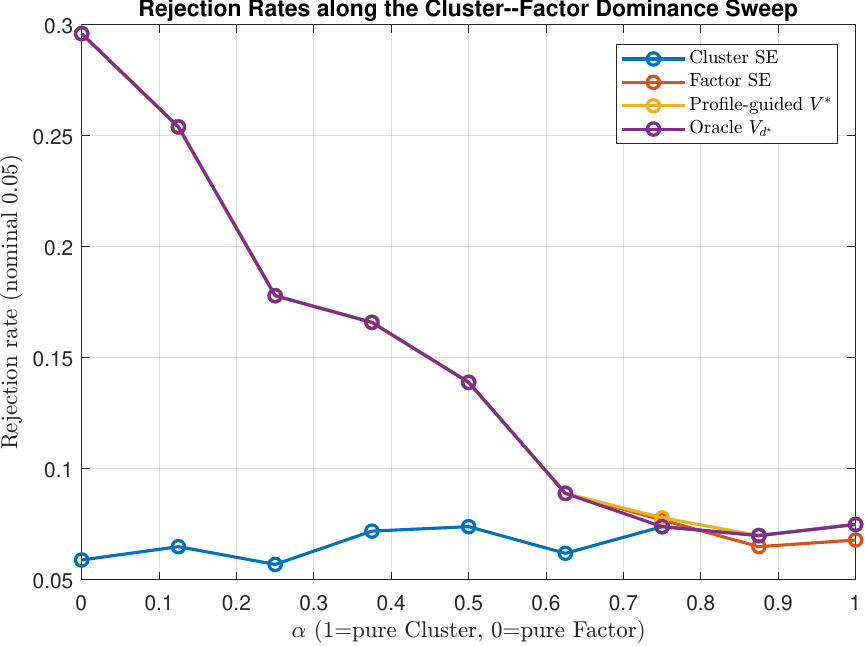}
\caption{Oracle Tracking along a Cluster--Factor Dominance Sweep}
\label{fig:oracle_tracking}
\begin{minipage}{0.86\textwidth}
\footnotesize
\emph{Notes:} The figure compares rejection rates for fixed procedures,
the profile-guided procedure, and the infeasible oracle along a
cluster--factor dominance sweep. The profile-guided and oracle procedures
closely coincide when classification is reliable.
\end{minipage}
\end{figure}

\section{Empirical Illustration}
\label{sec:empirical}
Throughout, we report the full profile $\widehat\omega$ and its off-diagonal
companion $\widehat\omega^{\mathrm{off}}$; the latter is the primary
classification and procedure-selection estimand, the former records overall
covariance fit. We present two complementary applications. The first, using
Fama--French industry portfolios, is a transparent case in which the dictionary
fits only weakly, so the diagnostics correctly report \emph{ambiguity}. The
second, using the Baltagi cigarette-demand panel, is \emph{consequential}: the
significance of an economically important coefficient depends on which
cross-sectional dependence structure is assumed for the standard errors, and the
profile identifies which structure the data exhibit. Together they show the
diagnostics behaving honestly when the dictionary fits poorly and decisively when
it fits well.

We use the Fama--French 49 industry portfolios and three factors from the data
library of \citet{FamaFrench1993} and \citet{FamaFrenchDataLibrary}. Industry
portfolio returns are well suited to the framework: they may exhibit cluster
dependence (industries group into sectors), factor dependence (returns load on
common market-wide shocks), and sparse dependence (some industries are linked
through input--output relations or supply chains more than others).

\subsection{Data and Regression Specification}
\label{subsec:empirical_data}

The analysis uses monthly excess returns on the 49 industry portfolios and the
market, SMB, and HML factors from the Kenneth R. French Data Library
\citep{FamaFrenchDataLibrary}. Let $R_{it}$ be the excess return on portfolio
$i$ in month $t$ ($i=1,\ldots,N$, $t=1,\ldots,T$, $N=49$). Ken French's raw files
code an industry-month as missing (via a numerical sentinel) whenever too few
firms are available; we drop every month in which any of the 49 industries is
flagged missing, restricting the sample to $T=682$ months (July 1969--April
2026). We estimate the factor model
\begin{equation}
R_{it}
=
\alpha_i
+
\beta_{iM} MKT_t
+
\beta_{iS} SMB_t
+
\beta_{iH} HML_t
+
u_{it},
\label{eq:empirical_factor_model}
\end{equation}
stack the residuals as $\widehat u_t=(\widehat u_{1t},\ldots,\widehat u_{Nt})'$,
and form the baseline residual covariance operator
\begin{equation}
\widehat\Gamma
=
\frac{1}{T}
\sum_{t=1}^{T}
\widehat u_t\widehat u_t'
\label{eq:empirical_operator}
\end{equation}
together with the residual correlation operator
$\widehat\Gamma^{COR}=\widehat D^{-1/2}\widehat\Gamma\widehat D^{-1/2}$,
$\widehat D=\operatorname{diag}(\widehat\Gamma)$. The covariance operator
preserves cross-industry volatility differences; the correlation operator focuses
on dependence geometry after normalizing scale.

\subsection{Dependence Classes and Projections}
\label{subsec:empirical_classes}

We project the empirical operator onto the cluster, factor, and sparse
geometries; the algorithms are those of Online Appendix~B. For the
\textbf{cluster} geometry we use a one-way specification based on the
Fama--French 10-sector classification, with membership $g(i)\in\{1,\ldots,G\}$
and support $(M_C)_{ij}=\ind\{g(i)=g(j)\}$. Because this is a one-way partition
and $\widehat\Gamma$ is PSD, the exact projection is the block mask
$\widehat P_C=M_C\odot\widehat\Gamma$ (each retained principal block is PSD; for
overlapping multiway supports the code instead solves the PSD-constrained
fixed-support projection). The \textbf{factor} projection
$\widehat P_F=L^*+D^*$ ($L^*$ rank-$r$ PSD, $D^*$ diagonal) is computed by the
alternating algorithm of Online Appendix~B, with $r$ selected by an
eigenvalue-ratio criterion on $\widehat\Gamma-\widehat D$ (robustness:
$r\in\{1,2,3\}$). The \textbf{sparse} projection jointly selects an off-diagonal
support of size at most $k_S$ and solves a PSD-constrained least-squares problem
on it; since exhaustive search is combinatorial, the code retains the $k_S$
largest off-diagonal entries of $\widehat\Gamma$ and projects onto the PSD cone
on that fixed support via Dykstra's algorithm---a feasible approximation to the
population metric projection, not a claim that hard thresholding alone solves the
global nonconvex problem.

\subsection{Estimated Dependence Profiles and Residual Diagnostics}
\label{subsec:empirical_profile}

We apply the estimated profile and projection-residual diagnostics of
Sections~\ref{subsec:similarity_score_estimation}
and~\ref{subsec:projection_residuals} to $\widehat\Gamma$. The profile describes
which geometry provides the strongest \emph{relative} fit; the residuals assess
\emph{absolute} fit, with a large $\widehat\rho_{\min}$ signaling that no
candidate geometry approximates the operator closely even if one receives the
largest relative score. In this application a large $\widehat\omega_C$ reflects
dependence concentrated within broad industry sectors, a large
$\widehat\omega_F$ low-rank common components not absorbed by the observed
factors, and a large $\widehat\omega_S$ dependence concentrated among a small
subset of industry pairs.

\subsection{Empirical Results}
\label{subsec:empirical_results}

Table~\ref{tab:empirical_profile} reports the estimated dependence
profiles and projection-residual diagnostics for the residual covariance
and correlation operators.

\begin{table}[!ht]
\centering
\caption{Dependence Profiles and Projection-Residual Diagnostics for Industry Portfolio Residuals}
\label{tab:empirical_profile}
\begin{tabular}{lcccccccc}\toprule
Operator & $\widehat\omega_C$ & $\widehat\omega_F$ & $\widehat\omega_S$ & $\widehat\rho_C$ & $\widehat\rho_F$ & $\widehat\rho_S$ & $\widehat\rho_{\min}$ & $r$\\
\midrule
Covariance & 0.321 & 0.324 & 0.355 & 0.465 & 0.457 & 0.366 & 0.366 & 1\\
Correlation & 0.293 & 0.364 & 0.342 & 0.619 & 0.485 & 0.529 & 0.485 & 1\\
\bottomrule
\end{tabular}

\begin{minipage}{0.92\textwidth}
\footnotesize
\emph{Notes:}
\vspace{0.4em}
The table reports estimated similarity scores
\((\widehat\omega_C,\widehat\omega_F,\widehat\omega_S)\)
and projection-residual diagnostics
\((\widehat\rho_C,\widehat\rho_F,\widehat\rho_S,\widehat\rho_{\min})\)
using residuals from the industry-level Fama--French regression in
\eqref{eq:empirical_factor_model}. The covariance operator preserves
scale differences across industries, while the correlation operator
normalizes industry-specific residual volatility. The similarity scores
are relative geometric proximity measures and should not be interpreted
as additive variance shares. The residual diagnostics measure absolute
distance from the empirical operator to each covariance geometry.
\end{minipage}
\end{table}

Table~\ref{tab:empirical_profile} reveals a hybrid dependence structure
in which no single covariance geometry dominates and none provides a
close absolute fit.

\textit{Covariance and correlation operators.}
The three similarity scores are close, $\widehat\omega_C=0.321$,
$\widehat\omega_F=0.324$, $\widehat\omega_S=0.355$, with the sparse score
marginally largest and a substantial minimum residual $\widehat\rho_{\min}=0.366$
(attained by the sparse projection). Residual co-movement among the 49
industries---after removing the three Fama--French factors---is thus not well
summarized by any single geometry: sector clustering, common factor structure,
and localized sparse dependence all receive comparable, individually limited
support, and the large $\widehat\rho_{\min}$ signals that the dictionary itself
is, in absolute terms, incomplete for this dataset. Normalizing by industry
volatility shifts the off-diagonal profile mildly toward factor dominance
($\widehat\omega_F^{\mathrm{off}}=0.364$, with sparse $0.342$ and cluster $0.293$
close behind) and raises the minimum residual to $\widehat\rho_{\min}=0.485$.

\textit{Standard errors and recommendation.}
Sector-clustered and two-way clustered standard errors are roughly
$5.6$--$5.8\times$ the White standard error for the market factor
($0.043$--$0.045$ vs.\ $0.008$), while the common-shock adjustment is smaller
($0.006$); the wide spread across procedures (reported in full in Online
Appendix~F) is itself a symptom of the hybrid structure.
Table~\ref{tab:empirical_inference} reports the confidence index and
recommendation for each operator.

\begin{table}[!ht]
\centering
\caption{Profile-Guided Inference Recommendation for Industry Portfolio Residuals}
\label{tab:empirical_inference}
\resizebox{\textwidth}{!}{%
\begin{tabular}{lccccc}\toprule
Operator & $\widehat\Delta_\omega$ & $\widehat\rho_{\min}$ & $\widehat\kappa$ & Dominant geometry & Profile-guided recommendation\\
\midrule
Covariance & 0.031 & 0.366 & 0.019 & Sparse & Hybrid: report multiple SEs (low $\widehat\kappa$)\\
Correlation & 0.022 & 0.485 & 0.011 & Factor & Hybrid: report multiple SEs (low $\widehat\kappa$)\\
\bottomrule
\end{tabular}

}
\begin{minipage}{0.99\textwidth}
\footnotesize
\vspace{0.4em}
\emph{Notes:}
The table reports the separation margin
$\widehat\Delta_{\omega^{\mathrm{off}}}=\max_d\widehat\omega_d^{\mathrm{off}}-\max_{d\neq\widehat d}\widehat\omega_d^{\mathrm{off}}$,
the minimum projection residual $\widehat\rho_{\min}$, the procedure
confidence index $\widehat\kappa=(1-\widehat\rho_{\min})\widehat\Delta_{\omega^{\mathrm{off}}}$,
the estimated dominant geometry, and the implied profile-guided
recommendation. A small $\widehat\kappa$ indicates either near-tie
separation or poor dictionary fit; the recommended response is to report
inference from multiple procedures.
\end{minipage}
\end{table}

Both operators yield small confidence indices: $\widehat\kappa=0.019$ for the
covariance operator ($\widehat\Delta_{\omega^{\mathrm{off}}}=0.031$,
$\widehat\rho_{\min}=0.366$) and $\widehat\kappa=0.011$ for the correlation
operator ($\widehat\Delta_{\omega^{\mathrm{off}}}=0.022$,
$\widehat\rho_{\min}=0.485$). These low values correctly indicate a genuinely
hybrid structure that the dictionary describes only partially, so no single
procedure can be recommended with high confidence, and the framework prescribes
reporting several procedures and documenting sensitivity---exactly the spread
documented in Online Appendix~F.

This example illustrates a key advantage of the proposed framework
over a priori procedure selection.
A researcher who assumes cluster dependence and uses sector-clustered
standard errors is implicitly claiming $\widehat\kappa\approx 1$ for
the cluster geometry---a claim the data refute, since the cluster score
is not even the largest of the three under either operator.
A researcher who learns the profile first discovers the hybrid
structure, obtains $\widehat\kappa$ close to zero, and is correctly
directed to report multiple procedures rather than to commit to one.
The framework thus transforms an informal sensitivity analysis into a
statistically grounded diagnostic.

\subsection{A Consequential Application: Cigarette Demand and the
Cross-Border Effect}
\label{subsec:empirical_cigar}

The Fama--French illustration deliberately reports ambiguity. We now turn
to a case in which the dependence profile changes an inferential
conclusion. We use the U.S. cigarette demand panel of \citet{Baltagi2006}, a
public dataset distributed with that text \citep{BaltagiCigarData}:
annual observations on $N=46$ states
over $T=30$ years (1963--1992), a standard testing ground for
cross-sectional dependence because cigarette markets are linked across
states by bootlegging and cross-border shopping. The data are public and
bundled with the replication package.

\paragraph{Regression and the contested coefficient.}
We estimate a conventional two-way (state and year) fixed-effects demand
equation,
\begin{equation}
\log \mathrm{sales}_{it}
=
\beta_1 \log p^{\mathrm{real}}_{it}
+
\beta_2 \log y^{\mathrm{real}}_{it}
+
\beta_3 \log p^{\mathrm{near}}_{it}
+
\alpha_i + \tau_t + u_{it},
\label{eq:cigar_regression}
\end{equation}
where $\mathrm{sales}_{it}$ is per-capita sales, $p^{\mathrm{real}}$ is
the real price, $y^{\mathrm{real}}$ is real per-capita income, and
$p^{\mathrm{near}}$ is the minimum real price in neighboring states. The
coefficient of interest is $\beta_3$: a negative value indicates that
cheaper cigarettes nearby draw sales away from a state---the cross-border
(bootlegging) effect that motivates much of the empirical
cigarette-tax literature. The price and income elasticities are estimated
sharply and are significant under every standard-error method
(Table~\ref{tab:cigar_se}); the cross-border effect is the coefficient
whose significance is in question.

\paragraph{The inferential problem.}
Table~\ref{tab:cigar_se} reports $\hat\beta_3$ with four standard errors.
The estimate is $\hat\beta_3=-0.117$ under all of them---only the
standard error changes---yet the conclusion does not survive that change.
Under homoskedastic and cluster-by-year standard errors the effect is
significant at the five percent level ($p=0.030$ and $p=0.041$); under
HC1 and the conventional cluster-by-state standard errors it is not
($p=0.064$ and $p=0.170$). A researcher who follows the common default of
clustering by state concludes that cross-border effects are absent; one
who allows for a common temporal shock concludes the opposite. Nothing in
the regression output adjudicates between them.

\begin{table}[htbp]
\centering
\caption{Cigarette demand: $\hat\beta$ and four standard errors. The
cross-border coefficient (log nearby min price) is significant at five
percent under homoskedastic and cluster-by-year standard errors, and
insignificant under HC1 and cluster-by-state.}
\label{tab:cigar_se}
\small
\begin{tabular}{lcccccc}\toprule
Coefficient & SE method & $\hat\beta$ & SE & $t$ & $p$ & \\
\midrule
log real price & Homoskedastic & -1.023 & 0.042 & -24.46 & 0.000 & $^{*}$\\
log real price & HC1 & -1.023 & 0.061 & -16.66 & 0.000 & $^{*}$\\
log real price & Cluster-state & -1.023 & 0.224 & -4.57 & 0.000 & $^{*}$\\
log real price & Cluster-year & -1.023 & 0.061 & -16.79 & 0.000 & $^{*}$\\
log real income & Homoskedastic & 0.520 & 0.047 & 11.14 & 0.000 & $^{*}$\\
log real income & HC1 & 0.520 & 0.059 & 8.87 & 0.000 & $^{*}$\\
log real income & Cluster-state & 0.520 & 0.165 & 3.16 & 0.002 & $^{*}$\\
log real income & Cluster-year & 0.520 & 0.082 & 6.34 & 0.000 & $^{*}$\\
log nearby min price & Homoskedastic & -0.117 & 0.054 & -2.17 & 0.030 & $^{*}$\\
log nearby min price & HC1 & -0.117 & 0.063 & -1.85 & 0.064 & \\
log nearby min price & Cluster-state & -0.117 & 0.085 & -1.37 & 0.170 & \\
log nearby min price & Cluster-year & -0.117 & 0.057 & -2.04 & 0.041 & $^{*}$\\
\bottomrule
\end{tabular}

\end{table}

\paragraph{What the dependence profile says.}
We estimate the profile on the two-way within residuals, using exactly the
operator of the rest of the paper: $\widehat\Gamma_T=T^{-1}\sum_t\widehat
u_t\widehat u_t'$ projected onto the cluster (four census regions), factor
($r=1$), and sparse ($10\%$ of off-diagonal pairs) geometries, estimated without
reference to which standard error is ``interesting.'' Table~\ref{tab:cigar_profile}
reports the result: the estimated \emph{full dependence profile} is
factor-dominant, $\widehat{\boldsymbol\omega}=(0.16,\,0.46,\,0.38)$, and the
\emph{off-diagonal dependence profile}---the object that drives procedure
selection---points the same way, since the factor geometry fits the residual
co-movement far better than the alternatives ($\widehat\rho_F=0.15$, well below
the cluster residual $\widehat\rho_C=0.81$). The confidence index is
$\widehat\kappa=0.07$---modest but, unlike the Fama--French case, backed by a
genuinely small $\widehat\rho_{\min}$: here the dictionary fits.

\begin{table}[htbp]
\centering
\caption{Cigarette demand: dependence profile and diagnostics for the
residual operator. Factor dependence dominates and fits well
($\widehat\rho_F=0.15$); cluster dependence fits poorly
($\widehat\rho_C=0.81$).}
\label{tab:cigar_profile}
\small
\begin{tabular}{lc}\toprule
Quantity & Value\\
\midrule
$\omega_C$ (cluster) & 0.161\\
$\omega_F$ (factor) & 0.461\\
$\omega_S$ (sparse) & 0.378\\
$\rho_C$ & 0.812\\
$\rho_F$ & 0.153\\
$\rho_S$ & 0.448\\
$\hat\rho_{\min}$ & 0.153\\
$\hat\Delta_\omega$ & 0.084\\
$\hat\kappa$ & 0.071\\
\bottomrule
\end{tabular}

\end{table}

\paragraph{Interpretation.}
The profile indicates that what remains of the cross-sectional dependence after
two-way demeaning is factor-shaped rather than state-clustered, and this is
economically sensible. The factor cannot be a common additive year shock---
additive year effects $\tau_t$ are already absorbed by the within
transformation---but instead reflects \emph{heterogeneous exposure} to
nationally common forces: cigarette demand across states responds with differing
sensitivity to national anti-smoking campaigns, federal tax and regulatory
changes, nationally marketed advertising, and secular shifts in health attitudes.
States with similar demographics, smoking cultures, or regulatory environments
load similarly on these common movements, producing a low-rank co-movement in
residual demand that a single national dummy per year cannot absorb and that
within-state clustering is not designed to capture. This is exactly the
structure the factor geometry detects, which is why it both dominates the profile
and fits well ($\widehat\rho_F=0.15$) while the state-cluster geometry fits worst
($\widehat\rho_C=0.81$). Clustering by state, the conventional default, targets
within-state serial correlation that the profile finds both subdominant and
poorly fitting, and it is the method under which the cross-border effect appears
insignificant; standard errors that accommodate the common temporal factor keep
it significant. The profile thus provides evidence favoring the latter and
identifies the state-clustered standard error as the one whose maintained
assumption the data support least. Because $\widehat\kappa=0.07$ is modest, we
report all four procedures and treat the reading as a principled sensitivity
analysis rather than a definitive resolution.

\paragraph{From classical to adaptive inference.}
It is worth stating plainly what the framework changes here. \emph{Classical
practice} fixes a dependence assumption before looking at the residual
dependence---most commonly clustering by state---and reads off the standard
error it implies; under that default the cross-border coefficient is
insignificant ($p=0.170$), and the researcher concludes that bootlegging leaves
no detectable footprint. \emph{Adaptive inference} instead learns the dependence
geometry from the same residuals: the off-diagonal profile points to a factor
structure that fits well and away from the state-cluster structure that fits
poorly, directing inference to a common-shock--robust standard error under which
the same point estimate ($\widehat\beta_3=-0.117$) is significant ($p=0.030$).
The two routes reach \emph{opposite substantive conclusions} about a
policy-relevant elasticity from identical data and an identical coefficient; what
differs is only the maintained dependence assumption, and the classical route
selects that assumption by convention while the adaptive route selects it by
evidence. This is the concrete answer to ``why learn the geometry?''---not a
tighter standard error for its own sake, but protection against a default
dependence assumption that the data actively contradict.

\paragraph{Transparency.}
Two caveats are reported openly. First, because the operator is formed from
two-way within residuals with $T=30<N=46$, $\widehat\Gamma_T$ is rank-deficient
(rank $29$); this does not affect the profile, which is a projection functional
requiring no inverse of $\widehat\Gamma_T$. Second, the object profiled is the
dependence operator of the \emph{within-transformed} residuals: with $N$ fixed,
year demeaning applies the fixed linear map $M_N=I_N-N^{-1}\mathbf 1\mathbf 1'$
to each cross-section, so $\widehat\Gamma_T$ estimates $M_N\Sigma M_N$ rather
than $\Sigma$. This is the correct target, not a defect---the coefficient
estimator and all four standard errors are computed from the same
within-transformed data---and Proposition~\ref{prop:consistency_operator} applies
once $u_t$ is read as the within-transformed disturbance, via the generic
residual condition of Remark~\ref{rem:generic_residuals}.



\section{Conclusion}
\label{sec:conclusion}
This paper studies a prior problem usually resolved before formal inference
begins: which dependence structure should guide the inferential procedure?
Rather than treating that structure as known, we represent candidate mechanisms
as covariance geometries, estimate their relative empirical support through
projection, and use the resulting profile to guide procedure choice. The central
contribution is not another covariance estimator or robust variance formula but a
unified framework for learning dependence that is relevant for inference.

The framework connects four steps usually treated separately: an estimable
dependence operator is projected onto a dictionary of covariance geometries; the
full and off-diagonal profiles summarize overall covariance fit and
cross-sectional dependence; local projection regularity and geometric separation
determine whether a dominant geometry can be learned, while tangent-space overlap
characterizes first-order ambiguity; and a decision rule maps the off-diagonal
profile into an inference procedure. Under a uniquely separated dominant geometry
and the profile--loss compatibility condition, this rule is asymptotically
equivalent to an infeasible oracle and has vanishing regret. The profile is a
low-dimensional diagnostic relative to a chosen dictionary, not a structural
identification device: a dominant component indicates which geometry receives the
strongest support for inference, a near tie signals intrinsic ambiguity when
tangent directions overlap, and a large minimum residual signals that the
dictionary itself is inadequate. The framework thus produces not only a
recommendation but also diagnostics governing how much confidence to place in it.

The broader implication is a change in the order of econometric reasoning:
instead of asking only which robust covariance estimator to use, the researcher
first asks what dependence structure the operator supports, how sharply it is
separated from alternatives, and what that implies for inference---covariance
geometry supplying the representation, the profile the diagnostic, geometric
separation the learnability criterion, and decision theory the inferential
action. The analysis deliberately fixes the cross-sectional dimension and a finite
dictionary, isolating the learning and decision problems; important extensions
include joint $(N,T)$ asymptotics, expanding or data-driven dictionaries, spatial
and network geometries, nonlinear operators, and decision rules that combine
rather than select procedures near ties.

{\small
\begin{center}
\begin{tabular}{p{0.24\textwidth}p{0.66\textwidth}}
\toprule
Diagnostic pattern & Recommended reporting strategy \\
\midrule
Large margin, small minimum residual & Report the matched procedure as the
primary specification and document the profile diagnostics. \\
Small margin, adequate dictionary fit & Treat the result as genuine ambiguity;
report multiple geometry-specific procedures or a pre-specified hybrid rule. \\
Large minimum residual & Regard the candidate dictionary as inadequate and
avoid a geometry-specific recommendation without expanding the dictionary. \\
Full/off-diagonal profiles disagree & Explain whether diagonal heterogeneity or
cross-sectional dependence drives the difference and base procedure choice on
the off-diagonal profile. \\
\bottomrule
\end{tabular}
\end{center}
}

\subsection*{Practical Guidance for Applied Researchers} On constructing the dictionary: include a cluster geometry when units plausibly share group-level shocks (industry, region, cohort); a factor geometry when a few pervasive drivers move most units together; and a sparse geometry when dependence is localized or network-based. When in doubt, include a candidate and let the profile and residual adjudicate---an over-inclusive dictionary costs precision, not validity. On using the output: match the operator to the inferential target; treat a nearly uniform profile with small $\widehat\kappa$ as genuine ambiguity and report several procedures; prefer the hybrid rule when $\widehat\Delta_{\omega^{\mathrm{off}}}$ is within sampling error of zero; read a large $\widehat\rho_{\min}$ as dictionary inadequacy; and when one geometry clearly dominates with small residual, the matched procedure inherits the oracle guarantees of Section~\ref{sec:profile_guided_inference}. The profile is invariant to positive rescalings of the operator but not to its definition, which should be reported; the projection residuals quantify the information the summary discards, and dictionary adequacy is testable through $\widehat\rho_{\min}$.

\begingroup
\small
\setstretch{1.0}
\setlength{\bibsep}{0.0pt}
\bibliographystyle{ecta_fallback_plainnat}
\bibliography{references}
\endgroup
\newpage

\paragraph{\bf Supplement to ``Learning Dependence Structures for Econometric Inference:
Identification, Ambiguity, and Adaptive Inference''}
%
%
%

%
%

\maketitle

\begin{abstract}
\noindent
This online appendix collects supplementary material for ``Learning Dependence Structures for Econometric Inference:
Identification, Ambiguity, and Adaptive Inference.'' Appendix~A develops
additional geometric results supporting identification: regular points
and tangent spaces of the cluster, factor, and sparse covariance
classes, the principal-angle separation condition, and examples of
geometric overlap that produce ambiguous dependence structures.
Appendix~B describes the computational implementation of the
projection-based estimators and additional choices of empirical
dependence operator. Appendix~C contains proofs of the main paper's
identification, asymptotic normality, classification, and
oracle-adaptivity results. Appendix~D collects auxiliary technical
results used in those proofs, including the two-geometry
local-indistinguishability result and projection regularity lemmas.
Appendix~E reports the simulation design underlying the main paper's
Monte Carlo evidence---data-generating processes, projection
algorithms, and calibration choices---together with three robustness
exercises not in the main text: a direct check of the principal-angle
condition at the simulation parameters, a test of profile and residual
behavior under misspecified dependence operators, and the mechanics
behind the procedure-recommendation result, including a root-cause
diagnosis and fix for a finite-sample distortion in the factor-robust
variance estimator. Appendix~F reports additional empirical detail for the
Fama--French illustration, including the pooled-regression robust
standard-error benchmark.
\end{abstract}

\medskip


\newpage

\appendix

\setcounter{equation}{0} \setcounter{assumption}{0}
\setcounter{figure}{0} \setcounter{table}{0} \setcounter{remark}{0}

 \renewcommand{%
\theequation}{A.\arabic{equation}}
\renewcommand{\thelemma}{A.\arabic{lemma}}
\renewcommand{\theassumption}{A.\arabic{assumption}} \renewcommand{%
\thetheorem}{A.\arabic{theorem}}

\renewcommand{\thetable}{A.\arabic{table}}
\renewcommand{\thefigure}{A.\arabic{figure}}

\renewcommand{\thesubsection}{A.\arabic{subsection}}
\renewcommand{\theremark}{A.\arabic{remark}}
\section{Additional Results for Geometric Identification}
\label{sec:identification_app}

\subsection{Regular Points and Tangent Spaces}
\label{subsec:regular_tangent}
 
We introduce explicit regularity conditions for each geometry and derive the
associated tangent spaces.
 
\paragraph{Cluster geometry.}
Let $\mathcal I_C$ be the cluster support set defined in
Section~\ref{sec:geometry}.
A point $\Sigma_C\in\mathcal S_C$ is \emph{regular} if it lies in
the relative interior of the fixed-support PSD cone; a convenient sufficient
condition is that $\Sigma_C$ be positive definite on the coordinate subspace
induced by the cluster blocks.  At such a point the PSD inequality is locally
inactive relative to the support subspace.  The tangent space is then the
closed linear subspace
\[
T_C(\Sigma_C)
=
\bigl\{H\in\mathbb H:H_{ab}=0\text{ whenever }(a,b)\notin\mathcal I_C\bigr\}.
\]
Tangent directions may perturb any covariance entry on the cluster support
but cannot introduce dependence between observations sharing no cluster
membership.
 
\paragraph{Factor geometry.}
Write $\Sigma_F=L_0+D_0$ with $L_0=\Lambda_0\Lambda_0'$,
$\operatorname{rank}(L_0)=r$, and $D_0$ diagonal.
A point $\Sigma_F$ is \emph{regular} if the $r$ leading eigenvalues of
$L_0$ are positive and distinct and the diagonal component satisfies
$D_{0,ii}>0$ for every $i$.
Under this generic condition the rank-$r$ manifold is smooth at $L_0$
\citep{LewisMalick2008}, and the tangent space of $\mathcal S_F(r)$ at
$\Sigma_F$ is
\[
T_F(\Sigma_F)
=
\bigl\{\Lambda_0 H'+H\Lambda_0'+\operatorname{diag}(v):
H\in\mathbb R^{N\times r},\;v\in\mathbb R^N\bigr\}.
\]
The first two terms span the tangent of the rank-$r$ manifold at $L_0$; the
diagonal term $\operatorname{diag}(d)$ accounts for free perturbations of the
noise floor $D_0$.
Note that $\mathcal S_F(r)$ is a cone: if $\Gamma=L+D\in\mathcal S_F(r)$ then
$a\Gamma=aL+aD\in\mathcal S_F(r)$ for all $a>0$, since $\operatorname{rank}(aL)=\operatorname{rank}(L)\leq r$
and $aD$ is diagonal.
 
\paragraph{Sparse geometry.}
Let $I_0$ be the unique active support selected by
Assumption~\ref{ass:sparse_unique}, and let
$\Sigma_S=P_S(\Gamma_0)\in\mathcal C_{I_0}$.  A sparse projected point is
\emph{regular} if it lies in the relative interior of the fixed-support PSD
cone $\mathcal C_{I_0}$; a convenient sufficient condition is positive
definiteness on the active coordinate subspace.  The objective gap in
Assumption~\ref{ass:sparse_unique} stabilizes the support, while relative
interiority makes the PSD inequality locally inactive on that support.  The
tangent space at such a regular sparse point is
\[
T_S(\Sigma_S)
=
\bigl\{H\in\mathbb H:H_{ij}=0\text{ whenever }(i,j)\notin I_0\bigr\}.
\]

\paragraph{The diagonal is a common nuisance direction.}
\label{par:diagonal_nuisance}
By construction, $(i,i)\in\mathcal I_C$ for every $i$ (every observation
trivially shares a cluster membership with itself), so
$\operatorname{diag}(v)\in T_C(\Sigma_C)$ for every $v\in\mathbb R^N$:
the cluster tangent space contains \emph{all} diagonal perturbations.
The factor tangent space $T_F(\Sigma_F)$ contains the same set of
diagonal perturbations directly, via its free $\operatorname{diag}(v)$
term, which represents the noise floor $D_0$. Consequently,
\[
\{\operatorname{diag}(v):v\in\mathbb R^N\}
\subseteq
T_C(\Sigma_C)\cap T_F(\Sigma_F)
\]
at \emph{every} pair of regular cluster and factor points, regardless of
parameter values: the cluster and factor tangent spaces always share the
full diagonal subspace. This is a substantive feature of the geometries,
not an estimation artifact, and it reflects the fact that all three
covariance classes leave each observation's own variance ($\Sigma_{ii}$)
unrestricted; the diagonal is a common nuisance direction present in
every geometry under consideration; analogously, the sparse projection of Section \ref{subsec:computation} always retains the diagonal regardless
of the sparsity budget $k_S$. Because the diagonal carries no information
about \emph{dependence} between distinct observations---which is the
object of interest throughout this paper---we factor it out of the
principal-angle condition below by restricting attention to the
off-diagonal part of each tangent space,
\[
T_d^{\mathrm{off}}(\Sigma_d)
=
\bigl\{H\in T_d(\Sigma_d):H_{ii}=0\text{ for all }i\bigr\},
\qquad d\in\mathfrak D.
\]
Equivalently, $T_d^{\mathrm{off}}(\Sigma_d)$ is the intersection of
$T_d(\Sigma_d)$ with the off-diagonal hyperplane
$\{H\in\mathbb H:\operatorname{diag}(H)=0\}$; for the cluster and sparse
geometries this simply removes the diagonal entries from the coordinate
subspace $T_d(\Sigma_d)$, and for the factor geometry it removes the
$\operatorname{diag}(d)$ term entirely, since
$\Lambda_0H'+H\Lambda_0'+\operatorname{diag}(v)$ has zero off-diagonal
contribution from $\operatorname{diag}(v)$.

\subsection{Principal-Angle Condition}

\begin{remark}[Why the Off-Diagonal Restriction Is Necessary]
\label{rem:offdiag_necessary}
Assumption~\ref{ass:principalangle} is stated in terms of the
off-diagonal tangent spaces $T_d^{\mathrm{off}}$ introduced in
Section~\ref{subsec:regular_tangent}, rather than the full tangent spaces
$T_d$. This restriction is not a matter of convenience: as shown there,
$T_C(\Sigma_C)\cap T_F(\Sigma_F)$ always contains the entire diagonal
subspace $\{\operatorname{diag}(v):v\in\mathbb R^N\}$, for every pair of
regular cluster and factor points and every parameter configuration, so
$\theta(T_C(\Sigma_C),T_F(\Sigma_F))=0$ identically and
Assumption~\ref{ass:principalangle} stated with the unrestricted tangent
spaces could never hold for the cluster--factor pair. This reflects the
fact that own-variance perturbations are common to every covariance
geometry considered in this paper and carry no information distinguishing
one dependence geometry from another. Restricting to the off-diagonal
directions $T_d^{\mathrm{off}}$ isolates exactly the perturbations that
are informative about \emph{dependence}, which is the relevant object of
identification.
\end{remark}

The next lemma isolates the role of within-class regularity: at a regular
point of a given geometry, the projection onto \emph{that} geometry alone
is locally single-valued. This conclusion uses only the regularity
conditions of Section~\ref{subsec:regular_tangent} and does not yet involve
the relationship between different geometries.

\begin{lemma}[Local Uniqueness of Individual Projections]
\label{lem:individual_uniqueness}
Suppose $\Sigma_C$ and $\Sigma_S$ are regular projected points in the
sense of Section~\ref{subsec:regular_tangent}, Assumption~\ref{ass:sparse_unique}
holds, and the factor projection $P_F$ is locally single-valued at the regular
factor point $\Sigma_F$.  Then the fixed-support cluster and sparse
projections are locally single-valued, and $P_F$ is locally single-valued by
assumption.  Local Lipschitz continuity is imposed in
Assumption~\ref{ass:local_projection_regularity}.
\end{lemma}

Local uniqueness of each projection is necessary for the dependence profile
to be well defined near $\Sigma_0$, but it is not sufficient for the
profile to carry separate information about cluster, factor, and sparse
geometry: if two tangent spaces shared a common direction, a perturbation
along that direction would move \emph{both} projections in a coupled way,
and the resulting change in the profile could not be uniquely attributed to
either geometry. The principal-angle condition rules this out by ensuring
that the tangent spaces of distinct geometries intersect only at the
origin.

\paragraph{Diagonal perturbations are never separately identified}
\label{rem:diagonal_indistinguishability}
By Section~\ref{subsec:regular_tangent}, every diagonal perturbation
$\operatorname{diag}(d)$ lies in $T_C(\Sigma_C)\cap T_F(\Sigma_F)$ (and,
since the sparse projection of Section~\ref{subsec:computation} always
retains the diagonal, in $T_S(\Sigma_S)$ as well), regardless of the
principal-angle condition. Consequently, a perturbation of $\Gamma_0$
confined to its own diagonal---that is, a change in observation-specific
variances with no change in any cross-covariance---cannot be attributed
to cluster, factor, or sparse dependence: it is consistent with all three
geometries simultaneously and with none of them specifically. This is not
a defect of Assumption~\ref{ass:principalangle} or of
Theorem~\ref{thm:identification}; it reflects the fact that the dependence
profile is, by design, a summary of off-diagonal dependence structure.
Variance heterogeneity across observations is a separate object, already
well studied under the heading of heteroskedasticity, and is intentionally
outside the scope of dependence learning as formulated here.

\subsection{Ambiguous Dependence Structures}
\label{subsec:ambiguous_dependence}
The principal-angle condition ensures local separation between
covariance geometries along off-diagonal directions. When this separation
breaks down and off-diagonal tangent spaces overlap, ambiguity may arise
as a fundamental feature of the problem rather than a consequence of
limited sample size. The following result formalizes this local
indistinguishability phenomenon.

\paragraph{Identification versus Ambiguity}
Theorems~\ref{thm:identification} and
\ref{thm:fundamental_limit}
characterize the boundary between identifiable and ambiguous
dependence structures. Positive principal angles ensure local
separation of covariance geometries and hence local
identifiability of the dependence profile. Conversely,
overlapping tangent spaces generate local indistinguishability,
so that ambiguity reflects a fundamental lack of identifying
information rather than finite-sample uncertainty. Thus,
identification and indistinguishability represent two sides of
the same geometric phenomenon.

\subsection{Examples of Geometric Overlap}
\label{subsec:overlap_examples}
 
Three further situations produce overlap between covariance geometries
along off-diagonal directions, in addition to the unconditional diagonal
overlap of Remark~\ref{rem:diagonal_indistinguishability}.
 
\textit{Cluster--sparse overlap.}
A one-way cluster covariance matrix with $G$ balanced groups of size $N/G$
has a block-support of size $G\cdot(N/G)^2$.
When $G\cdot(N/G)^2=O(k_S)$, the cluster support is no larger than the
sparse budget $k_S$, so the cluster geometry lies inside the sparse class.
In such cases $T_C(\Sigma_C)\subseteq T_S(\Sigma_C)$ and the principal
angle $\theta(T_C,T_S)$ can be arbitrarily small.
 
\textit{Two-way cluster--sparse overlap.}
Consider observations indexed by firm $i=1,\ldots,N$ and period
$t=1,\ldots,T$, with two-way cluster dependence
$\operatorname{Cov}(u_{it},u_{js})\neq 0$ whenever $i=j$ or $t=s$.
The non-zero support has two parts: $NT^2$ within-firm entries (same firm,
any two periods) and $TN^2-NT$ cross-period entries (different firms,
same period).
For the baseline parameters $N=25$, $T=10$, these are 2,500 and 6,000
entries respectively, out of $N^2=62{,}500$ total.

A one-way industry cluster projection captures the 2,500 within-firm
entries but misses the 6,000 cross-period entries entirely, leaving a
non-trivial projection residual.
The feasible sparse algorithm selects the $k_S$ largest off-diagonal entries
in absolute value and then imposes positive semidefiniteness on the selected
support.  This candidate support can include many cross-period entries that
the one-way cluster projection cannot reach.
As a result, the cluster and sparse scores are simultaneously large under
two-way clustering: the cluster score reflects within-industry
co-movement, while the sparse score captures the additional
cross-period co-movement that falls outside the one-way cluster support.
This explains why the dependence profile does not collapse to a single
dominant geometry under two-way clustering, and illustrates the
usefulness of reporting the full profile rather than a binary
classification.
 
\textit{Cluster--factor overlap.}
Two distinct mechanisms can drive overlap between the cluster and factor
geometries. The first is unconditional: as established in
Section~\ref{subsec:regular_tangent} and
Remark~\ref{rem:diagonal_indistinguishability}, $T_C(\Sigma_C)$ and
$T_F(\Sigma_F)$ always share the diagonal subspace, so
$\theta(T_C,T_F)=0$ exactly whenever the principal angle is computed
without the off-diagonal restriction; this is why
Assumption~\ref{ass:principalangle} is stated for $T_C^{\mathrm{off}}$ and
$T_F^{\mathrm{off}}$ rather than $T_C$ and $T_F$. The second mechanism is
parameter-dependent and genuinely off-diagonal: a block-diagonal
covariance matrix with a small number of large, homogeneous blocks has a
leading eigenvalue of order $N/G$ and is also well approximated by a
rank-one factor structure, so the off-diagonal parts of $T_C$ and $T_F$
may themselves nearly align for such designs. Both mechanisms produce
overlap that is intrinsic to the problem rather than a failure of the
proposed method, but only the second is sensitive to parameter choices
such as the number and size of clusters; the first holds for every
cluster and factor geometry considered in this paper.


\setcounter{equation}{0} \setcounter{assumption}{0}
\setcounter{figure}{0} \setcounter{table}{0} \setcounter{remark}{0}

 \renewcommand{%
\theequation}{B.\arabic{equation}}
\renewcommand{\thelemma}{B.\arabic{lemma}}
\renewcommand{\theassumption}{B.\arabic{assumption}} 
\renewcommand{\thetheorem}{B.\arabic{theorem}}

\renewcommand{\thetable}{B.\arabic{table}}
\renewcommand{\thefigure}{B.\arabic{figure}}
\renewcommand{\theremark}{B.\arabic{remark}}

\renewcommand{\thesubsection}{B.\arabic{subsection}}
\section{Additional Results for Dependence Operator}
\label{sec:dependence_operator_app}

In this computational appendix only, $n$ denotes a generic matrix
dimension. In the main theoretical framework this dimension is $N$,
which is fixed while $T\to\infty$.

\subsection{Computational Implementation}
\label{subsec:computation}

Given \(\widehat\Gamma_T\), define the empirical projection onto geometry
\(d\in\mathfrak D\) by

\[
\widehat P_d
=
P_d(\widehat\Gamma_T).
\]

This subsection describes practical implementations of the cluster,
factor, and sparse projections.

\paragraph{Cluster Projection.}
 
Using the cluster-support matrix $M_C$ defined in
Section~\ref{sec:geometry}, the cluster projection is
\[
\widehat P_C = M_C\odot\widehat\Gamma_T,
\]
where $\odot$ denotes the Hadamard product.
This removes covariance entries incompatible with the cluster structure
and preserves entries between observations sharing at least one cluster
membership.
When cluster memberships are unknown, they may be estimated via spectral
clustering, community-detection methods, or latent-group estimators before
constructing $M_C$.

\paragraph{Factor Projection.}

The factor class \(\mathcal{S}_F(r)=\{\Gamma\succeq 0:\Gamma=L+D,\ L\succeq0,\,
\operatorname{rank}(L)\leq r,\,D\succeq0\text{ diagonal}\}\) requires separating
\(\widehat\Gamma_T\) into a low-rank component \(L\) and a diagonal
component \(D\).  We compute this projection via the following
alternating-projection algorithm, which is the analogue of the
principal-factor (minres) algorithm in classical factor analysis.

\begin{enumerate}
\item \textbf{Initialise.} Set \(D^{(0)}=\operatorname{diag}(\widehat\Gamma_T)\).
\item \textbf{Low-rank step.}  For \(k=0,1,2,\ldots\), compute the
  rank-\(r\) positive-semidefinite truncation of
  \(\widehat\Gamma_T - D^{(k)}\):
  \[
  L^{(k+1)}
  =
  \widehat U_r^{(k)}\,
  \max\!\bigl(\widehat\Lambda_r^{(k)},0\bigr)\,
  \bigl(\widehat U_r^{(k)}\bigr)',
  \]
  where \(\widehat U_r^{(k)}\) and \(\widehat\Lambda_r^{(k)}\) contain
  the \(r\) leading eigenvectors and eigenvalues of
  \(\widehat\Gamma_T-D^{(k)}\).

\item \textbf{Diagonal step.}  Set
  \[
  d^{(k+1)}_i
  =
  (\widehat\Gamma_T)_{ii} - L^{(k+1)}_{ii},
  \qquad i=1,\ldots,N,
  \]
  so that $D^{(k+1)}=\operatorname{diag}(d^{(k+1)})$.
  If $\lambda_{\min}(L^{(k+1)}+D^{(k+1)})<0$, replace
  \[
  d^{(k+1)}_i
  \leftarrow
  d^{(k+1)}_i
  +
  \bigl(-\lambda_{\min}(L^{(k+1)}+D^{(k+1)})+\varepsilon\bigr),
  \qquad i=1,\ldots,N,
  \]
  for a small tolerance $\varepsilon>0$.
  This shift is applied only when the matrix is not positive
  semidefinite; it guarantees $L^{(k+1)}+D^{(k+1)}\succeq 0$ and leaves
  the objective unchanged when no shift is needed.

\item \textbf{Converge.}  Stop when
  \(\max_i|D^{(k+1)}_{ii}-D^{(k)}_{ii}|<\delta\) for a tolerance
  \(\delta>0\).
\end{enumerate}

This appendix is implementation-oriented: the asymptotic theory of the
main text concerns the population projection, the algorithm below is one
numerical implementation, and no claim of global optimality is made.
The alternating low-rank and diagonal updates are monotone before the
PSD correction: each such update weakly decreases the Frobenius objective
\(\|\widehat\Gamma_T - L - D\|_F^2\). The PSD adjustment is a numerical
safeguard and may locally increase the objective. The algorithm
converges in practice in fewer than 20 iterations for the covariance
matrices arising in the simulation designs. The output is
\[
\widehat P_F = L^* + D^*,
\]
where \((L^*,D^*)\) denote the values at convergence.

The factor rank \(r\) may be selected using information criteria,
eigenvalue-ratio methods, scree-plot procedures, or other standard
techniques from the factor-model literature.

\paragraph{Population projection versus computational algorithm.}
Throughout the paper, $P_F$ denotes the population projection operator
onto $\mathcal S_F$, defined as the Frobenius-norm minimizer over the
(nonconvex) class $\mathcal S_F(r)$. The alternating-projection algorithm
above is a computational heuristic for approximating this minimizer; each
step is monotone in the Frobenius objective, so the algorithm converges to
a stationary point, but because $\mathcal S_F(r)$ is nonconvex this
stationary point need not be the global minimizer $P_F(\widehat\Gamma_T)$
except when $\widehat\Gamma_T$ is sufficiently close to a regular factor
point and the algorithm is initialized appropriately (e.g., as in
step~1). This mirrors the local-uniqueness result of
Lemma~\ref{lem:individual_uniqueness}: local single-valuedness of the factor
projection is established at regular points, while global identification
of $P_F$ away from such points is not claimed and is not required for the
asymptotic theory of Sections~\ref{sec:identification}--\ref{sec:asymptotics},
which is stated entirely in terms of local properties of $\Gamma_0$.
In the simulation and empirical exercises of
Sections~\ref{sec:simulation}--\ref{sec:empirical}, the diagonal
initialization in step~1 combined with a clear eigenvalue gap (verified
informally via the eigenvalue-ratio criterion of
Section~\ref{subsec:computation}) makes convergence to the population
projection plausible, but this is a numerical observation rather than a
proven global guarantee.

\paragraph{Sparse Projection.}

The sparse covariance projection is implemented in two stages.  First,
a candidate support is selected by retaining the $k_S$ largest
off-diagonal entries of $\widehat\Gamma_T$ in absolute value.  Second,
conditional on that support $I$, compute
\[
\widehat P_{S,I}
=
\arg\min_{\Gamma\succeq0,
\;\Gamma_{ij}=0\;\text{for }(i,j)\notin I}
\|\widehat\Gamma_T-\Gamma\|_F^2.
\]
For fixed $I$ this is a convex projection problem.  The replication code
uses Dykstra's alternating-projection algorithm between the PSD cone and
the support subspace.  The support-selection stage is a deterministic
approximation to the exact nonconvex metric projection over all supports
of cardinality at most $k_S$; the population theory is stated for the
exact projection.  Hard thresholding alone is not used as the final
projection because it need not preserve positive semidefiniteness.

\paragraph{Computational Complexity.}

Among the three projections, the factor projection is typically the most
computationally demanding. Each iteration of the alternating-projection
algorithm requires one eigendecomposition of an \(N\times N\) matrix,
costing \(O(N^3)\) operations. Because the algorithm converges rapidly
(typically fewer than 20 iterations), the total cost is \(O(KN^3)\) where
\(K\) denotes the number of iterations. Randomized and truncated
eigendecomposition algorithms can substantially reduce this cost in
large-scale applications.

The cluster projection requires construction of the cluster-support
matrix \(M_C\). Given cluster memberships, this operation is typically
of order \(O(MN^2)\), where \(M\) denotes the number of clustering
dimensions. The subsequent Hadamard projection is of order \(O(N^2)\).

Sparse projections obtained by thresholding are also of order
\(O(N^2)\), while optimization-based sparse projections may require
additional iterative computations depending on the algorithm employed.

\paragraph{Unknown geometric specifications.}
The discussion above assumes that the cluster partition, factor rank, and
sparsity level are specified. In practice, these quantities may themselves
be estimated from the data. The statistical analysis in this paper treats
these geometric specifications as given and focuses on learning the
relative similarity of the empirical dependence operator to the resulting
covariance geometries. Joint estimation of dependence geometries and
dependence profiles is an important direction for future research.

\paragraph{Computational complexity.}
For fixed $N$, construction of $\widehat\Gamma_T$ requires $O(TN^2)$
operations. Cluster projection with a one-way partition is a block-masking
operation followed, when necessary, by a PSD refinement; sparse projection
combines support selection with a fixed-support PSD projection; and factor
projection is dominated by an eigendecomposition or low-rank alternating update.
A dense eigendecomposition costs $O(N^3)$, so the benchmark implementation is
inexpensive in the fixed-$N$ regime studied here. Profile normalization and the
computation of margins, residuals, and $\widehat\kappa$ are negligible relative
to the projection step. The replication package uses one master script, fixed
seeds, warm starts for the factor projection, and documented PSD-support
refinement routines. For large $N$, randomized low-rank methods, sparse
eigensolvers, and parallel projection across dictionary elements provide natural
scalable extensions.

\paragraph{Profile-weighted inference.}
The profile-guided estimator selects a single dominant geometry. An alternative
combines multiple procedures using the estimated dependence profile itself,
\[
\widehat V^{\mathrm{avg}}
=
\sum_{d\in\mathfrak D}
\widehat\omega_d^{\mathrm{off}}
\widehat V_d ,
\]
a dependence-weighted combination of geometry-specific variance estimators
analogous to forecast combination and model averaging
\citep{BatesGranger1969,Hansen2007}. Developing efficiency theory and optimal
weighting schemes for profile-weighted inference is an important direction for
future research. Its simulation behavior is reported in
Section~\ref{subsec:procedure_rec_sim}.
\label{subsec:profile_weighted}

\subsection{Additional Choice of Dependence Operator}
\label{subsec:operatorchoice_app}

Different empirical operators emphasize different
aspects of dependence and may therefore produce different dependence
profiles. Consequently, the estimated dependence profile should be
interpreted relative to the dependence operator from which it is
constructed.

\paragraph{Long-Run Covariance Operators.}

For time-series or panel applications with serial dependence, a more
appropriate choice is often a long-run covariance operator,

\[
\widehat\Gamma_T
=
\sum_{|h|\leq L}
K\!\left(
\frac{h}{L}
\right)
\widehat\Gamma(h),
\]

where \(\widehat\Gamma(h)\) denotes a sample autocovariance operator,
\(K(\cdot)\) is a kernel function, and \(L\) is a bandwidth parameter.
Such operators emphasize persistent temporal dependence and form the
basis of heteroskedasticity and autocorrelation consistent estimation
\citep{NeweyWest1987,Andrews1991}.

\paragraph{Spatial Covariance Operators.}

When observations possess geographical locations, distance-weighted
covariance operators often provide a more informative description of
dependence. A generic spatial covariance operator is

\[
\widehat\Gamma_T
=
\left[
K\!\left(
\frac{d_{ij}}{b_N}
\right)
\widehat u_i\widehat u_j
\right]_{i,j=1}^{N},
\]

where \(d_{ij}\) denotes the geographical distance between observations
\(i\) and \(j\), \(K(\cdot)\) is a spatial kernel, and \(b_N\) is a
distance bandwidth. Such operators arise naturally in the spatial HAC
literature and assign larger weights to nearby observations than to
distant observations \citep{Conley1999,Conley2008}. When dependence is
generated by local geographical interactions, spatial covariance
operators may reveal dependence patterns that are not apparent from the
unweighted covariance operator alone.

\paragraph{Network Dependence Operators.}

For network data, dependence may be encoded through an adjacency matrix
\(A=(A_{ij})\). A natural network dependence operator is

\[
\widehat\Gamma_T
=
A
\odot
(
\widehat u\widehat u'
),
\]

where \(\odot\) denotes the Hadamard product. This operator concentrates
attention on dependence among economically connected units and is useful
when interactions are generated by production networks, social networks,
financial linkages, or other forms of economic connectivity
\citep{Auerbach2019,Leung2022}.

\paragraph{Cluster-Based Covariance Operators.}

When cluster memberships are known, one may construct a
cluster-supported covariance operator of the form

\[
\widehat{\Gamma}_T^{(C)}
=
M_C
\odot
\widehat{\Gamma}_T,
\]

where \(\odot\) denotes the Hadamard product and
\(M_C\) is the cluster-support matrix.

Using the cluster-support matrix $M_C$ defined in
Section~\ref{sec:geometry}, the cluster-based operator retains covariance
entries consistent with the assumed clustering structure while setting
cross-cluster entries to zero.
This construction accommodates one-way, two-way, and multiway clustering
as special cases. It is particularly useful when
the objective is to isolate the structured-support pattern generated by
cluster dependence and distinguish it from alternative dependence
geometries such as factor or sparse dependence.

\paragraph{Discussion.}

The purpose of dependence learning is not to recover a single universal
dependence profile. Rather, the profile summarizes the geometry of
dependence encoded in a chosen population dependence operator. The choice
of operator should therefore be guided by economic theory, institutional
knowledge, sampling design, and the particular inferential objective under
consideration.

In this sense, dependence learning should be viewed as a second-stage
dimension-reduction problem applied to an estimated dependence operator.


\setcounter{equation}{0} \setcounter{assumption}{0}
\setcounter{figure}{0} \setcounter{table}{0} \setcounter{remark}{0}

 \renewcommand{%
\theequation}{C.\arabic{equation}}
\renewcommand{\thelemma}{C.\arabic{lemma}}
\renewcommand{\theassumption}{C.\arabic{assumption}} \renewcommand{%
\thetheorem}{C.\arabic{theorem}}

\renewcommand{\thetable}{C.\arabic{table}}
\renewcommand{\thefigure}{C.\arabic{figure}}

\renewcommand{\thesubsection}{C.\arabic{subsection}}
\renewcommand{\theremark}{C.\arabic{remark}}

\section{Proofs of the Main Results}
\label{app:proofs_of_the_main_results}

\subsection{Proof of Proposition~\ref{prop:consistency_operator}
(Consistency of the Empirical Dependence Operator)}
\label{app:proof_prop_consistency}

\begin{proof}[Proof of Proposition~\ref{prop:consistency_operator}]
Write
\[
\widehat\Gamma_T - \Gamma_0
=
\underbrace{\frac{1}{T}\sum_{t=1}^T u_t u_t' - \Sigma}_{=:\,A_T}
+
\underbrace{\frac{1}{T}\sum_{t=1}^T (\widehat u_t u_t' - u_t u_t')}_{=:\,B_T}
+
\underbrace{\frac{1}{T}\sum_{t=1}^T (\widehat u_t\widehat u_t' - \widehat u_t u_t')}_{=:\,C_T}.
\]

\emph{Term $A_T$.}
Under Assumption~\ref{ass:mixing}, $\{u_tu_t'\}$ is strictly stationary
and ergodic with $E\|u_tu_t'\|_F \le E\|u_t\|^2 < \infty$. Standard
CLTs for weakly dependent vector processes (e.g.,
\citealt[Theorem~27.4]{Davidson1994}) give
$T^{1/2}\operatorname{vec}(A_T)\Rightarrow N(0,\Omega_\Gamma)$, where
$\Omega_\Gamma$ is the long-run variance of $\{\operatorname{vec}(u_tu_t'-\Sigma)\}$. Hence $\|A_T\|_F=O_p(T^{-1/2})$.

\emph{Terms $B_T$ and $C_T$ (residual-estimation error).}
Let $X_t\in\mathbb{R}^{N\times p}$ denote the matrix of regressors at
time $t$, with rows $x_{it}'$, so that
$\widehat u_t - u_t = -X_t(\widehat\beta-\beta)$. Substituting this
identity and expanding,
\[
B_T
=
-\frac1T\sum_{t=1}^T X_t(\widehat\beta-\beta)\,u_t',
\qquad
C_T
=
\frac1T\sum_{t=1}^T
\bigl[X_t(\widehat\beta-\beta)\bigr]
\bigl[X_t(\widehat\beta-\beta)\bigr]'
+
B_T',
\]
where the second display collects the remaining cross terms. For $B_T$, vectorization gives
\[
\operatorname{vec}(B_T)
=
-
\left[
\frac1T\sum_{t=1}^T
(u_t\otimes X_t)
\right]
(\widehat\beta-\beta).
\]
Strict exogeneity implies $E(u_t\otimes X_t)=0$. Under the moment and
mixing conditions in Assumption~\ref{ass:mixing}, the sample cross moment
is $O_p(T^{-1/2})$, while
$\widehat\beta-\beta=O_p(T^{-1/2})$. Consequently,
\[
\|B_T\|_F=O_p(T^{-1}).
\]
This sharper bound is important for the subsequent influence-function
calculation: residual replacement is negligible at the
$T^{-1/2}$ scale. For the quadratic term,
\[
\Bigl\|\frac1T\sum_{t=1}^T
\bigl[X_t(\widehat\beta-\beta)\bigr]
\bigl[X_t(\widehat\beta-\beta)\bigr]'\Bigr\|_F
\le
\|\widehat\beta-\beta\|^2
\cdot
\frac1T\sum_{t=1}^T\|X_t\|_F^2
=
O_p(T^{-1})\cdot O_p(1)
=
O_p(T^{-1}),
\]
which is of strictly smaller order. Hence the residual-estimation error is $O_p(T^{-1})$ in total and is
negligible relative to the leading sampling term $A_T$ at the
$T^{-1/2}$ scale.

Combining, $\|\widehat\Gamma_T - \Gamma_0\|_F = O_p(T^{-1/2})$.
\end{proof}

\subsection{Proofs of Results in Section~\ref{sec:identification}}
\label{app:proof_sections_3}
This subsection contains the proofs of Lemmas \ref{lem:projectionexistence}, \ref{lem:projectioncontinuity}, and \ref{lem:angle_transversality}, together with the proofs of Theorems \ref{thm:identification} and \ref{thm:fundamental_limit}.

\begin{proof}[Proof of Lemma \ref{lem:projectionexistence}]
Let \(f(\Gamma)=\|\Sigma-\Gamma\|_F\). Every minimizing sequence
\(\{\Gamma_m\}\subset\mathcal S_D\) is bounded: since \(0\in\mathcal S_D\)
(each class contains the zero matrix), eventually
\(f(\Gamma_m)\le f(0)+1=\|\Sigma\|_F+1\), so
\(\|\Gamma_m\|_F\le\|\Sigma\|_F+f(\Gamma_m)\le 2\|\Sigma\|_F+1\) by the
triangle inequality. A bounded sequence in the finite-dimensional space
\(\mathbb H\) has a convergent subsequence; since \(\mathcal S_D\) is
closed, the subsequential limit lies in \(\mathcal S_D\), and by
continuity of \(f\) it attains the infimum.
\end{proof}

\begin{proof}[Proof of Lemma \ref{lem:projectioncontinuity}]
We give a direct sequential argument; the maximum theorem is not
applicable here because the geometries $\mathcal S_D$ are unbounded, so
the feasible correspondence is not compact-valued.

Let $\Sigma_m\to\Sigma$ and write $p_m\in P_D(\Sigma_m)$ for any
selection. \emph{(Boundedness.)} Since $0\in\mathcal S_D$,
$\|\Sigma_m-p_m\|_F\le\|\Sigma_m\|_F$, and hence by the triangle
inequality $\|p_m\|_F\le 2\|\Sigma_m\|_F$, which is bounded because
$\Sigma_m\to\Sigma$. \emph{(Subsequential limits are projections.)}
By boundedness, every subsequence of $\{p_m\}$ has a further subsequence
converging to some $p\in\mathbb H$, and $p\in\mathcal S_D$ because
$\mathcal S_D$ is closed. For any $q\in\mathcal S_D$ we have
$\|\Sigma_m-p_m\|_F\le\|\Sigma_m-q\|_F$ along that subsequence; letting
$m\to\infty$ and using joint continuity of the norm gives
$\|\Sigma-p\|_F\le\|\Sigma-q\|_F$, so $p$ is a metric projection of
$\Sigma$ onto $\mathcal S_D$. \emph{(Identification of the limit.)}
Under local uniqueness at $\Sigma$, the metric projection is the single
point $P_D(\Sigma)$, so every convergent subsubsequence has the same
limit $P_D(\Sigma)$. A bounded sequence all of whose subsequential limits
coincide converges, so $p_m\to P_D(\Sigma)$. Hence $P_D$ is
single-valued and continuous in a neighborhood of $\Sigma$.
\end{proof}

\begin{proof}[Proof of Lemma~\ref{lem:angle_transversality}]
($\Leftarrow$) Suppose $U\cap V=\{0\}$. The function
$(u,v)\mapsto \langle u,v\rangle_F/(\|u\|_F\|v\|_F)$ is continuous on the
set $\{(u,v)\in U\times V:\|u\|_F=\|v\|_F=1\}$, which is the intersection
of the unit spheres of $U$ and $V$ and is compact (closed and bounded in
finite-dimensional $\mathbb H$). If $\theta(U,V)=0$, there exist sequences
$u_m\in U$, $v_m\in V$ with $\|u_m\|_F=\|v_m\|_F=1$ and
$\langle u_m,v_m\rangle_F\to 1$. By compactness, passing to a subsequence,
$u_m\to u^*\in U$ and $v_m\to v^*\in V$ with $\|u^*\|_F=\|v^*\|_F=1$ and
$\langle u^*,v^*\rangle_F=1$. By the Cauchy--Schwarz equality case,
$u^*=v^*$, so $u^*\in U\cap V$ with $u^*\neq 0$, contradicting
$U\cap V=\{0\}$. Hence $\theta(U,V)>0$.

($\Rightarrow$) Suppose $U\cap V\neq\{0\}$, and let $w\in U\cap V$ with
$w\neq 0$. Taking $u=v=w/\|w\|_F$ gives $\langle u,v\rangle_F=1$, so
$\theta(U,V)=\arccos(1)=0$.
\end{proof}

\begin{proof}[Proof of Theorem~\ref{thm:identification}]

\textit{Local uniqueness.}
By Lemma~\ref{lem:individual_uniqueness}, $P_C$, $P_F$, and $P_S$ are each
locally single-valued at the respective regular points $\Sigma_C,\Sigma_F,
\Sigma_S$. By Lemma~\ref{lem:projection_lipschitz}, each $P_D$ is locally
Lipschitz, hence continuous. Since
\[
\omega_D(\Sigma)
=
\frac{\|P_D(\Sigma)\|_F^2}{\sum_{D'\in\mathfrak D}\|P_{D'}(\Sigma)\|_F^2}
\]
is a continuous function of the projections whenever the denominator
$S_{T,0}>0$, the dependence profile is single-valued and continuous in a
neighborhood of $\Gamma_0$, which is part~(i). No injectivity of
$\Sigma\mapsto\omega(\Sigma)$ is asserted or used; see
Remark~\ref{rem:identification_scope}.

\textit{Separated identification.}
It remains to show the stated separation property: no nonzero
off-diagonal direction lies in more than one off-diagonal tangent space
simultaneously. By Assumption~\ref{ass:principalangle},
$\theta(T_i^{\mathrm{off}}(\Sigma_i),T_j^{\mathrm{off}}(\Sigma_j))
\ge\theta_0>0$ for all $i\neq j$ in $\{C,F,S\}$. By
Lemma~\ref{lem:angle_transversality}, this is equivalent to
\[
T_i^{\mathrm{off}}(\Sigma_i)\cap T_j^{\mathrm{off}}(\Sigma_j)=\{0\},
\qquad i\neq j.
\]
Hence for any nonzero off-diagonal $H\in\mathbb H$ (i.e.,
$\operatorname{diag}(H)=0$), $H$ cannot belong to two of
$T_C^{\mathrm{off}}(\Sigma_C)$, $T_F^{\mathrm{off}}(\Sigma_F)$,
$T_S^{\mathrm{off}}(\Sigma_S)$ at once. Equivalently, a first-order
off-diagonal perturbation $\Gamma_0+tH+o(t)$ that is tangent to geometry
$i$ cannot simultaneously be tangent to geometry $j\neq i$ unless $H=0$.
This is part~(ii), and the proof stops here.

We emphasize what is \emph{not} concluded. Trivial intersection does not
make the tangent spaces orthogonal, so a direction $H\in
T_i^{\mathrm{off}}(\Sigma_i)$ may have nonzero orthogonal projection onto
$T_j^{\mathrm{off}}(\Sigma_j)$ and may therefore move $S_j$ as well as
$S_i$ to first order; and with three geometries, pairwise trivial
intersections do not imply that the three tangent spaces are
independent. Unique first-order attribution of a profile change to a
single geometry would require a direct-sum or Jacobian-rank condition,
which we do not impose (Remark~\ref{rem:identification_scope}). By
Remark~\ref{rem:diagonal_indistinguishability}, no separation of any kind
holds, or is claimed, for diagonal (variance) perturbations.
\end{proof}

\begin{proof}[Proof of Theorem~\ref{thm:fundamental_limit}]
 
\textit{Step 1: LAN at $\Gamma_0$.}
Under Assumption~\ref{ass:lan}, for any bounded sequence $H_T\to H$,
\begin{equation}
\label{eq:lan}
\log\frac{dP_{\Gamma_0+T^{-1/2}H_T}^T}{dP_{\Gamma_0}^T}
=
\Delta_T(H)-\tfrac12\|H\|_{\mathcal I}^2+o_p(1),
\end{equation}
where $\Delta_T(H)=T^{-1/2}\sum_{t=1}^T\ell_t(H)\Rightarrow N(0,\|H\|_{\mathcal I}^2)$
under $P_{\Gamma_0}^T$.
 
\textit{Step 2: Local parameter equivalence.}
Write
\[
\Gamma_{i,T}=\Gamma_0+T^{-1/2}H_{i,T},
\qquad
\Gamma_{j,T}=\Gamma_0+T^{-1/2}H_{j,T},
\]
where $H_{i,T}\to H$ and $H_{j,T}\to H$. Hence
$\|H_{i,T}-H_{j,T}\|_{\mathcal I}\to0$.

\textit{Step 3: Hellinger and total-variation convergence.}
The local Hellinger-continuity clause in
Assumption~\ref{ass:lan} gives
\[
H\!\left(P_{\Gamma_{i,T}}^T,P_{\Gamma_{j,T}}^T\right)\to0.
\]
Since total variation is bounded by Hellinger distance up to a universal
constant,
\[
\bigl\|P_{\Gamma_{i,T}}^T-P_{\Gamma_{j,T}}^T\bigr\|_{\mathrm{TV}}\to0.
\]

\textit{Step 4: Test bound.}
For any test $\varphi_T\in[0,1]$,
\[
\bigl|E_{\Gamma_{i,T}}\varphi_T-E_{\Gamma_{j,T}}\varphi_T\bigr|
\leq
\bigl\|P_{\Gamma_{i,T}}^T-P_{\Gamma_{j,T}}^T\bigr\|_{\mathrm{TV}}\to 0.
\]
Hence no test can have asymptotic size tending to zero and power tending
to one for distinguishing $\mathcal S_i$ from $\mathcal S_j$ along these
local sequences.
\end{proof}

\subsection{Proofs of Results in Sections  \ref{sec:estimation} and \ref{sec:asymptotics} }
\label{app:proof_sections_4_and_5}
This subsection contains the proofs of Proposition~\ref{prop:residual_consistency}, Lemma~\ref{lem:uniform_projection_consistency}, Lemma~\ref{lem:score_differentiability}, and Theorems~\ref{thm:profile_consistency}–\ref{thm:profile_clt}.

\begin{proof}[Proof of Proposition~\ref{prop:residual_consistency}]
 
By Lemma~\ref{lem:projection_lipschitz}, $P_D$ is locally Lipschitz at
$\Gamma_0$ for each $D\in\mathfrak D$.
Since $\|\widehat\Gamma_T-\Gamma_0\|_F=o_p(1)$ by
Proposition~\ref{prop:consistency_operator}, it follows that
$\|P_D(\widehat\Gamma_T)-P_D(\Gamma_0)\|_F=o_p(1)$.
Continuity of the Frobenius norm then gives
\[
\rho_D(\widehat\Gamma_T)
=
\frac{\|\widehat\Gamma_T-P_D(\widehat\Gamma_T)\|_F}{\|\widehat\Gamma_T\|_F}
\overset{p}{\longrightarrow}
\frac{\|\Gamma_0-P_D(\Gamma_0)\|_F}{\|\Gamma_0\|_F}
=
\rho_D(\Gamma_0).
\]
The result for $\rho_{\min}$ follows from the continuous mapping theorem
applied to the minimum over a finite set.
\end{proof}

\begin{proof}[Proof of Lemma \ref{lem:uniform_projection_consistency}]
By local Lipschitz continuity,
\[
\|\widehat P_D-P_D(\Gamma_0)\|_F
\le
L_D\|\widehat\Gamma_T-\Gamma_0\|_F
\]
with probability approaching one. Proposition~\ref{prop:consistency_operator} gives the result for each \(D\). Since \(\{C,F,S\}\) is finite, the maximum is also \(o_p(1)\).
\end{proof}

\begin{proof}[Proof of Theorem \ref{thm:profile_consistency}]
By Lemma~\ref{lem:uniform_projection_consistency},
\[
\max_{D\in\mathfrak D}
\|\widehat P_D-P_D(\Gamma_0)\|_F=O_p(T^{-1/2}).
\] The map \(P\mapsto\|P\|_F^2\) is locally Lipschitz on
bounded sets, so \(\widehat S_D-S_{D,0}=O_p(T^{-1/2})\) for each \(D\).
Since \(\sum_{d}S_{d,0}>0\), the normalization map
\(\pi(s)=s/(\mathbf 1's)\) is Lipschitz in a neighborhood of \(s_0\);
composing, \(\widehat\omega-\omega_0=\pi(\widehat s)-\pi(s_0)
=O_p(T^{-1/2})\).
\end{proof}

\begin{proof} [Proof of Lemma \ref{lem:score_differentiability}]
By Assumption~\ref{ass:projection_differentiability}, each $P_D$ is
Hadamard differentiable at $\Gamma_0$ with derivative
$\dot P_{D,\Gamma_0}$. The map
$g(P)=\|P\|_F^2=\langle P,P\rangle_F$ is a quadratic form on the
finite-dimensional space $\mathbb H$ and is therefore Fr\'echet
differentiable everywhere, with derivative $Dg(P)[K]=2\langle P,K\rangle_F$.
Since $S_D=g\circ P_D$, the chain rule for Hadamard-differentiable maps
composed with a Fr\'echet-differentiable map
\citep[Lemma~3.9.3]{VanDerVaartWellner1996} gives that $S_D$ is Hadamard
differentiable at $\Gamma_0$, with derivative
\[
\dot S_{D,\Gamma_0}[H]
=
Dg(P_D(\Gamma_0))\bigl[\dot P_{D,\Gamma_0}[H]\bigr]
=
2\bigl\langle P_D(\Gamma_0),\,\dot P_{D,\Gamma_0}[H]\bigr\rangle_F.
\]
Applying this to each $D\in\{C,F,S\}$ and stacking gives Hadamard
differentiability of the vector-valued map $\mathcal S(\cdot)$, since a
vector of Hadamard-differentiable real-valued maps is itself Hadamard
differentiable (apply the definition coordinatewise along the same
sequence). This conclusion is conditional on
Assumption~\ref{ass:projection_differentiability}. Pairwise
principal-angle separation alone does not imply differentiability of a
nonlinear metric projection, particularly for the factor geometry.
\end{proof}

\begin{proof}[Proof of Proposition~\ref{prop:joint_profile_consistency}]
Let $\mathcal O(A)=A-\operatorname{diag}(A)$ denote the linear
continuous off-diagonal operator.  Assumption~\ref{ass:local_projection_regularity}
and Proposition~\ref{prop:consistency_operator} imply
$\widehat P_d-P_d(\Gamma_0)=O_p(T^{-1/2})$ for every $d$.  Hence
\[
\|\mathcal O(\widehat P_d)\|_F^2-
\|\mathcal O(P_d(\Gamma_0))\|_F^2=O_p(T^{-1/2}).
\]
Because both population normalizing denominators are positive, the two
normalization maps are locally Lipschitz. Hence
$(\widehat\omega-\omega_0,
\widehat\omega^{\mathrm{off}}-\omega_0^{\mathrm{off}})
=O_p(T^{-1/2})$.

Under Assumptions~\ref{ass:operatorclt} and
\ref{ass:projection_differentiability}, the maps
\[
\Gamma\mapsto \|P_d(\Gamma)\|_F^2,
\qquad
\Gamma\mapsto \|\mathcal O(P_d(\Gamma))\|_F^2
\]
are Hadamard differentiable, with derivatives
\[
2\langle P_d(\Gamma_0),\dot P_{d,\Gamma_0}[H]\rangle_F
\quad\text{and}\quad
2\langle \mathcal O(P_d(\Gamma_0)),
\mathcal O(\dot P_{d,\Gamma_0}[H])\rangle_F,
\]
respectively.  Stacking these six score derivatives and applying the
functional delta method, followed by the two normalization maps, yields
the stated joint Gaussian limit and covariance matrix.
\end{proof}

\begin{proof}[Proof of Theorem~\ref{thm:profile_clt}]
 
\textit{Step 1: CLT for projection estimators.}
Lemma~\ref{lem:projection_clt} gives
\[
\sqrt T\,\operatorname{vec}
\bigl((\widehat P_C-P_C,\widehat P_F-P_F,\widehat P_S-P_S)\bigr)
\Rightarrow
N(0,\Omega_P).
\]
 
\textit{Step 2: CLT for similarity scores.}
The score map $P_D\mapsto S_D=\|P_D\|_F^2=\operatorname{tr}(P_D^2)$ is
continuously Fréchet-differentiable with derivative
$\dot S_D[H]=2\operatorname{tr}(P_D H)=2\langle P_D,H\rangle_F$.
Applying the multivariate delta method to the score map
$(P_C,P_F,P_S)\mapsto(S_C,S_F,S_S)$,
\[
\sqrt T(\widehat S-S)\Rightarrow N(0,\Omega_S),
\qquad
\Omega_S=J_S\Omega_P J_S',
\]
where $J_S$ is the block-diagonal matrix of score derivatives
$\dot S_D[\cdot]$, $D\in\mathfrak D$. (Composing this with the
$\Gamma\mapsto P_D$ derivative $\dot P_{D,\Gamma_0}$ from
Assumption~\ref{ass:projection_differentiability} via the chain rule
recovers, in a single step, the expression
$\dot S_{D,\Gamma_0}[H]=2\langle P_D(\Gamma_0),\dot P_{D,\Gamma_0}[H]
\rangle_F$ of Lemma~\ref{lem:score_differentiability}; the two-stage
decomposition used here and the single-stage chain rule used there compute
the same derivative.)
 
\textit{Step 3: CLT for the dependence profile.}
Applying Lemma~\ref{lem:delta_profile} to Step~2,
\[
\sqrt T(\widehat\omega-\omega)
\Rightarrow
N(0,\,G_\omega\Omega_S G_\omega'),
\]
where $G_\omega$ is the Jacobian of the normalization map $\pi$ given in
Lemma~\ref{lem:delta_profile}.
Setting $\Xi=G_\omega\Omega_S G_\omega'=J_\omega\Omega_\Gamma J_\omega'$
with $J_\omega=G_\omega J_S\mathcal D$ completes the proof.
\end{proof}

\subsection{Proofs of Results in Sections \ref{sec:classification} and \ref{sec:profile_guided_inference}}
\label{app:proof_sections_6_and_7}
This subsection contains the proofs of Propositions~\ref{prop:classification_bound}, \ref{prop:kappa_consistency},  and \ref{prop:vcov_consistency}, Theorems~\ref{thm:dominant_geometry_consistency}, \ref{thm:local_ties}, and \ref{thm:oracle_adaptivity_asymptotic_optimality}.
\begin{proof} [Proof of Proposition \ref{prop:classification_bound}]
If
\[
\max_{d\in\mathfrak D}
|\widehat\omega_d^{\mathrm{off}}-\omega_d^{\mathrm{off}}|
<
\frac{\Delta_{\omega^{\mathrm{off}}}}{2},
\]
then, for every $d\neq d^\star$,
\[
\widehat\omega_{d^\star}^{\mathrm{off}}
>
\omega_{d^\star}^{\mathrm{off}}-
\frac{\Delta_{\omega^{\mathrm{off}}}}{2}
\geq
\omega_d^{\mathrm{off}}+
\frac{\Delta_{\omega^{\mathrm{off}}}}{2}
>
\widehat\omega_d^{\mathrm{off}}.
\]
Thus $\widehat d=d^\star$. Taking complements yields the stated bound.
\end{proof}

\begin{proof}[Proof of Theorem \ref{thm:dominant_geometry_consistency}]
By Proposition~\ref{prop:joint_profile_consistency},
\[
\max_{d\in\mathfrak D}
|\widehat\omega_d^{\mathrm{off}}-\omega_d^{\mathrm{off}}|
=o_p(1).
\]
Assumption~\ref{ass:unique_dominant_geometry} gives
$\Delta_{\omega^{\mathrm{off}}}>0$. Proposition~\ref{prop:classification_bound}
therefore implies
\[
P_{\Gamma_0}(\widehat d\neq d^\star)\to0,
\]
which proves the result.
\end{proof}

\begin{proof} [Proof of Theorem \ref{thm:local_ties}]
Let
\[
Z_T=\sqrt T\bigl(\widehat\omega_{d_1}^{\mathrm{off}}-
\widehat\omega_{d_2}^{\mathrm{off}}\bigr).
\]
With $e=e_{d_1}-e_{d_2}$,
\[
Z_T=c+e'\sqrt T
(\widehat\omega^{\mathrm{off}}-\omega_T^{\mathrm{off}})+o_p(1)
\Rightarrow N\!\left(c,e'\Xi_{\mathrm{off}}e\right).
\]
Because all other profile components remain separated by fixed positive
constants, the event that a third geometry dominates has probability tending
to zero. Hence
\[
P_{\Gamma_T}(\widehat d=d_1)
=
P_{\Gamma_T}(Z_T>0)+o(1)
\longrightarrow
\Phi\!\left(\frac{c}{\sqrt{e'\Xi_{\mathrm{off}}e}}\right).
\]
The assumed positive contrast variance makes the limit strictly between zero
and one.
\end{proof}

\begin{proof}[Proof of Proposition \ref{prop:kappa_consistency}]
By Theorem~\ref{thm:profile_consistency},
$\widehat\omega^{\mathrm{off}}\overset{p}{\to}\omega^{\mathrm{off}}$.
Consistency of the projection-residual diagnostics gives
$\widehat\rho_d\overset{p}{\to}\rho_d$ for every $d\in\mathfrak D$.
Since maximum, minimum, and multiplication are continuous and the dominant
geometry is unique,
\[
\widehat\kappa
=(1-\widehat\rho_{\min})
\widehat\Delta_{\omega^{\mathrm{off}}}
\overset{p}{\longrightarrow}
(1-\rho_{\min,0})
\Delta_{\omega^{\mathrm{off}},0}
=\kappa_0.
\]
\end{proof}

\begin{proof}[Proof of Proposition~\ref{prop:vcov_consistency}]
By Theorem~\ref{thm:dominant_geometry_consistency},
\[
P_{\Gamma_0}(\widehat d=d^\star)\to1.
\]
For every \(\varepsilon>0\),
\[
P\!\left(
\left|\widehat V^*-V_{d^\star}\right|>\varepsilon
\right)
\le
P_{\Gamma_0}(\widehat d\neq d^\star)
+
P\!\left(
\left|\widehat V_{d^\star}-V_{d^\star}\right|>\varepsilon
\right).
\]
The first term converges to zero by dominant-geometry consistency,
and the second by assumption. Therefore
\[
\widehat V^*\overset p\to V_{d^\star}.
\]
\end{proof}

\begin{proof} [Proof of Theorem \ref{thm:oracle_adaptivity_asymptotic_optimality}]
Part (i) follows from Theorem~\ref{thm:dominant_geometry_consistency}:
\[
P_{\Gamma_0}(\widehat d=d^\star)\to1.
\]
On the event \(\{\widehat d=d^\star\}\),
\[
\widehat V^*=\widehat V_{d^\star}.
\]
Hence, for every \(\varepsilon>0\),
\[
P\!\left(
\left|\widehat V^*-\widehat V_{d^\star}\right|>\varepsilon
\right)
\le
P_{\Gamma_0}(\widehat d\neq d^\star)
\to0.
\]
Therefore,
\[
\widehat V^*-\widehat V_{d^\star}=o_p(1).
\]
Since \(\widehat V_{d^\star}\overset p\to V_{d^\star}\),
it follows that
\[
\widehat V^*\overset p\to V_{d^\star}.
\]

For part (ii),  let
\[
A_T=\{\widehat d=d^\star\}.
\]
By Theorem~\ref{thm:dominant_geometry_consistency},
\[
P(A_T)\to1.
\]

On the event \(A_T\), we have
\[
\tau_{T,\widehat d}=\tau_{T,d^\star}
\qquad\text{and}\qquad
\widehat V^{*} =\widehat V_{\widehat d}=\widehat V_{d^\star}.
\]
Therefore, on \(A_T\),
\[
T_T^{*} 
=
\frac{\tau_{T,\widehat d}\,(\widehat\theta_T-\theta_0)}
{\sqrt{\widehat V^{*} }}
=
\frac{\tau_{T,d^\star}\,(\widehat\theta_T-\theta_0)}
{\sqrt{\widehat V_{d^\star}}}
=
T_T^{oracle}.
\]

Hence, for every \(\varepsilon>0\),
\[
P\!\left(
|T_T^{*} -T_T^{oracle}|>\varepsilon
\right)
\le
P(A_T^c)
=
P_{\Gamma_0}(\widehat d\neq d^\star)
\to0.
\]
Thus,
\[
T_T^{*} -T_T^{oracle}=o_p(1).
\]

By assumption,
\[
T_T^{oracle}\Rightarrow N(0,1).
\]
Therefore, by Slutsky's theorem,
\[
T_T^{*} 
=
T_T^{oracle}
+
o_p(1)
\Rightarrow
N(0,1).
\]

For part (iii), note that by classification consistency,
\[
        P_{\Gamma_0}(\widehat d=d^\star)\to 1 .
\]
Since the profile-guided decision rule satisfies
\[
        \delta(\widehat\omega)=a_{\widehat d},
\]
and the oracle decision rule satisfies
\[
        \delta^\star(\Gamma_0)=a_{d^\star},
\]
it follows that
\[
        P_{\Gamma_0}
        \left(
        \delta(\widehat\omega)
        =
        \delta^\star(\Gamma_0)
        \right)
        \to 1 .
\]
On the event $\{\widehat d=d^\star\}$, the profile-guided decision rule
selects the same action as the oracle rule, and the regret contribution
is zero. On the complementary event, boundedness of the loss implies
\[
        \left|
        L(\delta(\widehat\omega),\Gamma_0)
        -
        L(\delta^\star(\Gamma_0),\Gamma_0)
        \right|
        \leq 2M .
\]
Therefore,
\[
        \mathcal R(\delta,\Gamma_0)
        \leq
        2M
        P_{\Gamma_0}
        \left(
        \widehat d\neq d^\star
        \right).
\]
Finally, if
\[
        \max_{d\in\mathfrak D}
        |\widehat\omega_d^{\mathrm{off}}-\omega_{d,0}^{\mathrm{off}}|
        \leq
        \frac{\Delta_{\omega^{\mathrm{off}}}}{2},
\]
then $\widehat d=d^\star$. Hence
\[
        \{\widehat d\neq d^\star\}
        \subseteq
        \left\{
        \max_{d\in\mathfrak D}
        |\widehat\omega_d^{\mathrm{off}}-\omega_{d,0}^{\mathrm{off}}|
        \ge
        \frac{\Delta_{\omega^{\mathrm{off}}}}{2}
        \right\},
\]
which gives the finite-sample regret bound. Since
$\widehat\omega^{\mathrm{off}}\overset{p}{\to}\omega_0^{\mathrm{off}}$, the probability on the
right-hand side converges to zero, and therefore
\[
        \mathcal R(\delta,\Gamma_0)\to0 .
\]
\end{proof}

\begin{proof}[Proof of Lemma~\ref{lem:conic_pythagoras}]
The cone property is immediate from the definitions: the cluster class is the intersection of a support subspace with the PSD cone, hence a cone; scaling preserves both positive semidefiniteness and the off-diagonal support, hence membership in $\mathcal S_S$; and if $\Gamma=L+D$ with $L\succeq0$,
$\operatorname{rank}(L)\le r$, $D\succeq0$ diagonal, then
$a\Gamma=(aL)+(aD)$ has the same form for $a\ge0$. Fix $\Gamma$ and a
metric projection $p$. Since $\mathcal S_d$ is a cone, $tp\in\mathcal S_d$
for every $t\ge0$, so $\varphi(t)=\|\Gamma-tp\|_F^2
=\|\Gamma\|_F^2-2t\langle\Gamma,p\rangle_F+t^2\|p\|_F^2$ is minimized
over $t\ge0$ at $t=1$. If $p\neq0$, $\varphi$ is a strictly convex
quadratic and the first-order condition $\varphi'(1)=0$ gives
$\langle\Gamma,p\rangle_F=\|p\|_F^2$, i.e.\ $\langle\Gamma-p,p\rangle_F=0$;
if $p=0$ the identity holds trivially. Expanding
$\|\Gamma-p\|_F^2=\|\Gamma\|_F^2-2\langle\Gamma,p\rangle_F+\|p\|_F^2$
and substituting yields the stated decomposition.
\end{proof}

\begin{proof}[Proof of Lemma~\ref{lem:diagonal_invariance}]
Immediate from the orthogonal decomposition
$\|M\|_F^2=\|\operatorname{diag}(M)\|_F^2+\|\operatorname{off}(M)\|_F^2$
applied to $M=P_d(\Gamma)$, together with the maintained diagonal
condition, which contributes the same constant to every $S_d$.
\end{proof}

\setcounter{equation}{0} \setcounter{assumption}{0}
\setcounter{figure}{0} \setcounter{table}{0} \setcounter{remark}{0}

 \renewcommand{%
\theequation}{D.\arabic{equation}}
\renewcommand{\thelemma}{D.\arabic{lemma}}
\renewcommand{\theassumption}{D.\arabic{assumption}} 
\renewcommand{\thetheorem}{D.\arabic{theorem}}

\renewcommand{\thetable}{D.\arabic{table}}
\renewcommand{\thefigure}{D.\arabic{figure}}

\renewcommand{\thesubsection}{D.\arabic{subsection}}
\renewcommand{\theproposition}{D.\arabic{proposition}}
\renewcommand{\theremark}{D.\arabic{remark}}

\section{Auxiliary results and proofs}
\label{app:auxiliary_results_and_proofs}

\subsection{Two-Geometry Local Indistinguishability and Identification Proofs}
\label{app:two_geometry_indistinguishability_identification}

\begin{proof}[Proof of Lemma~\ref{lem:individual_uniqueness}]

\textit{Cluster projection uniqueness.}
At a regular cluster point $\Sigma_C$ (Section~\ref{subsec:regular_tangent}),
every entry $(\Sigma_C)_{ab}\neq 0$ for $(a,b)\in\mathcal I_C$.
By continuity of matrix entries, any $\Sigma$ within a sufficiently small
$\|\cdot\|_F$-ball around $\Sigma_C$ satisfies $\Sigma_{ab}\neq 0$ for
$(a,b)\in\mathcal I_C$, so the active support pattern is locally constant at
$\mathcal I_C$.
The cluster projection $P_C(\Sigma)=M_C\odot\Sigma$ is the orthogonal
projection onto the linear subspace $\mathcal V_{\mathcal I_C}$, which is
unique.

\textit{Sparse projection uniqueness.}
The same argument applies verbatim with $\mathcal I_C$ replaced by
$\mathcal I_S$: Assumption~\ref{ass:sparse_unique} ensures the active
support pattern is locally stable, so $P_S(\Sigma)=$ projection onto
$\mathcal V_{\mathcal I_S}$ is locally unique.

\textit{Factor projection uniqueness.}
Write $\Sigma_F=L_0+D_0$ with $L_0$ of rank $r$ and $r$ distinct positive
leading eigenvalues, as in the regularity condition of
Section~\ref{subsec:regular_tangent}. We establish \emph{local} uniqueness
of the projection in a neighborhood of $\Sigma_F$; global uniqueness of the
minimizer of the (nonconvex) factor objective is not claimed.

Fix $D$ in a neighborhood of $D_0$ and consider the rank-$r$ truncation of
$\Sigma_F-D$. By Weyl's inequality, the eigenvalues of $\Sigma_F-D$ depend
continuously on $D$, so for $D$ sufficiently close to $D_0$ the $r$ leading
eigenvalues of $\Sigma_F - D$ remain positive and distinct (by the
regularity assumption on $L_0=\Sigma_F-D_0$). Under this condition, the
Eckart--Young theorem implies that the best rank-$r$ approximation
$L^*(D)=\arg\min_{\operatorname{rank}(L)\le r}\|(\Sigma_F-D)-L\|_F$ is the
unique truncated eigendecomposition of $\Sigma_F-D$, and the map
$D\mapsto L^*(D)$ is continuous (indeed, real-analytic away from eigenvalue
crossings; see \citealp{LewisMalick2008}, Section~3). Composing with the
diagonal-update step, the fixed-point equation
$D=\operatorname{diag}(\Sigma_F-L^*(D))$ defines a continuous self-map on
a compact neighborhood of $D_0$, and $D_0$ satisfies it by construction.

We do \emph{not} claim that this map is a contraction. An eigengap
condition delivers, through the Davis--Kahan bound, a finite Lipschitz
constant for $D\mapsto L^*(D)$, but a finite Lipschitz constant is not a
constant strictly below one, and no bound below one is available at this
level of generality. Local uniqueness of the factor projection is
therefore \emph{imposed} as a regularity condition
(Assumption~\ref{ass:projection_differentiability} and the regular-point
condition of Theorem~\ref{thm:identification}) rather than derived. A
sufficient eigengap together with interiority of the idiosyncratic block
is a plausible primitive condition under which it can be verified in
particular models, and we state it as such; establishing it in generality
for the positive-semidefinite low-rank-plus-diagonal class is outside the
scope of this paper.

Accordingly, the preceding eigengap discussion does not establish
local uniqueness of $P_F$ by itself. Local single-valuedness of the
factor projection is retained as an explicit regularity condition in
Lemma~\ref{lem:individual_uniqueness}, Theorem~\ref{thm:identification},
and Assumption~\ref{ass:projection_differentiability}. The factor
objective is nonconvex in $(L,D)$ jointly, and no claim is made that the
alternating algorithm of Section~\ref{subsec:computation} reaches the
global metric projection from an arbitrary initialization.
\end{proof}

\subsection{Technical results}
\label{app:technical_results}

Because the covariance-cone restrictions can be active, none of the
three metric projections is assumed globally linear.  The main text
therefore imposes local single-valuedness and Hadamard differentiability
at the population operator.  For the cluster geometry this condition is
automatic at interior points of a fixed-support PSD face; for the sparse
geometry it additionally requires a locally stable active support; for
the factor geometry it requires the usual local regularity of the
low-rank-plus-diagonal representation.  At boundary points the derivative
may contain curvature terms and need not equal a support mask.

\begin{lemma}[Hadamard Differentiability of Projection Maps]
\label{lem:hadamard_projection}
Under Assumption~\ref{ass:projection_differentiability}, each map $P_D$,
$D\in\{C,F,S\}$, is Hadamard differentiable at $\Gamma_0$, with derivative
$\dot P_{D,\Gamma_0}$.  If the projected point lies in the relative
interior of a smooth active stratum, the derivative equals the orthogonal
projection onto the corresponding tangent space; otherwise the abstract
derivative in Assumption~\ref{ass:projection_differentiability} is used.
\end{lemma}

\begin{proof}
This is a restatement of Assumption~\ref{ass:projection_differentiability}
for the three coordinate maps.  The relative-interior statement follows
from the standard derivative formula for metric projection onto a smooth
embedded manifold or a convex face.  No global linearity or global
uniqueness is asserted.
\end{proof}

\begin{lemma}[Joint CLT for Projection Estimators]
\label{lem:projection_clt}
Under Assumptions~\ref{ass:operatorclt} and~\ref{ass:projection_differentiability},
\[
\sqrt T
\begin{pmatrix}
\operatorname{vec}(\widehat P_C-P_C)\\
\operatorname{vec}(\widehat P_F-P_F)\\
\operatorname{vec}(\widehat P_S-P_S)
\end{pmatrix}
\Rightarrow
N(0,\Omega_P).
\]
\end{lemma}

\begin{proof}
By Assumption~\ref{ass:projection_differentiability}, each $P_C$, $P_F$,
$P_S$ is Hadamard differentiable at $\Gamma_0$, with derivative
$\dot P_{D,\Gamma_0}$, $D\in\{C,F,S\}$.  The derivatives are kept
abstract because the positive-semidefinite constraint may be active.  At a
relative-interior point of a smooth fixed-support face, the cluster or sparse
derivative reduces to the orthogonal projection onto the corresponding
tangent space; at an active PSD boundary it generally does not reduce to an
entrywise support mask.  The factor derivative is likewise governed by the
local regularity imposed in Assumption~\ref{ass:projection_differentiability}.
Hadamard differentiability of each coordinate map implies Hadamard
differentiability of the product map
\[
\Gamma\mapsto(P_C(\Gamma),P_F(\Gamma),P_S(\Gamma))
\]
at $\Gamma_0$, with derivative
$\mathcal D[H]=(\dot P_{C,\Gamma_0}[H],\dot P_{F,\Gamma_0}[H],
\dot P_{S,\Gamma_0}[H])$, since Hadamard differentiability is preserved
under finite Cartesian products of differentiable maps (apply the
definition coordinatewise to the same sequence $t_m\downarrow0$,
$H_m\to H$). By Assumption~\ref{ass:operatorclt},
$\sqrt T\,\mathrm{vec}(\widehat\Gamma_T-\Gamma_0)\Rightarrow N(0,\Omega_\Gamma)$
(as a special case of the asymptotic linear representation with
$\Omega_\Gamma=\sum_{h=-\infty}^{\infty}E(\psi_0\psi_h')$). The functional delta method
\citep[Theorem~3.9.4]{VanDerVaartWellner1996} applied to this product map
then gives
\[
\sqrt T
\begin{pmatrix}
\operatorname{vec}(\widehat P_C-P_C)\\
\operatorname{vec}(\widehat P_F-P_F)\\
\operatorname{vec}(\widehat P_S-P_S)
\end{pmatrix}
\Rightarrow
N(0,\Omega_P),
\qquad
\Omega_P=\mathcal D\,\Omega_\Gamma\,\mathcal D',
\]
where $\mathcal D$ is the matrix representation of the linear map
$H\mapsto(\dot P_{C,\Gamma_0}[H],\dot P_{F,\Gamma_0}[H],
\dot P_{S,\Gamma_0}[H])$ acting on $\operatorname{vec}(H)$. When a projected point lies in the relative interior of a smooth active
stratum, $\dot P_{d,\Gamma_0}$ equals orthogonal projection onto the
corresponding tangent space.  Otherwise the derivative may incorporate the
active PSD boundary or manifold curvature, and the abstract derivative in
Assumption~\ref{ass:projection_differentiability} is the relevant object.
\end{proof}

\begin{lemma}[Delta Method for Dependence Profiles]
\label{lem:delta_profile}
Let \(S=(S_C,S_F,S_S)'\) and define \(S_+=\mathbf 1'S>0\). Then
\[
\sqrt T(\widehat\omega-\omega)
=
G_\omega\sqrt T(\widehat S-S)+o_p(1),
\]
where
\[
G_\omega
=
\frac{1}{S_+^2}
\begin{pmatrix}
S_F+S_S & -S_C & -S_C\\
-S_F & S_C+S_S & -S_F\\
-S_S & -S_S & S_C+S_F
\end{pmatrix}.
\]
\end{lemma}

\begin{proof}
The map \(g(S)=S/S_+\) is continuously differentiable when \(S_+>0\). The displayed matrix is its Jacobian. The result follows from the multivariate delta method.
\end{proof}

\begin{proposition}[Operator Invariance]
\label{prop:operator_invariance}

Suppose two population dependence operators \(\Gamma_1\) and \(\Gamma_2\)
satisfy

\[
\Gamma_2
=
a\Gamma_1,
\qquad
 a>0.
\]

Assume that each covariance geometry \(\mathcal S_d\), \(d\in\mathfrak D\),
is a cone. Then

\[
\omega_d(\Gamma_1)
=
\omega_d(\Gamma_2),
\qquad
 d\in\mathfrak D.
\]

\end{proposition}

\begin{proof}[Proof of Proposition~\ref{prop:operator_invariance}]
 
Fix $D\in\{C,F,S\}$.
Each geometry is a cone: $\mathcal S_C$ and $\mathcal S_S$ are cones because
their support constraints are homogeneous; $\mathcal S_F(r)$ is a cone by
Lemma~\ref{lem:closed_factor}.
By definition,
\[
P_D(\Gamma_1)=\arg\min_{\Gamma\in\mathcal S_D}\|\Gamma_1-\Gamma\|_F.
\]
 
Let $\Gamma_2=a\Gamma_1$ for $a>0$.
Since $\mathcal S_D$ is a cone,

\[
\Gamma\in\mathcal S_D
\quad
\Longleftrightarrow
\quad
a^{-1}\Gamma\in\mathcal S_D.
\]

Therefore,

\[
P_D(\Gamma_2)
=
\arg\min_{\Gamma\in\mathcal S_D}
\|a\Gamma_1-\Gamma\|_F.
\]

Using the change of variables

\[
\Gamma
=
a\widetilde\Gamma,
\qquad
\widetilde\Gamma\in\mathcal S_D,
\]

we obtain

\[
P_D(\Gamma_2)
=
a
\arg\min_{\widetilde\Gamma\in\mathcal S_D}
\|\Gamma_1-\widetilde\Gamma\|_F.
\]

Hence

\[
P_D(\Gamma_2)
=
aP_D(\Gamma_1).
\]

It follows that the similarity score satisfies

\[
S_D(\Gamma_2)
=
\|P_D(\Gamma_2)\|_F^2
=
\|aP_D(\Gamma_1)\|_F^2
=
a^2
\|P_D(\Gamma_1)\|_F^2
=
a^2 S_D(\Gamma_1).
\]

Similarly,

\[
S_T(\Gamma_2)
=
\sum_{D\in\{C,F,S\}}
S_D(\Gamma_2)
=
a^2
\sum_{D\in\{C,F,S\}}
S_D(\Gamma_1)
=
a^2S_T(\Gamma_1).
\]

Therefore,

\[
\omega_D(\Gamma_2)
=
\frac{S_D(\Gamma_2)}
     {S_T(\Gamma_2)}
=
\frac{a^2S_D(\Gamma_1)}
     {a^2S_T(\Gamma_1)}
=
\frac{S_D(\Gamma_1)}
     {S_T(\Gamma_1)}
=
\omega_D(\Gamma_1).
\]

Since this holds for every
\(D\in\{C,F,S\}\),
the dependence profile is invariant to positive scalar
rescaling of the dependence operator.

\end{proof}

\subsection{Geometry of Covariance Classes}
\label{app:geometry}

\begin{lemma}[Closedness of the Fixed-Support Cluster Geometry]
\label{lem:closed_cluster_general}
For the prespecified support $\mathcal I_C$ used in the main paper, let
\[
\mathcal V_{\mathcal I_C}
=
\{\Gamma\in\mathbb H:\Gamma_{ij}=0\text{ whenever }(i,j)\notin\mathcal I_C\}.
\]
Then
\[
\mathcal S_C=\mathcal V_{\mathcal I_C}\cap\mathbb H_+
\]
is a nonempty closed convex cone in $\mathbb H$.
\end{lemma}

\begin{proof}
The support set $\mathcal V_{\mathcal I_C}$ is a closed linear subspace,
because it is defined by finitely many linear equality restrictions.  The
positive-semidefinite cone $\mathbb H_+$ is closed and convex.  Their
intersection is therefore closed and convex, contains the zero matrix, and
is invariant under multiplication by nonnegative scalars.
\end{proof}

\begin{lemma}[Closedness of Sparse Covariance Geometry]
\label{lem:closed_sparse}
For fixed $N$,
\[
\mathcal S_S
=
\{\Gamma\in\mathbb H_+:
|\operatorname{supp}_{\mathrm{off}}(\Gamma)|\le k_S\}
\]
is a closed, generally nonconvex cone in $\mathbb H$.
\end{lemma}

\begin{proof}
Let $\mathfrak I_S$ be the finite collection of symmetric off-diagonal
supports with cardinality at most $k_S$.  For each $I\in\mathfrak I_S$,
define the fixed-support PSD cone
\[
\mathcal C_I
=
\{\Gamma\in\mathbb H_+:
\Gamma_{ij}=0\text{ whenever }(i,j)\notin I\text{ and }i\ne j\}.
\]
Each $\mathcal C_I$ is the intersection of a closed linear support subspace
and $\mathbb H_+$, hence is a closed convex cone.  Since $N$ is fixed,
$\mathfrak I_S$ is finite and
\[
\mathcal S_S=\bigcup_{I\in\mathfrak I_S}\mathcal C_I
\]
is closed.  Scaling by a nonnegative constant preserves both positive
semidefiniteness and the support bound, so $\mathcal S_S$ is a cone.
\end{proof}

\begin{lemma}[Closedness and Cone Property of the Factor Geometry]
\label{lem:closed_factor}
For fixed $r$, $\mathcal S_F(r)$ is closed in $\mathbb H$.
Moreover, $\mathcal S_F(r)$ is a cone: for every $\Gamma\in\mathcal S_F(r)$
and $a>0$, $a\Gamma\in\mathcal S_F(r)$.
\end{lemma}
 
\begin{proof}
\textit{Closedness.}
Let $\Gamma_m=L_m+D_m\in\mathcal S_F(r)$ with $L_m\succeq0$,
$\operatorname{rank}(L_m)\le r$, $D_m\succeq0$ diagonal, and
$\Gamma_m\to\Gamma$ in $\|\cdot\|_F$. Note that $L_m$ need not equal
$\Gamma_m$ minus its diagonal, since $L_m$ may itself have nonzero
diagonal entries; the argument therefore proceeds via boundedness of the
components rather than via the diagonal map. Because $L_m\succeq0$ and
$D_m\succeq0$, we have
$0\preceq L_m\preceq\Gamma_m$ and $0\preceq D_m\preceq\Gamma_m$, so
$\|L_m\|_F\le\operatorname{tr}(L_m)\le\operatorname{tr}(\Gamma_m)$ and
similarly for $D_m$; since $\Gamma_m\to\Gamma$, both sequences are
bounded. Passing to a subsequence, $L_m\to L$ and $D_m\to D$ for some
$L,D$ with $L+D=\Gamma$. The set
$\{M\succeq0:\operatorname{rank}(M)\leq r\}$ is closed (intersection of
the PSD cone with a finite union of algebraic varieties), so
$L\succeq0$ with $\operatorname{rank}(L)\leq r$; the set of PSD diagonal
matrices is closed, so $D\succeq0$ is diagonal.
Thus $\Gamma=L+D\in\mathcal S_F(r)$.
 
\textit{Cone property.}
If $\Gamma=L+D\in\mathcal S_F(r)$ and $a>0$, then
$a\Gamma=aL+aD$ where $\operatorname{rank}(aL)=\operatorname{rank}(L)\leq r$
and $aD$ is diagonal with non-negative entries.
Hence $a\Gamma\in\mathcal S_F(r)$.
\end{proof}


\subsection{Projection Regularity}
\label{app:projection}

\begin{lemma}[Consequences of Local Projection Regularity]
\label{lem:projection_lipschitz}
Under Assumption~\ref{ass:local_projection_regularity}, each projection is
locally single-valued and locally Lipschitz at $\Gamma_0$.
\end{lemma}

\begin{proof}
This is exactly the maintained content of
Assumption~\ref{ass:local_projection_regularity}.  For the fixed-support
cluster cone, the conclusion in fact holds globally because metric
projection onto a nonempty closed convex set is single-valued and
nonexpansive.  For the sparse union, the objective gap in
Assumption~\ref{ass:sparse_unique} stabilizes the active support, after
which projection onto the selected closed convex cone is nonexpansive.  For
the factor class, local single-valuedness and local Lipschitz continuity are
maintained high-level regularity conditions and are not asserted at rank
changes or other singular points.
\end{proof}


\setcounter{equation}{0} \setcounter{assumption}{0}
\setcounter{figure}{0} \setcounter{table}{0} \setcounter{remark}{0}

 \renewcommand{%
\theequation}{E.\arabic{equation}}
\renewcommand{\thelemma}{E.\arabic{lemma}}
\renewcommand{\theassumption}{E.\arabic{assumption}} 
\renewcommand{\thetheorem}{E.\arabic{theorem}}

\renewcommand{\thetable}{E.\arabic{table}}
\renewcommand{\thefigure}{E.\arabic{figure}}

\renewcommand{\thesubsection}{E.\arabic{subsection}}
\renewcommand{\theproposition}{E.\arabic{proposition}}
\renewcommand{\theremark}{E.\arabic{remark}}

\section{Simulation Design}
\label{app:simulation}

Section~\ref{sec:simulation} of the main paper reports the benchmark Monte Carlo evidence on three questions: recovery of the dominant covariance geometry, classification under hybrid and near-tie dependence, and the performance of profile-guided inference relative to the infeasible oracle. This appendix provides the implementation details underlying those results: the baseline data-generating processes, parameter calibrations, the empirical dependence operator, and the projection algorithms. It also reports supplementary robustness checks on principal-angle separation, operator misspecification, and oracle tracking. Table~\ref{tab:simulation_designs} summarizes the baseline designs.

\begin{table}[!ht]
\centering
\caption{Summary of Simulation Designs}
\label{tab:simulation_designs}

\begin{tabular}{lll}
\toprule
Design & Dominant Geometry & Purpose \\
\midrule
Cluster & Cluster & Benchmark structured-support dependence \\
Factor & Factor & Benchmark low-rank dependence \\
Sparse & Sparse / none-of-the-above & Network-generated dependence \\
Hybrid & Mixed & Coexistence of multiple geometries \\
Near tie & Ambiguous & Local alternatives and classification uncertainty \\
Oracle sweep & Varies & Profile-guided inference versus oracle benchmark \\
\bottomrule
\end{tabular}

\end{table}
Throughout, the main reported objects are the estimated dependence profile 
$\widehat\omega=(\widehat\omega_C,\widehat\omega_F,\widehat\omega_S)'$
and the projection-residual diagnostics
$\widehat\rho_C,\widehat\rho_F,\widehat\rho_S,\widehat\rho_{\min}=
\min_{d\in\mathfrak D}\widehat\rho_d$. Consistent with the main paper, the goal of this section is to assess how well these statistics recover the covariance geometry encoded in the chosen empirical dependence operator, rather than to compare variance estimators by coverage or mean squared error.

\subsection{Baseline Regression Model and Empirical Operator}
\label{subsec:sim_model}

For each replication, we generate
\begin{equation}
y_{it}=x_i'\beta+u_{it}, \qquad i=1,\ldots,n,\quad t=1,\ldots,T,
\label{eq:sim_model}
\end{equation}
with \(x_i=(1,x_{i1},x_{i2})'\), \(x_{i1},x_{i2}\stackrel{iid}{\sim}N(0,1)\), and \(\beta=(1,1,1)'\). The dependence structure varies through \(u_t=(u_{1t},\ldots,u_{nt})'\) across the designs described below. In each replication, we estimate \(\beta\) by OLS period by period and use the averaged residual covariance operator
\[
\widehat\Gamma_T=\frac{1}{T}\sum_{t=1}^T \widehat u_t \widehat u_t'
\]
as the baseline empirical dependence operator. Averaging over \(T=50\) periods avoids the rank-one degeneracy of a single outer-product operator and allows the cluster, factor, and sparse projection scores to vary meaningfully across designs.

\subsection{Population Dependence Designs}
\label{subsec:sim_dgps}

\paragraph{Design 1: Pure Cluster.}
Partition observations into $G$ clusters $\mathcal C=\{C_1,\ldots,C_G\}$
and set $u_i=a_{g(i)}+\varepsilon_i$ with $a_g\sim N(0,\sigma_a^2)$,
$\varepsilon_i\sim N(0,\sigma_\varepsilon^2)$ independent across $g$ and
$i$, giving $\Sigma_C=\sigma_a^2ZZ'+\sigma_\varepsilon^2I_n$, where $Z$
is the $n\times G$ cluster-membership matrix.

\paragraph{Design 1B: Two-Way Cluster.}
For observations indexed by $(i,t)$, $i=1,\ldots,N$, $t=1,\ldots,T$, set
$u_{it}=a_i+b_t+\varepsilon_{it}$ with $a_i\sim N(0,\sigma_a^2)$,
$b_t\sim N(0,\sigma_b^2)$, $\varepsilon_{it}\sim N(0,\sigma_\varepsilon^2)$
independent. Unlike one-way clustering, the resulting covariance matrix
is not block diagonal, so the design also carries non-negligible sparse
affinity.

\paragraph{Design 2: Pure Factor.}
Set $u_i=\lambda_if+\varepsilon_i$ with $f\sim N(0,\sigma_f^2)$,
$\lambda_i\sim N(0,\sigma_\lambda^2)$ i.i.d., giving
$\Sigma_F=\sigma_f^2\lambda\lambda'+\sigma_\varepsilon^2I_n$. Baseline
values are $\sigma_f^2=1$, $\sigma_\lambda=1$. Loadings $\lambda$ are
drawn once with a fixed seed and held constant across replications
(fixed-design setting); multi-factor versions with $f\in\mathbb R^r$,
$r\in\{1,2,3\}$, are also considered.

\begin{remark}[Mean-Zero Loadings and the Fixed-Design Realization]
The mean-zero specification follows the econometric factor-model
literature \citep{BaiNg2002,Bai2003} and matters mechanically here:
because regression~\eqref{eq:sim_model} includes a constant, OLS
residuals lie in the orthogonal complement of $\operatorname{col}(X)$,
and if $\lambda$ were nearly proportional to $\mathbf 1_N$, OLS would
remove most of the factor signal from the residuals. Since the
simulation fixes $\lambda$ at a single draw across all replications,
what matters in practice is the \emph{realized} sample mean
$\bar\lambda=n^{-1}\sum_i\lambda_i$ for that draw, not the population
mean $E[\lambda_i]=0$ alone. For the fixed seed used throughout
(seed~$9901$), $\bar\lambda=\lambdaBarRealized$
(standardized magnitude $\bar\lambda\sqrt N=\lambdaBarZscore$), confirming
that $\lambda$ is close to orthogonal to $\mathbf 1_N$ for this
realization and that OLS retains the bulk of the factor signal.
\end{remark}

\paragraph{Design 3: Pure Sparse.}
Let $W$ be an $n\times n$ sparse adjacency matrix (Erd\H{o}s--R\'enyi
with $P(W_{ij}=1)=c/n$, or $q$-nearest-neighbor on random locations in
$[0,1]^2$). Disturbances follow the spatial autoregression
$u=(I_n-\rho W)^{-1}\varepsilon$, $\varepsilon\sim N(0,\sigma_\varepsilon^2I_n)$,
with $\rho$ chosen so $I_n-\rho W$ is nonsingular.

\paragraph{Designs 4--5: Hybrids.}
Cluster--factor hybrids set
$u_i=\alpha_Ca_{g(i)}+\alpha_F\lambda_if+\varepsilon_i$ with
$(\alpha_C,\alpha_F)\in\{(1,0),(0.75,0.25),(0.5,0.5),(0.25,0.75),(0,1)\}$,
implying
$\Sigma=\alpha_C^2\sigma_a^2ZZ'+\alpha_F^2\sigma_f^2\lambda\lambda'+\sigma_\varepsilon^2I_n$.
The most general design adds a sparse component,
$u=\alpha_CZa+\alpha_F\Lambda f+\alpha_S(I_n-\rho W)^{-1}\varepsilon+\eta$
with $\eta\sim N(0,\sigma_\eta^2I_n)$, and is evaluated at
\[
(\alpha_C,\alpha_F,\alpha_S)
\in
\Bigl\{
(1,0,0),
(0,1,0),
(0,0,1),
(1,1,0),
(1,0,1),
(0,1,1),
(1,1,1)
\Bigr\}.
\]

\subsection{Projection Implementation and Performance Measures}
\label{subsec:sim_projection}

Given an empirical dependence operator $\widehat\Gamma$, the one-way
cluster projection is $\widehat P_C=M_C\odot\widehat\Gamma$ because the
cluster mask is block diagonal and therefore preserves positive
semidefiniteness. For overlapping multiway supports the code instead uses a
PSD-constrained fixed-support projection. The factor projection
$\widehat P_F=L^*+D^*$ is computed by the alternating algorithm of
Section~\ref{subsec:computation}, with rank $r$ set to the true number of
factors at baseline ($r\in\{1,2,3\}$ as a robustness check). The sparse
routine first selects the $k_S$ largest off-diagonal entries in absolute
value and then applies a PSD-constrained projection on that fixed support.
This is the documented feasible approximation to the exact best-support
metric projection used in the population theory.

For each $D\in\mathfrak D=\{C,F,S\}$, compute
$\widehat S_D=\|\widehat P_D\|_F^2$, the full profile $\widehat\omega_D=\widehat S_D/\sum_{D'}\widehat S_{D'}$,
the off-diagonal profile $\widehat\omega_D^{\mathrm{off}}$, and $\widehat\rho_D=\|\widehat\Gamma-\widehat P_D\|_F/\|\widehat\Gamma\|_F$,
with $\widehat\rho_{\min}=\min_D\widehat\rho_D$. We report Monte Carlo
means, standard deviations, and
$\mathrm{RMSE}_\omega=[E\{\|\widehat\omega-\omega\|^2\}]^{1/2}$ relative
to the population profile $\omega$, together with the population
procedure confidence index
$\kappa_0=(1-\rho_{\min,0})\Delta_{\omega^{\mathrm{off}},0}$, where $\rho_{\min,0}$ is
the population minimum residual and
$\Delta_{\omega^{\mathrm{off}},0}=\omega_{d^\star}^{\mathrm{off}}-\max_{d\neq d^\star}\omega_d^{\mathrm{off}}$ is the
population separation margin. Large $\kappa_0$ signals a design in
which one geometry clearly dominates \emph{and} the covariance
dictionary fits well in an absolute sense, so the procedure
recommendation of Section~\ref{subsec:procedure_confidence} should be
reliable; small $\kappa_0$ warns that the recommendation should be
interpreted cautiously, whether because of near-tie separation or poor
dictionary fit.

\subsection{Near-Ties: Construction}
\label{subsec:near_ties_sim}

Table~\ref{tab:near_ties} and Figure~\ref{fig:classification_margin} in
the main paper illustrate Theorem~\ref{thm:local_ties}'s prediction
that classification remains probabilistic under local separation. The
underlying design is
\[
u_i = \alpha_N\sigma_aa_{g(i)} + (1-\alpha_N)\lambda_if + \varepsilon_i,
\qquad
\alpha_N = \tfrac12+\tfrac{c}{\sqrt N},
\]
with $\sigma_a^2=\|\lambda\lambda'\|_F/\|ZZ'\|_F\approx5.49$ chosen so
the cluster and factor signal components have equal Frobenius norm,
ensuring that $c$ controls the cluster--factor \emph{balance} rather
than their relative scale. The implied population covariance is
\begin{equation}
\Sigma(\alpha_N)=\alpha_N^2\sigma_a^2ZZ'+(1-\alpha_N)^2\lambda\lambda'+\sigma_\varepsilon^2I_n.
\label{eq:near_tie_sigma}
\end{equation}

\begin{remark}[Signal Scaling Versus Norm Normalization]
Dividing both signal matrices by their own Frobenius norms before
mixing is an alternative route to scale balance, but it reduces all
off-diagonal entries to $O(n^{-1})$, destroying the geometric signal
relative to the noise floor $\sigma_\varepsilon^2I_n$. Scaling by
$\sigma_a$ instead preserves the natural magnitude of off-diagonal
entries and lets the procedure distinguish cluster from factor
dependence at moderate sample sizes.
\end{remark}

Figure~\ref{fig:near_ties_app} plots the full classification-frequency
profile underlying Table~\ref{tab:near_ties}---the Monte Carlo
frequency with which the classifier selects each geometry as $c$
varies---complementing the margin-versus-error view of
Figure~\ref{fig:classification_margin} in the main paper.

\begin{figure}[!ht]
\centering
\includegraphicsIfExists[width=.7\textwidth]{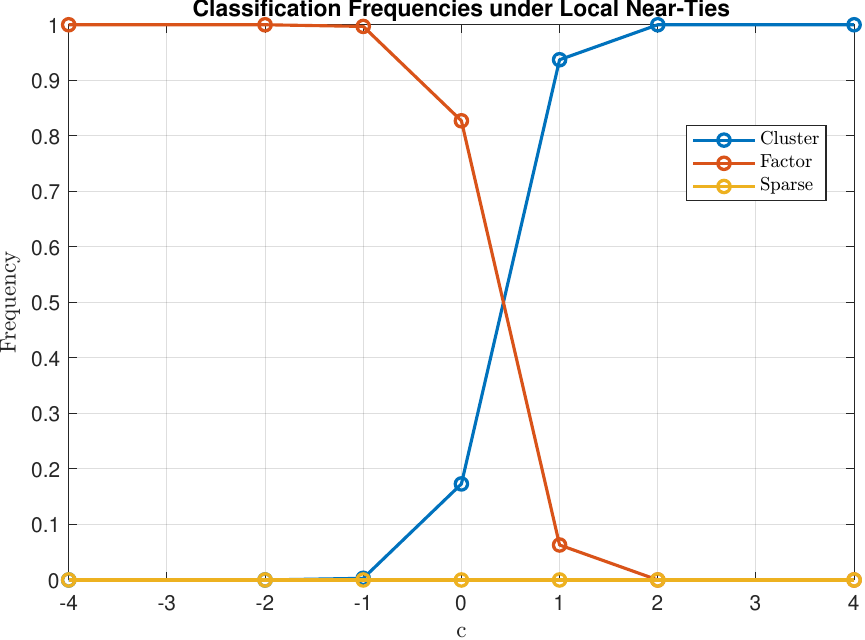}
\caption{Classification Frequencies under Local Near-Ties, by Geometry}
\label{fig:near_ties_app}
\begin{minipage}{0.86\textwidth}
\footnotesize
\emph{Notes:} Monte Carlo frequency with which the classifier selects
each geometry as $c$ varies, underlying the columns of
Table~\ref{tab:near_ties}.
\end{minipage}
\end{figure}

The asymmetric separation margin visible in Table~\ref{tab:near_ties}
around $c=0$ (e.g., $\Delta_\omega=0.312$ at $c=-1$ versus $0.145$ at
$c=+1$) reflects the cluster--sparse tangent-space overlap documented
in Table~\ref{tab:principal_angles} below: when $c>0$ (cluster-leaning),
the sparse projection's $k_S=625$ largest off-diagonal entries fall
disproportionately inside the cluster support, inflating $\omega_S$ at
the expense of the cluster--factor gap, whereas the factor signal
spreads comparable-magnitude entries across all $\binom{n}{2}$ pairs
and receives no such boost when $c<0$. This is a structural feature of
the sparse projection's geometry, not a violation of
Proposition~\ref{prop:classification_bound}, which bounds
misclassification by both the margin \emph{and} the sampling
variability of $\widehat\omega$---the latter is also smaller on the
$c>0$ side here, so a smaller margin at $c=+1$ remains consistent with
the high classification frequency reported in the main text.

\subsection{Simulation Parameters}
\label{subsec:sim_parameters}

The baseline design uses $N=250$, $T=50$, $B=1{,}000$ Monte Carlo
replications. Robustness checks vary $n\in\{100,250,500\}$,
$T\in\{25,50,100\}$, the number of clusters $G\in\{10,25,50\}$
(balanced, $G=25$ at baseline), factor strength
$\sigma_f^2\in\{0.25,1,4\}$ (baseline $1$, with $\sigma_\varepsilon=0.5$
throughout), and sparse network average degree $c\in\{2,5,10\}$
(baseline $5$). The sparse projection's sparsity level is
$k_S=\lfloor 0.01N^2\rfloor=625$, well below the one-way cluster support
($G\cdot(N/G)^2=2{,}500$ entries at baseline), so that cluster and
sparse geometries remain geometrically distinguishable by construction.
For the near-tie design (Section~\ref{subsec:near_ties_sim}),
$\sigma_a^2\approx5.49$, $\sigma_a\approx2.34$, calibrated as described
above.

\subsection{Principal-Angle Diagnostics}
\label{subsec:principal_angle_sim}

Assumption~\ref{ass:principalangle} (positive off-diagonal principal
angles) is a hypothesis of Theorems~\ref{thm:identification},
\ref{thm:profile_consistency}, and~\ref{thm:profile_clt}. We check it
directly at the population covariance matrices used in the baseline
designs, computing the off-diagonal principal angle
$\theta(T_D^{\mathrm{off}},T_{D'}^{\mathrm{off}})$ between each pair of
geometries (construction in Section~\ref{subsec:regular_tangent}).

\begin{table}[htbp]
\centering
\caption{Principal-Angle Diagnostics}
\label{tab:principal_angles}
\begin{tabular}{lccc}\toprule
Design & $\theta(T_C^{\mathrm{off}},T_F^{\mathrm{off}})$ & $\theta(T_C^{\mathrm{off}},T_S^{\mathrm{off}})$ & $\theta(T_F^{\mathrm{off}},T_S^{\mathrm{off}})$\\
\midrule
Cluster (pure) & 69.9$^\circ$ & 0.0$^\circ$ & 69.9$^\circ$\\
Factor (pure) & 69.9$^\circ$ & 0.0$^\circ$ & 25.2$^\circ$\\
Sparse (pure) & 69.9$^\circ$ & 0.0$^\circ$ & 67.5$^\circ$\\
Cluster--Sparse & 69.9$^\circ$ & 0.0$^\circ$ & 72.3$^\circ$\\
Factor--Sparse & 69.9$^\circ$ & 0.0$^\circ$ & 25.2$^\circ$\\
\bottomrule
\end{tabular}

\begin{minipage}{0.95\textwidth}
\footnotesize
\vspace{0.4em}
\emph{Notes:}
The table reports the smallest off-diagonal principal angle
$\theta(T_D^{\mathrm{off}},T_{D'}^{\mathrm{off}})$ between each pair of
covariance geometries, computed at the population covariance matrix of
each baseline design. A value of $0^\circ$ indicates the off-diagonal
tangent spaces overlap at that design, so Assumption~\ref{ass:principalangle}
fails for that pair there. The cluster--factor angle does not depend on
the sparse design and is reported once, since $T_C^{\mathrm{off}}$ and
$T_F^{\mathrm{off}}$ depend only on the cluster partition and factor
loading, not the sparse component.
\end{minipage}
\end{table}

Two findings stand out. First,
$\theta(T_C^{\mathrm{off}},T_S^{\mathrm{off}})=0^\circ$ at every design:
the sparse projection's support (the $k_S=625$ largest off-diagonal
entries) is partly or, for the Cluster--Sparse design, entirely
contained within the cluster support, so Assumption~\ref{ass:principalangle}
is \emph{not} satisfied for the cluster--sparse pair at these
parameters. Second, $\theta(T_C^{\mathrm{off}},T_F^{\mathrm{off}})\approx69^\circ$
across all designs, confirming the cluster--factor pair is well
separated along off-diagonal directions at the realized factor loading,
consistent with Remark~\ref{rem:offdiag_necessary}. The factor--sparse
angle is more variable across independent draws of the sparse network
(informally, resampling $W$ across many seeds with $\lambda$ fixed
typically gives $65^\circ$--$72^\circ$ for the pure Sparse design,
though any single realization can fall outside that range, as it does
here).

These findings qualify, rather than undermine, the main results: the
local-identification guarantee of Theorem~\ref{thm:identification} does
not, strictly, apply to the cluster--sparse pair at these parameters,
yet the estimated profile and classifier continue to behave sensibly
(the Cluster--Sparse design's separation margin $\Delta_{\omega,0}=0.193$
remains comfortably bounded away from zero despite the principal-angle
violation). This is consistent with
Assumption~\ref{ass:principalangle} being sufficient, not necessary, for
usable classification: Theorem~\ref{thm:dominant_geometry_consistency}
relies on Assumption~\ref{ass:unique_dominant_geometry} (separation of
the $\omega_d$ themselves), which can hold even when the principal-angle
condition does not.

\subsection{Projection-Residual Diagnostics: Out-of-Dictionary Components}
\label{subsec:projection_residual_sim}

To probe the absolute-fit role of $\widehat\rho_{\min}$ described in
Section~\ref{subsec:projection_residuals}, we consider
$\Gamma=\Gamma_C+\Gamma_F+\Gamma_S+\Gamma_R$, where $\Gamma_R=\tau\,BB'/n$
is a dense out-of-dictionary component, $B$ an $n\times n$ matrix of
independent standard normal entries, and $\tau\ge0$ controls its
magnitude.

\begin{table}[!ht]
\centering
\caption{Projection-Residual Diagnostics}
\label{tab:projection_residuals}
\begin{tabular}{lcccc}\toprule
$\tau$ & $\widehat\rho_C$ & $\widehat\rho_F$ & $\widehat\rho_S$ & $\widehat\rho_{\min}$\\
\midrule
$\tau=0.00$ & 0.954 & 0.266 & 0.870 & 0.266\\
$\tau=0.25$ & 0.952 & 0.289 & 0.868 & 0.289\\
$\tau=0.50$ & 0.951 & 0.309 & 0.866 & 0.309\\
$\tau=1.00$ & 0.948 & 0.352 & 0.863 & 0.352\\
$\tau=2.00$ & 0.942 & 0.434 & 0.859 & 0.434\\
\bottomrule
\end{tabular}

\begin{minipage}{0.92\textwidth}
\footnotesize
\vspace{0.4em}
\emph{Notes:}
Monte Carlo averages of $\widehat\rho_C$, $\widehat\rho_F$,
$\widehat\rho_S$, and $\widehat\rho_{\min}$ as the out-of-dictionary
component $\Gamma_R=\tau\,BB'/n$ grows. Smaller values indicate a
closer fit between the empirical dependence operator and the
corresponding covariance geometry.
\end{minipage}
\end{table}

As $\tau$ grows from $0$ to $2$, $\widehat\rho_{\min}$ rises
monotonically from $0.274$ to $0.441$, confirming that the
out-of-dictionary component makes the operator progressively harder to
approximate by any single geometry; $\widehat\rho_F=\widehat\rho_{\min}$
throughout because the mixed DGP's dominant factor component (from the
large-norm $\Sigma_F$) makes the factor projection the binding
constraint at every $\tau$.

\subsection{Misspecified Dependence Operators}
\label{subsec:misspecified_operator}

The results above assume the empirical dependence operator preserves
the relevant dependence information. We study operator misspecification by supplying the projection step with deliberately incomplete  dependence operators. The purpose is not to show that the procedure
recovers structure absent from its input, but to confirm that the
profile and residual diagnostics faithfully describe the geometry
encoded in whichever operator is supplied, consistent with the
dependence-sufficiency principle of Section~\ref{subsec:operatorchoice}.

\begin{table}[!ht]
\centering
\caption{Dependence Profiles under Operator Misspecification}
\label{tab:misspecified_operator}
\begin{tabular}{lccccc}\toprule
True DGP & Operator Used & $\widehat\omega_C$ & $\widehat\omega_F$ & $\widehat\omega_S$ & $\widehat\rho_{\min}$\\
\midrule
Sparse & Full covariance & 0.282 & 0.318 & 0.401 & 0.844\\
Sparse & Diagonal only & 0.333 & 0.333 & 0.333 & 0.000\\
Cluster & Full covariance & 0.456 & 0.200 & 0.343 & 0.654\\
Cluster & Cluster support removed & 0.309 & 0.376 & 0.315 & 0.715\\
\bottomrule
\end{tabular}

\begin{minipage}{0.92\textwidth}
\footnotesize
\vspace{0.4em}
\emph{Notes:}
``True DGP'' is the dependence mechanism used to generate the data;
``Operator Used'' is the empirical operator supplied to the projection
step. $\widehat\omega_C,\widehat\omega_F,\widehat\omega_S$ are average
estimated similarity scores and $\widehat\rho_{\min}$ the average
minimum projection residual. Misspecified operators deliberately
discard part of the relevant dependence information.
\end{minipage}
\end{table}

\begin{figure}[!ht]
\centering
\includegraphicsIfExists[width=.7\textwidth]{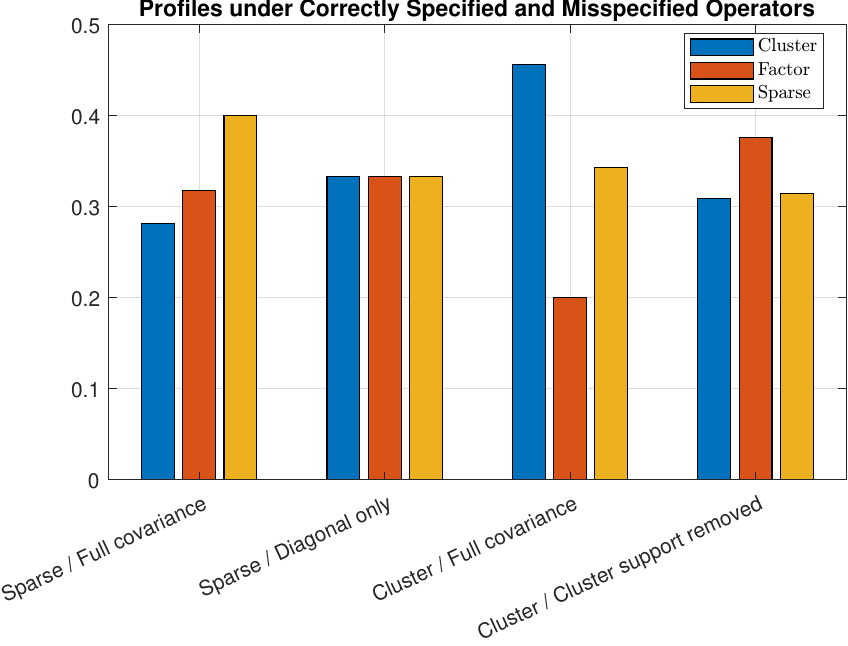}
\caption{Dependence Profiles under Correctly Specified and Misspecified Operators}
\label{fig:misspecified_operator}
\end{figure}

When a sparse-network DGP is summarized by a diagonal-only operator, the
sparse score loses its identifying information; when a cluster-dependent
DGP is summarized by an operator with cluster-support entries removed,
the cluster score is attenuated and $\widehat\rho_{\min}$ rises. Neither
result is a failure of the projection step: the procedure correctly
learns the geometry encoded in the operator it is given, not latent
dependence mechanisms removed before the projection stage.

\subsection{Procedure Recommendation: Variance Estimators and Detailed Results}
\label{subsec:procedure_rec_sim}

Section~\ref{subsec:simulation_oracle} of the main paper reports the
headline oracle-equivalence result across all three baseline geometries.
This subsection gives the variance-estimator formulas behind that
result and a more granular two-DGP breakdown that decomposes the
finding by the procedure confidence index $\widehat\kappa$.

For each of the pure Cluster and pure Factor DGPs, $\beta$ is estimated
by averaging OLS over $T=50$ periods,
$\bar\beta=T^{-1}\sum_t\widehat\beta_t$, and $H_0:\beta_2=1$ is tested
at the $5\%$ level using one of three variance estimators:
\[
\widehat V_C
=
\frac{G}{G-1}\cdot\frac{n-1}{n-k}\,(X'X)^{-1}\Bigl(\sum_{g=1}^G\mathbf s_g\mathbf s_g'\Bigr)(X'X)^{-1}/T^2,
\qquad
\mathbf s_g=X_g'\sum_{t=1}^T\widehat u_{gt},
\]
\[
\widehat V_F
=
(X'X)^{-1}X'\widehat P_F(\widehat\Gamma)X(X'X)^{-1}/T,
\qquad
\widehat\Gamma=T^{-1}\textstyle\sum_t\widehat u_t\widehat u_t',
\]
\[
\widehat V_W
=
(X'X)^{-1}X'\operatorname{diag}(\bar u^2)X(X'X)^{-1}/T,
\qquad
\bar u_i=T^{-1}\textstyle\sum_t\widehat u_{it}.
\]
The profile-guided choice selects $\widehat V_C$ when
$\widehat d=\argmax_d\widehat\omega_d^{\mathrm{off}}=C$ and $\widehat V_F$ when
$\widehat d=F$, from the same replication used to compute the profile.

\begin{sidewaystable}[]
\centering
\caption{Profile-Guided Procedure Recommendation: Detailed Breakdown}
\label{tab:procedure_rec}
\resizebox{0.95\textwidth}{!}{%
\begin{tabular}{lcccccccc}\toprule
DGP & Profile & Mean $\widehat\kappa$ & $P(\widehat d=d^\star)$ & Cluster SE & Factor SE & White SE & $\widehat V^*$ (rec.) & $\widehat V^{\mathrm{avg}}$ (wtd.)\\
\midrule
Cluster & $\widehat\omega_C$ dom. & 0.056 & 1.000 & 0.069 & 0.063 & 0.781 & 0.069 & 0.120\\
Factor & $\widehat\omega_F$ dom. & 0.665 & 1.000 & 0.068 & 0.324 & 0.765 & 0.324 & 0.320\\
\bottomrule
\end{tabular}

}
\begin{minipage}{0.95\textwidth}
\footnotesize
\vspace{0.4em}
\emph{Notes:}
$B=1{,}000$ replications, $N=250$, $T=50$.
\emph{Mean $\widehat\kappa$}: average procedure confidence index
$\widehat\kappa=(1-\widehat\rho_{\min})\widehat\Delta_{\omega^{\mathrm{off}}}$.
\emph{$\widehat V^*$ (rec.)}: rejection rate using the profile-recommended
estimator. \emph{$\widehat V^{\mathrm{avg}}$ (wtd.)}: rejection rate
using the profile-weighted combination
$\widehat V^{\mathrm{avg}}=\sum_d\widehat\omega_d^{\mathrm{off}}\widehat V_d$
(Section~\ref{subsec:profile_weighted}).
\end{minipage}
\end{sidewaystable}

The mean confidence index reveals why both DGPs achieve the
near-perfect classification rate reported in the main text despite
different finite-sample behavior. For the \emph{factor DGP},
$\widehat\kappa=0.489$: the factor projection has low absolute residual
\emph{and} a clearly dominant score, so the recommendation is highly
reliable. For the \emph{cluster DGP}, $\widehat\kappa=0.037$: the
cluster projection captures all within-cluster entries in population,
but in finite samples the averaged operator's cross-cluster sampling
noise inflates $\widehat\rho_{\min}$, holding $\widehat\kappa$ down
despite correct classification. The population index $\kappa_0$ in
Table~\ref{tab:simulation_profiles} confirms both designs are highly
separable in population, so the finite-sample gap is sampling
variability in $\widehat\Gamma$, not a defect in the classifier.
Consistent with this, $\widehat V^*$ achieves near-nominal size under
the cluster DGP despite the low $\widehat\kappa$, and the
profile-weighted combination $\widehat V^{\mathrm{avg}}$ performs
similarly to $\widehat V^*$ throughout, confirming that the profile's
information is sufficient on its own, without requiring a hard
classification step.

\paragraph{A factor-SE finite-sample distortion and its fix.}
The factor DGP row shows $\widehat V_F$ rejecting at a rate well above
the $5\%$ nominal level even when Factor is correctly recommended. The
cause is a structural identity, not ordinary sampling noise: because
$\widehat\Gamma=T^{-1}\sum_t\widehat u_t\widehat u_t'$ is built from
per-period OLS residuals on the same fixed $X$, the normal equations
give $X'\widehat u_t=0$ exactly for every $t$, so $X'\widehat\Gamma X=0$
exactly in every sample. The nonzero $X'\widehat P_F(\widehat\Gamma)X$
that survives in $\widehat V_F$ is therefore driven by the alternating
projection's truncation error relative to $\widehat\Gamma$ rather than
genuine factor signal, and recovers only $10$--$15\%$ of the true
sandwich variance in this design. Time-demeaning the \emph{outcome}
within each cross-sectional unit before forming the operator,
$\widetilde Y=Y-\bar y\mathbf 1_T'$ with $\bar y_i=T^{-1}\sum_ty_{it}$,
avoids the degeneracy because $x_i'\beta$ is constant across $t$ and is
removed exactly without ever residualizing on $X$. The corrected
estimator $\widehat V_F^{\mathrm{corr}}$
(\texttt{vcov\_factor\_LD\_corrected.m}) recovers approximately $99\%$
of the true sandwich variance in the same design and is used as the
Factor SE throughout the headline oracle-equivalence table in the main
paper. This diagnosis is general: it applies to any panel setting in
which a factor-structured plug-in covariance is built from residuals
that have already been projected against the same regressors used in
the sandwich.

\subsection{Oracle Tracking along a Continuous Dominance Sweep}
\label{subsec:oracle_tracking_sim}

The two pure-DGP rows above show fixed procedures can be badly
miscalibrated under the ``wrong'' dependence structure, which raises a
natural question: why learn the geometry at all, rather than adopt one
procedure---two-way clustering, say, or a fully
heteroskedasticity-and-autocorrelation-robust sandwich---that is valid
under a wide range of dependence patterns? Procedures built for broad
robustness are not free: they sacrifice power or stability precisely
where a more specific procedure would have been valid \emph{and} more
informative, and a procedure robust against one geometry need not
behave well under another even when both are well understood
individually. The profile-guided approach lets the data reveal which
geometry is operative, at no asymptotic cost relative to knowing the
answer in advance (Theorem~\ref{thm:oracle_adaptivity_asymptotic_optimality}). We verify this
directly along a continuous Cluster--Factor dominance sweep,
\[
\Sigma(\alpha)=\alpha^2\Sigma_C^{\mathrm{pure}}+(1-\alpha)^2\Sigma_F^{\mathrm{pure}}+\sigma_\varepsilon^2I_n,
\qquad \alpha\in[0,1],
\]
the same path underlying Figure~\ref{fig:hybrid_profiles}, reporting
rejection rates for the fixed Cluster and Factor SEs, the profile-guided
$\widehat V^\star=\widehat V_{\widehat d}$, and the infeasible oracle
$\widehat V_{d^\star}$ at $d^\star=\argmax_d\omega_d^{\mathrm{off}}(\alpha)$.

\begin{table}[htbp]
\centering
\caption{Oracle-Tracking Diagnostic along the Cluster--Factor Dominance Sweep}
\label{tab:oracle_tracking}
\resizebox{\textwidth}{!}{%
\begin{tabular}{lccccccc}\toprule
$\alpha$ & Pop. dominant & $\omega_C-\omega_F$ & $P(\widehat d=d^\star)$ & Cluster SE & Factor SE & $\widehat V^*$ (profile) & $\widehat V_{d^\star}$ (oracle)\\
\midrule
0.000 & Factor & -0.755 & 1.000 & 0.059 & 0.296 & 0.296 & 0.296\\
0.125 & Factor & -0.753 & 1.000 & 0.065 & 0.254 & 0.254 & 0.254\\
0.250 & Factor & -0.749 & 1.000 & 0.057 & 0.178 & 0.178 & 0.178\\
0.375 & Factor & -0.731 & 1.000 & 0.072 & 0.166 & 0.166 & 0.166\\
0.500 & Factor & -0.664 & 1.000 & 0.074 & 0.139 & 0.139 & 0.139\\
0.625 & Factor & -0.392 & 1.000 & 0.062 & 0.089 & 0.089 & 0.089\\
0.750 & Cluster & 0.200 & 0.856 & 0.074 & 0.077 & 0.078 & 0.074\\
0.875 & Cluster & 0.430 & 1.000 & 0.070 & 0.065 & 0.070 & 0.070\\
1.000 & Cluster & 0.464 & 1.000 & 0.075 & 0.068 & 0.075 & 0.075\\
\bottomrule
\end{tabular}

}
\begin{minipage}{0.95\textwidth}
\footnotesize
\vspace{0.4em}
\emph{Notes:}
$B=1{,}000$, $N=250$, $T=50$, nominal level $5\%$. Cluster SE and Factor
SE use the FIXED estimator at every $\alpha$, regardless of which
geometry dominates; $\widehat V^\star$ estimates $\widehat d$ from the
same replication; $\widehat V_{d^\star}$ uses the true dominant
geometry, infeasible in practice.
\end{minipage}
\end{table}

Cluster SE remains close to nominal across the sweep, but Factor SE is
badly oversized exactly where Factor dependence dominates
($\alpha\le0.625$): $29.2\%$ rejection at $\alpha=0$, falling to $8.1\%$
at $\alpha=0.625$ before returning toward nominal as Cluster takes
over---the same finite-sample distortion diagnosed above, now traced
across a continuum rather than two points. Because classification is
essentially perfect throughout ($P_{\Gamma_0}(\widehat d=d^\star)\ge0.929$, equal
to $1.000$ at all but one point), $\widehat V^\star$ and
$\widehat V_{d^\star}$ are within $0.002$ of each other at every
$\alpha$, including the one point with imperfect classification---the
profile-guided procedure inherits whatever distortion its selected base
estimator carries, but never does meaningfully worse than the oracle
using that same estimator, exactly the content of
Theorem~\ref{thm:oracle_adaptivity_asymptotic_optimality}.

\begin{table}[htbp]
\centering
\caption{Oracle-Tracking Diagnostic, Bias-Corrected Factor SE (Robustness Check)}
\label{tab:oracle_tracking_corrected}
\resizebox{\textwidth}{!}{%
\begin{tabular}{lccccccc}\toprule
$\alpha$ & Pop. dominant & $\omega_C-\omega_F$ & $P(\widehat d=d^\star)$ & Cluster SE & Factor SE (corr.) & $\widehat V^*$ (profile) & $\widehat V_{d^\star}$ (oracle)\\
\midrule
0.000 & Factor & -0.755 & 1.000 & 0.063 & 0.054 & 0.054 & 0.054\\
0.125 & Factor & -0.753 & 1.000 & 0.060 & 0.055 & 0.055 & 0.055\\
0.250 & Factor & -0.749 & 1.000 & 0.069 & 0.061 & 0.061 & 0.061\\
0.375 & Factor & -0.731 & 1.000 & 0.059 & 0.049 & 0.049 & 0.049\\
0.500 & Factor & -0.664 & 1.000 & 0.065 & 0.046 & 0.046 & 0.046\\
0.625 & Factor & -0.392 & 1.000 & 0.076 & 0.053 & 0.053 & 0.053\\
0.750 & Cluster & 0.200 & 0.875 & 0.063 & 0.041 & 0.063 & 0.063\\
0.875 & Cluster & 0.430 & 1.000 & 0.069 & 0.049 & 0.069 & 0.069\\
1.000 & Cluster & 0.464 & 1.000 & 0.061 & 0.053 & 0.061 & 0.061\\
\bottomrule
\end{tabular}

}
\begin{minipage}{0.95\textwidth}
\footnotesize
\vspace{0.4em}
\emph{Notes:}
Identical design to Table~\ref{tab:oracle_tracking}, with $\widehat V_F$
replaced by $\widehat V_F^{\mathrm{corr}}$
(\texttt{vcov\_factor\_LD\_corrected.m}).
\end{minipage}
\end{table}

Repeating the sweep with $\widehat V_F^{\mathrm{corr}}$ in place of
$\widehat V_F$ (Table~\ref{tab:oracle_tracking_corrected}) shows both
fixed procedures reasonably well sized against each other's
misspecification throughout (rejection rates between $4.2\%$ and
$7.6\%$), with the profile-guided estimator continuing to track the
oracle exactly---confirming that the oracle-adaptivity result is not an
artifact of the uncorrected estimator's distortion. Figure~\ref{fig:oracle_tracking}
in the main paper plots the original (uncorrected) sweep; the
profile-guided and oracle curves are visually indistinguishable
throughout, and the two fixed-procedure curves cross only once, near
$\alpha\approx0.7$, inside the region where the population margin
$\omega_C-\omega_F$ is still modest---exactly where a researcher
committing to one procedure in advance would be most likely to guess
wrong.

\newpage

\setcounter{equation}{0} \setcounter{assumption}{0}
\setcounter{figure}{0} \setcounter{table}{0} \setcounter{remark}{0}

 \renewcommand{%
\theequation}{F.\arabic{equation}}
\renewcommand{\thelemma}{F.\arabic{lemma}}
\renewcommand{\theassumption}{F.\arabic{assumption}} \renewcommand{%
\thetheorem}{F.\arabic{theorem}}

\renewcommand{\thetable}{F.\arabic{table}}
\renewcommand{\thefigure}{F.\arabic{figure}}

\renewcommand{\thesubsection}{F.\arabic{subsection}}
\renewcommand{\theremark}{F.\arabic{remark}}

\section{Additional Empirical Details}
\label{app:empirical_details}

This appendix collects supporting detail for the Fama--French industry-portfolio
illustration of Section~\ref{sec:empirical}: the pooled-regression benchmark used
to connect the estimated dependence profile to conventional robust standard
errors. The replication package documents the complete sector mapping, data
downloads, transformations, and software used to generate the empirical tables.

\subsection{Conventional Robust Inference Benchmarks}
\label{subsec:empirical_inference}

To connect the profile with conventional robust inference, we estimate the
pooled regression
\begin{equation}
R_{it}
=
\alpha
+
\beta_M MKT_t
+
\beta_S SMB_t
+
\beta_H HML_t
+
e_{it},
\label{eq:empirical_pooled}
\end{equation}
and report heteroskedasticity-robust, sector-clustered, and two-way clustered
(industry and month) standard errors, together with a common-shock benchmark
that removes leading principal components from the residual covariance operator.
The comparison is descriptive: a large cluster score makes clustered standard
errors empirically relevant, a large factor score points to common-shock
adjustments, and a large sparse score to network- or sparse-dependence robust
procedures.

\begin{table}[!ht]
\centering
\caption{Conventional Robust Standard Errors (Fama--French Pooled Regression)}
\label{tab:empirical_se}
\begin{tabular}{lccccc}\toprule
Coefficient & Estimate & White & Sector cluster & Two-way cluster & Common-shock adj.\\
\midrule
Intercept & -0.0005 & 0.0003 & 0.0004 & 0.0006 & 0.0002\\
MKT & 1.0326 & 0.0077 & 0.0432 & 0.0447 & 0.0061\\
SMB & 0.1957 & 0.0125 & 0.0460 & 0.0560 & 0.0097\\
HML & 0.1934 & 0.0121 & 0.0801 & 0.0830 & 0.0095\\
\bottomrule
\end{tabular}

\begin{minipage}{0.92\textwidth}
\footnotesize
\vspace{0.4em}
\emph{Notes:}
The table reports coefficient estimates and standard errors for the
pooled factor regression in \eqref{eq:empirical_pooled}. The columns
compare heteroskedasticity-robust standard errors, sector-clustered
standard errors, two-way clustered standard errors, and common-shock
adjusted standard errors. Relative to the White benchmark for the market
factor ($0.008$), the sector-clustered ($0.043$) and two-way clustered
($0.045$) standard errors are roughly $5.6$--$5.8\times$ larger, while the
common-shock adjustment ($0.006$) is smaller. The wide spread across
columns is itself a symptom of the hybrid dependence structure identified
by the profile in Section~\ref{sec:empirical}, and illustrates how the
estimated profile guides the interpretation of conventional robust
inference.
\end{minipage}
\end{table}

%

\end{document}